\documentclass{article}
\usepackage[top=1in, bottom=1in, left=1in, right=1in]{geometry}

\usepackage{amsthm,amsmath,amssymb,amsfonts}
\usepackage{url}
\usepackage{tikz}
\usepackage{mathrsfs}
\usepackage{nicefrac,bbm}
\usepackage{hyperref, cleveref}
\usepackage[boxruled]{algorithm2e}
\usepackage{enumitem,mathtools,pifont}
\usepackage[framemethod=TikZ]{mdframed}
\usepackage{color,fancybox,graphicx,subfigure}

\usepackage{savesym}
\savesymbol{st}
\usepackage{soul}
\restoresymbol{SOUL}{st}

\SetAlgorithmName{Algorithm}{algorithm}{List of Algorithms}
\SetKwInput{kwInit}{Input}
\SetKwFor{For}{For}{do}{end}
\SetKwFor{ForEach}{For each}{do}{end}
\newcommand{\commentalg}[1]{{\small\em\quad \textcolor{gray}{//{#1}}}\hspace{-3mm}}

\newtheorem{theorem}{Theorem}[section]
\newtheorem{lemma}[theorem]{Lemma}
\newtheorem{definition}[theorem]{Definition}

\newtheorem{fact}[theorem]{Fact}

\newtheorem*{theorem*}{Theorem}
\newtheorem{remark}[theorem]{Remark}

\crefname{section}{Section}{Sections}
\crefname{theorem}{Theorem}{Theorems}
\crefname{observation}{Observation}{Observations}
\crefname{proposition}{Proposition}{Propositions}
\crefname{claim}{Claim}{Claims}
\crefname{condition}{Condition}{Conditions}
\crefname{example}{Example}{Examples}
\crefname{fact}{Fact}{Facts}
\crefname{lemma}{Lemma}{Lemmas}
\crefname{corollary}{Corollary}{Corollaries}
\crefname{definition}{Definition}{Definitions}
\crefname{remark}{Remark}{Remarks}
\crefname{todo}{Todo}{Remarks}

\newcommand{\yesnum}{\addtocounter{equation}{1}\tag{\theequation}}
\newcommand{\tagnum}[1]{\addtocounter{equation}{1}{\tag{#1; \theequation}}}
\makeatletter
\newcommand{\customlabel}[2]{%
\protected@write \@auxout {}{\string \newlabel {#1}{{#2}{\thepage}{#2}{#1}{}} }%
\hypertarget{#1}{}
}
\makeatother

\mdfdefinestyle{FrameBox}{ linecolor=white, outerlinewidth=0pt, roundcorner=5pt,
innertopmargin=5pt, innerbottommargin=5pt,
innerrightmargin=10pt, innerleftmargin=10pt, backgroundcolor=white}

\mdfdefinestyle{FrameBox2}{ linecolor=black, outerlinewidth=0pt, roundcorner=5pt,
innertopmargin=5pt, innerbottommargin=5pt,
innerrightmargin=10pt, innerleftmargin=10pt, backgroundcolor=white}

\mdfdefinestyle{FrameQuestion2}{ linecolor=white, outerlinewidth=0pt,
roundcorner=5pt,
innertopmargin=5pt, innerbottommargin=0pt,
innerrightmargin=16pt, innerleftmargin=16pt, backgroundcolor=white}

\newcommand{\white}[1]{\textcolor{white}{#1}}

\newcommand{\N}{\mathbb{N}}

\newcommand{\R}{\mathbb{R}}
\newcommand{\Z}{\mathbb{Z}}

\newcommand{\cA}{\mathcal{A}}

\newcommand{\cD}{\mathcal{D}}
\newcommand{\cE}{\ensuremath{\mathcal{E}}}
\newcommand{\cF}{\mathcal{F}}
\newcommand{\cK}{\mathcal{K}}

\newcommand{\cM}{\mathcal{M}}

\newcommand{\cU}{\mathcal{U}}

\newcommand{\cX}{\mathcal{X}}

\newcommand{\nfrac}{\nicefrac}

\newcommand{\st}{\mathrm{s.t.}}

\newcommand{\eps}{\varepsilon}
\renewcommand{\epsilon}{\varepsilon}

\newcommand{\argmax}{\operatornamewithlimits{argmax}}
\newcommand{\Ex}{\operatornamewithlimits{\mathbb{E}}}

\newcommand{\poly}{\mathop{\mbox{\rm poly}}}

\def\abs#1{\left| #1 \right|}

\newcommand{\sinbrace}[1]{\{#1\}}
\newcommand{\sinsquare}[1]{[#1]}

\newcommand{\inparen}[1]{\left(#1\right)}
\newcommand{\inbrace}[1]{\left\{#1\right\}}
\newcommand{\insquare}[1]{\left[#1\right]}

\newcommand{\ceil}[1]{\left\lceil#1\right\rceil}

\colorlet{RED}{red}
\newcommand{\zo}{\{0,1\}}
\newcommand{\np}{{\bf NP}}
\newcommand{\apx}{{\bf APX}}
\newcommand{\suppMat}{Section}
\newcommand{\negsp}{\hspace{0mm}}

\newcommand{\sexp}[1]{{\hbox{\tiny$($}}#1{\hbox{\tiny$)$}}}

\newcommand{\target}{Program~\ref{prob:target_fair}}

\newcommand{\denoised}{Program~\ref{prob:denoised_fair}}

\newcommand{\relaxation}{\ref{prob:relaxed_denoised_fair}}
\newcommand{\dkp}{$d$-KP}

\newcommand{\wh}[1]{\widehat{#1}}

\newcommand{\hz}{\widehat{z}}

\newcommand{\xtr}{\ensuremath{x^\star}}
\newcommand{\xdn}{\ensuremath{\widehat{x}^\star}}

\newcommand{\FairExpec}{\ensuremath{\mathsf{FairExpec}}}
\newcommand{\FairExpecGrp}{\ensuremath{\mathsf{FairExpecGrp}}}
\newcommand{\Blind}{\ensuremath{\mathsf{Blind}}}
\newcommand{\Thresh}{\ensuremath{\mathsf{Thrsh}}}
\newcommand{\MultObj}{\ensuremath{\mathsf{MultObj}}}
\newcommand{\CntrFair}{\ensuremath{\mathsf{CntrFair}}}
\newcommand{\CntrFairRes}{\ensuremath{\mathsf{CntrFairResolving}}}

\newcommand{\folder}{./figures/}
\newcommand{\stm}{\textsc{Sty}_m}
\newcommand{\stf}{\textsc{Sty}_f}

\begin{document}

\title{Mitigating Bias in Set Selection with Noisy Protected Attributes}

\author{Anay Mehrotra \\ Yale University \and L. Elisa Celis \\ Yale University}

\maketitle

\begin{abstract}
	Subset selection algorithms are ubiquitous in AI-driven applications, including, online recruiting portals and image search engines, so it is imperative that these tools are not discriminatory on the basis of protected attributes such as gender or race.
	Currently, fair subset selection algorithms assume that the protected attributes are known as part of the dataset.
	However, protected attributes may be noisy due to errors during data collection or if they are imputed (as is often the case in real-world settings).
	While a wide body of work addresses the effect of noise on the performance of machine learning algorithms, its effect on fairness remains largely unexamined.
	We find that in the presence of noisy protected attributes, in attempting to increase fairness without considering noise, one can, in fact, \emph{decrease} the fairness of the result!

	Towards addressing this, we consider an existing noise model in which there is probabilistic information about the protected attributes (e.g.,~\cite{ManwaniS13, FrenayV14,ChenKMSU19,KallusMZ20}), and ask \emph{is fair selection possible under noisy conditions?}
	We formulate a ``denoised'' selection problem which functions for a large class of fairness metrics; given the desired fairness goal, the solution to the denoised problem violates the goal by at most a small multiplicative amount with high probability.
	Although this denoised problem turns out to be $\np$-hard, we give a linear-programming based approximation algorithm for it.
	We evaluate this approach on both synthetic and real-world datasets.
	Our empirical results show that this approach can produce subsets which significantly improve the fairness metrics despite the presence of noisy protected attributes, and, compared to prior noise-oblivious approaches, has better Pareto-tradeoffs between utility and fairness.
\end{abstract}

\newpage
\setcounter{tocdepth}{2}
\tableofcontents
\newpage

\section{Introduction}\label{sec:intro}
	The subset selection problem arises in various contexts including
	online job portals (where an algorithm shortlists candidates to show to the recruiter),
	university admissions (where a panel admits a subset of students), and
	online search (where the platform selects a subset of the results in response to a user query)
	~\cite{DrosouJPS17, kearns2017meritocratic, KleinbergR18, StoyanovichYJ18}.
	The basic problem is as follows:
	There are $m$ {\em items}, and each item $i\in [m]$ has a utility $w_i\geq 0$, i.e., the value it adds to the subset.
	The goal is to select a subset of $n \ll m$ items which has the largest total utility.
	Given the pervasiveness of subset selection tasks, it is crucial to ensure that subset selection algorithms  do not propagate social biases.
	Consequently, there has been extensive work on developing fair algorithms for selection (and for the related problem of ranking); see \cite{DrosouJPS17, Castillo18} for an overview.
	Many of these approaches ensure that the number of individuals selected from different socially salient groups (e.g., those defined by gender or race) satisfy some fairness constraints and/or improve along a given fairness metric.
	Towards this, these algorithms assume (exact) access to the corresponding {\em protected attributes} of individuals.

	However, in practice, these attributes can be erroneous, unavailable for some individuals, or missing entirely~\cite{council2004eliminating,Kossinets06, saundersAccuracyOfRecordedEthnicInfo}.
	For instance, in healthcare, patients' ethnic information can be incorrectly recorded~\cite{saundersAccuracyOfRecordedEthnicInfo} or left blank~\cite{council2004eliminating}.\footnote{Recently, this received public attention when attempting to estimate the racial disparities in COVID19 infections showed large discrepancies~\cite{barboza2020coronavirus}.}
	When this data is missing, probabilistic methods based on other proxy information are used to ``impute'' these protected attributes~\cite{coldman1988classification,elliott2006usingGeocoding,elliott2008new, elliott2009UsingCencusSurnameList}.
	For instance, when assessing if lenders comply with fair lending policies, the Consumer Financial Protection Bureau uses last name and geolocation to impute consumers' race \cite{consumer2014UsingPublic}.
	Similar approaches have also been used in the context of healthcare~\cite{fiscella2006use,koh2011reducing}.
	Additionally, online job platforms (such as, LinkedIn) use a user's data to infer their demographic information based on the data they have on other users~\cite{linkedin_inferred_gender}.
	Furthermore, in some cases, such as with images on the internet, protected attributes are missing for the entire datasets (and labeling all images is not viable).
	Inferring protected attributes is bound to have errors, which can affect the groups differently~\cite{BuolamwiniG18}.
	Thus, using imputed attributes as a black-box in subsequent fair algorithms, without accounting for their noise, can have an unexpected (and adverse) impact on the fairness achieved.
	For instance, \cite{LamyZ19, celis2020fairclassification} observe that (noise oblivious) fair algorithms do not satisfy their fairness guarantee in the presence of noise.

	To gain some intuition, consider the setting where we are given a set of candidates and would like to ensure proportional representation across individuals with different skin-tones, coded as White and non-White.
	Assume that the utilities of all candidates have a similar distribution, and so picking candidates with top $n$ utilities proportionately represents them.
	Further, assume that the labels have a higher amount of noise for non-Whites than Whites.\footnote{{For instance, as observed in commercial image-based gender classifiers~\cite{BuolamwiniG18}.}}
  One can show that, any ``fair algorithm'' which assumes that these noisy labels are correct, and selects a proportionate number of White and non-White candidates based on them, would violate proportional representation.
	In this case, adding fairness constraints increased the disparity.
	This leads us to the question addressed in this paper:

	\begin{mdframed}[style=FrameQuestion2]
		\textit{Can we develop a framework for selection which  outputs an approximately fair subset despite noisy protected attributes?}
	\end{mdframed}

	\subsection{Our contributions}
	Building on prior work on {fairness constraints}~\cite{CKSDKV18,StoyanovichYJ18},
	we develop a framework for fair selection in the presence of noisy protected attributes.
	This framework allows for multiple and intersectional groups, and,
	given access to (unbiased\footnote{Here, unbiased refers to the statistical notion of an unbiased estimator.}) probabilistic information about the true protected attributes, it can satisfy a large class of fairness constraints (including, demographic parity, proportional representation, and the $80$\% rule) with high probability.

	Formally, we would like to solve an ideal optimization problem (\target{}) which satisfies the fairness constraints for the true (and unknown) protected attributes.
	Such problems have been studied by prior works, e.g., \cite{ZehlikeB0HMB17,SinghJ18, StoyanovichYJ18}.
	However, since we do not have the true protected attributes, we cannot solve it directly using their approaches.
	Instead, we formulate a ``denoised'' problem (\denoised{}); such that, an optimal solution of \denoised{} has {the} optimal utility for \target{} and violates the fairness constraints of \target{} by at most a small multiplicative factor with high probability (\cref{thm:relation_between_target_and_noisy}).
	Although \denoised{} turns out to be \np-hard (\cref{thm:hardness_results_main}),
	we develop a linear-programming based approximation algorithm for it.
	This, in turn, implies an approximation algorithm for \target{}.
	We empirically study the fairness achieved by this approach with respect to standard fairness metrics (e.g., risk difference) on both synthetic and real-world datasets.
	We also study the performance of existing fair algorithms in the presence of noise and benchmark our approach with them.
	We observe that our approach achieves the highest fairness and has a Pareto-optimal tradeoff between utility and fairness (on changing the strength of constraints).
	Interestingly, these observations also hold in our empirical results where, unlike what our theoretical results assume, we have skewed probabilistic information of the noisy attributes.
	Finally, our empirical results hint at potential applications of this approach, e.g., in online recruiting portals and image search engines.
	\subsection{Related work}\label{sec:related_work}
	\paragraph{Mitigating bias.}
	An extensive body of work strives to mitigate bias and improve diversity in subset selection and the closely related ranking problem.
	We refer the reader to \cite{DrosouJPS17} for a comprehensive overview of work on diverse selection, and an excellent talk \cite{Castillo18} which discusses work on curtailing bias in rankings.
	Closest to our setting, are approaches which use protected attributes to impose fairness constraints on algorithms for  selection~\cite{kearns2017meritocratic,StoyanovichYJ18} and ranking~\cite{ZehlikeB0HMB17,celis2018ranking,SinghJ18,linkedin_ranking_paper, yang2019balanced}.
	However, if the attributes are noisy, these could even increase the bias.
	A different approach is to learn ``unbiased utilities'' by either using a {\em causal model} to capture the relation between attributes and utilities~\cite{kusner2017counterfactual, yang2020causal} or by casting it as a multi-objective unconstrained optimization problem~\cite{YangS17, Zehlike020}.
	The former approach explicitly uses the protected attributes to generate {\em counterfactuals}, so, it can lead to unfair outcomes in the presence of noise (also see \cref{sec:causal_simulation}).
	And the latter approach can lead to sub-optimal fairness if noisy data is not accounted for, {as shown by works on fair classification~\cite{LamyZ19,awasthi2020equalized, celis2020fairclassification}.}

	In \cite{Gupta2018proxy}, it is empirically shown that when protected attributes are missing, proxy attributes can be used to improve fairness in classification.
	However, they do not consider how necessary noise resulting from the proxy attributes affects the fairness or accuracy.

	\paragraph{Mitigating bias with noise.}
	Works on curtailing bias with noisy information are relatively recent.
	Closest to this paper are those which consider noise in the protected attributes.
	In \cite{awasthi2020equalized}, conditions on the noise under which the popular post-processing method for fair classification by \cite{hardt2016equality} reduces bias in terms of equalized odds are characterized.
	However, they only consider noise in the training samples and assume that the test samples are not noisy, which often doesn't hold in practice.
	In \cite{LamyZ19}, an in-processing approach to fair classification is suggested; they show that applying tighter fairness constraints in existing fair classification frameworks can mitigate bias in terms of equalized odds and statistical parity with binary protected attributes.
	However, this approach does not extend to nonbinary protected attributes and to other definitions of fairness.
	In \cite{celis2020fairclassification}, an in-processing approach for fair classification which can mitigate bias with nonbinary and noisy protected attributes is developed.
	However, they assume that the noise only depends on the (unknown) underlying protected attributes, whereas, we also allow the noise to vary with nonprotected attributes and utility.
	Furthermore, \cite{LamyZ19,awasthi2020equalized, celis2020fairclassification} mitigate bias in classification tasks, and it is not clear how to extend these methods to subset selection.

	In \cite{ChenKMSU19, KallusMZ20}, methods to reliably assess disparity in the setting where the protected attributes are entirely missing are proposed.
	We consider a similar noise model as the one they propose; however, the problem is fundamentally different as their goal is assessment rather than mitigation.

	\paragraph{Noise models in literature.}
	Several works in the machine learning literature consider noise in the predicted labels as opposed to in attributes, protected or otherwise~\cite{AngluinL87,ManwaniS13,FrenayV14, LiuT16}.
	In this paper, we consider a noise model that arises from this line of work, but applied to the protected attributes rather than the label.

\section{Model}
	For a natural number $n\in \N$ by $[n]$ we denote the set $\{1,2,\dots,n\}$, and for a real number $x\in \R$ by $\exp(x)$ we denote $e^x$.
	We use $\mathbb{I}[\cdot]$ to denote the indicator function, $o(1)$ to denote $O(\nfrac{1}{n})$, and $\cU(a,b)$ to denote the uniform distribution on interval $[a,b]$.
	Given a natural number $p\in \N$, $\Delta^p$ denotes the standard $p$-simplex.

	\subsection{Selection problem and noise model}
	\paragraph{Selection problem.}
	In the classical selection problem, one is given $m$ {\em items}, where each item $i\in [m]$ has a utility $w_i\geq 0$.
	An item's utility is the {\em value} it adds to the selection.
	The goal is to find a subset of $n$ items which has the most total value.
	It is convenient to encode a subset with a binary selection vector $x\in \zo^m$.
	Then, the classical selection problem is
	\begin{align}
		\max\nolimits_{x\in \zo^m} \sum\nolimits_{i=1}^{m} w_i x_i\quad \st, \ \sum\nolimits_{i=1}^m x_i = n.\label{eq:classical_selection}
	\end{align}
	\paragraph{Protected attributes.} We consider $s\in \N$ protected attributes (such as, gender or race), where for $k\in [s]$, the $k$-th protected attribute can take $p_k\in \N$ values (such as, different genders or races).
	Let $\cX$ be the domain of all other nonprotected attributes.
	Fix a joint distribution over $\cD\coloneqq \R_{\geq 0}\times [p_1]\times \dots\times [p_s]\times \cX.$
	Then, each item $i\in [m]$ is represented by the tuple
	$$(w_i, z_i^{\sexp{1}}, \dots, z_i^{\sexp{s}}, a_i)\in \R_{\geq 0}\times [p_1]\times\dots\times [p_s]\times \cX,$$
	and is drawn independently from this joint distribution.
	We observe the utility $w_i$ and nonprotected attributes $a_i$, but do {\em not} observe the protected attributes $(z_i^{\sexp{1}}, \dots, z_i^{\sexp{s}})$.
	Instead, we observe a {\em noisy version} $(\hz_i^{\sexp{1}}, \dots, \hz_i^{\sexp{s}})$ of them (for each $i\in[m]$).

	For each attribute-value pair $k\negsp\in\negsp [s]$ and $\ell\negsp\negsp \in\negsp\negsp [p_k]$, there is a (\textit{unknown}) group {$G_{\ell}^{\sexp{k}}\subseteq[m]$: items whose $k$-th attribute has value $\ell$:}
	\begin{align*}
		G_\ell^{\sexp{k}}\coloneqq \inbrace{i\in [m]\colon  z_i^{\sexp{k}}=\ell}.
	\end{align*}
	For example, if the $k$-th protected attribute is race, then for different values of $\ell\negsp\in\negsp [p_k]$, $G_{\ell}^{\sexp{k}}$ is the subset candidates whose race is $\ell$.
	However, we only have noisy information about the protected attributes of each item; so, only noisy information of this subset.

	\paragraph{Intersectional groups.}
	In the above model, each protected attribute takes a unique value.
	It may appear that this does not allow for intersectional groups, e.g., say multiracial candidates.
	But this is only a matter of encoding, and is remedied by using attributes such as `\texttt{has-raceA?}' and `\texttt{has-raceB?}', which take \texttt{Yes} or \texttt{No} values.

	\begin{definition}[{\bf Noise}]\label{def:individual_flipping_noise}
		For each item $i\in [m]$ and $k\in [s]$, we have a probability vector $q_i^{\sexp{k}}\in \Delta^{p_k}$, such that, the $k$-th protected attribute of item $i$ takes value $\ell\in [p_k]$ with probability $q_{i\ell}^{\sexp{k}}$ conditioned on $(w_i, \hz_i^{\sexp{1}}, \dots, \hz_i^{\sexp{s}}, a_i)$:
		\begin{align}
			q_{i\ell}^{\sexp{k}}\coloneqq \Pr\insquare{i\in G_{\ell}^{\sexp{k}}\mid (w_i, \hz_i^{\sexp{1}}, \dots, \hz_i^{\sexp{s}}, a_i)}.
		\end{align}
		The event that $(i\in G_{\ell}^{\sexp{k}})$ is independent of all other items $j\in [m]\backslash\{i\}$ and all other attributes in $[s]\backslash\{k\}$. {Note that for all $i\in[m]$ and $k\in [s]$, $\sum_{\ell\in [p_k]}q_{i\ell}^{\sexp{k}}=1$.}
	\end{definition}

	\paragraph{Discussion of the noise model.} The above model says that given the utility ($w_i$), noisy protected attributes ($\hz_i^{\sexp{1}}, \dots, \hz_i^{\sexp{s}}$), and nonprotected attributes ($a_i$) of an item $i$, there is probabilistic information about its protected attributes.
	If items represent candidates for a job and the protected attribute is race, then we can use the candidate's last name (encoded in $a_i$) to derive probabilistic information about their race. This has been used in practice, e.g., by~\cite{elliott2009UsingCencusSurnameList}.
	We can also consider multiple nonprotected attributes such as both last-name and location, e.g., as used by~\cite{elliott2006usingGeocoding,elliott2008new}.
	As discussed in \cref{sec:intro}, this could be relevant for an online hiring platform, which may not have demographic information of some or all of its users~\cite{linkedin_inferred_gender}, and {image search engines where the images do not have gender labels.}

	\subsection{Target problem}
	{Studies have found that,} in the absence of other constraints, the selection problem~\eqref{eq:classical_selection}, can overrepresent individuals with certain protected attributes at the expense of others~\cite{KayMM15,costello2016views}.
	Towards mitigating this bias, we consider lower bounds and upper bounds on the number of items of a given protected attribute selected.

	Formally, the constraints ensure that for each attribute-value pair $k\in [s]$ and $\ell \in [p_k]$, the selection has at least $L_{\ell}^{\sexp{k}}\geq 0$ and at most $U_{\ell}^{\sexp{k}}\geq 0$ items from $G_\ell^{\sexp{k}}$.
	Then, a selection $x\in \zo^m$ satisfies the {(target) fairness constraints if: for all $k\negsp\in\negsp [s]$ and $\ell\negsp\in\negsp [p_k]$}
	\begin{align*}
		L_{\ell}^{\sexp{k}} \leq  \sum\nolimits_{i\in G_\ell^{\sexp{k}}} x_{i} \leq U_{\ell}^{\sexp{k}}.\tagnum{Fairness constraints}\customlabel{eq:target_constraints}{\theequation}
	\end{align*}
	Constraints similar to Equation \eqref{eq:target_constraints} have been studied by several works in algorithmic fairness~{\cite{Chierichetti0LV17, celis2018ranking,Chierichetti0LV19}},
	and are rich enough to encapsulate a variety of fairness and diversity metrics (e.g., see \cite{celis2019classification}).
	Thus, for the appropriate $L$ and $U$, the subset satisfying constraints~\eqref{eq:target_constraints} would be fair for one from a large class of fairness metrics.

	Overall, our {\em constrained subset selection problem} is:
	\begin{align}
		\hspace{-20mm}\max_{x\in  \zo^m}\ \ &\sum\nolimits_{i=1}^m w_{i}x_{i}\tag{Target}\label{prob:target_fair}\\
		\st\quad\  & L_{\ell}^{\sexp{k}} \leq  \sum\nolimits_{i\in G_\ell^{\sexp{k}}} x_{i} \leq U_{\ell}^{\sexp{k}},\quad \forall \ k\in [s],\ \ell\in [p_k],\label{eq:target_fair:fairness_constraint}\\
		&\sum\nolimits_{i=1}^m x_i = n\label{eq:target_fair:cardinality_constraint}
		.
	\end{align}
	If we know the protected attributes, and in turn $G_\ell^{\sexp{k}}$ (for each $k\in [s]$ and $\ell\in [p_k]$), then we can hope to solve \target{} directly.
	Indeed, prior works consider similar problems in rankings~\cite{celis2018ranking}, or its generalization, to multiple Matroids constraints~\cite{Chierichetti0LV19}.
	However, with only noisy information about protected attributes, we can not even verify if a selection vector $x$ is feasible for \target{}.
	To overcome this, we must go beyond exact algorithms which always satisfy fairness constraints.

	\subsection{Denoised problem}\label{sec:noise_model}
	The difficulty in solving \target{} is that we do not know the constraints (as we do not know $G_\ell^{\sexp{k}}$).
	We propose to solve a different problem, \denoised{}, which uses the noise estimates $q$ to approximate the constraints of \target{}.
	For some small $\delta\in (0,1)$, we define the denoised program as the following
	\begin{mdframed}[style=FrameBox]
		\begin{align*}
			\hspace{-5.5mm}\max_{x\in \zo^m}\ \  &\sum\nolimits_{i=1}^{m}w_{i}x_{i}\tag{Denoised}\label{prob:denoised_fair}\\
			\hspace{-14mm}&\hspace{-12.5mm}\st\ L_{\ell}^{\sexp{k}}\negsp \negsp -\delta n\leq \sum\nolimits_{i=1}^m q_{i\ell}^{\sexp{k}}x_{i} \leq U_{\ell}^{\sexp{k}}\negsp \negsp +\delta n,\
			\forall \ k\in [s], \ell\in [p_k],\hspace{-3mm}
			\yesnum\label{eq:denoised_fair:fairness_constraint}\\
			&\hspace{-13.0mm}\quad\ \sum\nolimits_{i=1}^{m}x_i = n.\yesnum\label{eq:denoised_fair:cardinality_constraint}
		\end{align*}
	\end{mdframed}
	Here, $\sum_{i=1}^m q_{i\ell}^{\sexp{k}}x_{i}$ is the expected number of items selected by $x$ whose $k$-th protected attribute is $\ell$.
	Then, intuitively, we can see \denoised{} that satisfies the constraints of \target{} in expectation, where the expectation is taken over the noise.

	{However, just satisfying constraints in expectation is not sufficient. For instance, this would allow algorithms that,
	in each use, violate the fairness constraints by a large amount, but ``average out'' their errors in aggregate.}
	Instead, our goal is to find an algorithm which violates the constraints by at most a small amount, almost always.
	Before presenting our theoretical results, we discuss an alternate noise model and why it is not suitable in our setting.

	\subsection{Group-level noise model}\label{sec:group_level_noise}
	Recent works on noisy fair classification~\cite{LamyZ19, awasthi2020equalized, celis2020fairclassification} consider a different noise model, which adapted to our setting, uses the following probabilities
	$$\overline{q}_{i\ell} \coloneqq \Pr[i\in G_\ell\mid (\hz_i^{\sexp{1}},\dots,\hz_i^{\sexp{s}})].$$
	Notice that unlike \cref{def:individual_flipping_noise}, $\overline{q}_{i\ell}$ does not condition on the utility $w_i$ or nonprotected attributes $a_i$.
	Thus, its estimates are the same for all items with the same set of noisy protected attributes.
	We call this the {\em group-level noise model} (\texttt{GLN}).
	In the next example, we discuss why \texttt{GLN} is not sufficient to mitigate bias in subset selection. %

	\paragraph{Toy example.}
	Consider a setting where there is one protected attribute which takes two values (i.e., $s = 1$ and $p_1 = 2$),
	and the relevant fairness metric is equal representation.
	Let the two groups (unknown) be $A,B\subseteq[m]$, and their \mbox{observed noisy versions be $\wh{A}$ and $\wh{B}$.}\footnote{Formally, $A\negsp=\negsp\sinbrace{i\colon z_i^{\sexp{1}} =1}$, $B=\sinbrace{i\colon z_i^{\sexp{1}} =2}$, $\wh{A}=\sinbrace{i\colon \hz_i^{\sexp{1}} =1}$ and  $\wh{B}=\sinbrace{i\colon \hz_i^{\sexp{1}} =2}$.}

	{According to $\overline{q}$, each candidate $i\negsp\in\negsp\wh{B}$ has the same probability of being in $A$.
	In this noise model, these candidates are indistinguishable apart from their utilities, so, if one picks $n_b\in \N$ candidates from $\wh{B}$, they would naturally be the ones with the highest utility.
	However, suppose that most individuals in $A$ have a higher utility than most individuals in $B$.\footnote{Such an bias in utilities is one reason why we need fairness constraints in the first place~\cite{KleinbergR18,celis2020interventions,EmelianovGGL20}.}
	In this case, the probabilities $\overline{q}$ will be ``distorted'' by the utilities, such that, candidates with higher utility in $\wh{B}$ are more likely to be in $A$ than those with lower utility in $\wh{B}$.
	In fact, if $m$ is much larger than $n$, then most of the top $n_b$ candidates in $\wh{B}$ would, in fact, be from $A$.}
	This example can be extended to more realistic settings, with more than two groups and a smaller amount of ``bias'' in the utilities.
	Even then, to overcome this distortion in probabilities, one needs to consider a stronger noise model, in which the noise estimate varies with utility (as in Definition~\ref{def:individual_flipping_noise}); either implicitly through proxy nonprotected attributes ($a_i$) or explicitly with utility ($w_i$).

	\begin{remark}
		In Sections~\ref{sec:theoretical_results} and \ref{sec:empirical_results}, for the sake of simplicity, we only consider the setting with one protected attribute ($s=1$) which takes $p\geq 1$ values.
    {We obtain analogous results for the the general case in \suppMat~\ref{sec:extended_theoretical_results}.}
		When $s=1$, we let $p\coloneqq p_1$ and drop superscripts {(representing the protected attribute) from all variables.}
	\end{remark}

\section{Theoretical results}\label{sec:theoretical_results}
	Our main algorithmic result is an efficient approximation algorithm for \target{}.
	\begin{theorem}[{\bf An approximation algorithm for \target{}}]\label{thm:algorithm_for_target_fair}
		There is an algorithm (Algorithm~\ref{alg:algorithm_for_target_fair}) that given an instance of \target{} for $s=1$ and noise $q$ from \cref{def:individual_flipping_noise},
		outputs a selection $x\in \zo^m$,
		such that, with probability at least $1-4p\exp\inparen{\nfrac{-\delta^2 n}{3}}$ over the noise in the protected attributes of each item, the selection $x$
		\begin{enumerate}[leftmargin=15pt]
			\item has a value at least as \ul{high} as the optimal value of \target{},
			\item {violates the cardinality constraint~\eqref{eq:target_fair:cardinality_constraint} by at most $p$ (additive), and}
			\item {violates the fairness constraints~\eqref{eq:target_fair:fairness_constraint} by at most $(p+2\delta n)$ (additive).}
		\end{enumerate}
		{The algorithm runs in polynomial time in the bit complexity of input.}
	\end{theorem}
	\noindent As desired, the algorithm outputs subset which violates the constraints of \target{} by at most a {small amount, with high probability.}

	Note that the approximation is only in the constraints and not in the value: with high probability, $x$ has an higher value than the optimal solution, say \xtr{}, of \target{}, i.e.,
	$$\sum\nolimits_{i=1}^m x_i w_i \geq \sum\nolimits_{i=1}^m \xtr_i w_i.$$
	In most real-world contexts $p$ is a small constant.
	Here, \cref{thm:algorithm_for_target_fair} implies that $x$ violates the fairness constraints  (Equation~\eqref{eq:target_fair:fairness_constraint}) by a multiplicative factor of at most $(1+2\delta+o(1))$
	and the constraint Equation~\eqref{eq:target_fair:cardinality_constraint} by a multiplicative factor of at most $(1+o(1))$ with high probability.\footnote{Using $L_\ell, U_\ell\leq n$; if not, we can set $L_\ell$ to $\min(L_\ell, n)$ and $U_\ell$ to $\min(U_\ell, n)$.}
	If $p$ is large, then $x$ (from \cref{thm:algorithm_for_target_fair}) can violate the constraints by a large amount.
	However, in this case it is \np-hard to even check if there is a solution to \denoised{} which violates the constraints by a constant additive factor ({let alone finds an optimal solution for \target{}}); see Theorem~\ref{thm:hardness_results_main}.

	Algorithm~\ref{alg:algorithm_for_target_fair} crucially uses the \denoised{}:
	it first solves the linear-programming relaxation of \denoised{}, and then, ``rounds'' this solution to integral coordinates.
	In the next section, we overview the proof of \cref{thm:algorithm_for_target_fair}.
	We defer the proof of \cref{thm:algorithm_for_target_fair} to {\suppMat~\ref{sec:proof:thm:algorithm_for_target_fair} due to space constraints.}

	\begin{remark}\label{rem:only_lower_bounds}
		We can strengthen \cref{thm:algorithm_for_target_fair} to guarantee that  Algorithm~\ref{alg:algorithm_for_target_fair} finds an $x\in\zo^m$ which does not violate the lower bound fairness constraint (left inequality in Equation~\eqref{eq:target_fair:fairness_constraint})
		and violates the upper bound fairness constraints by at most $(p+2\delta n)$ (without changing other conditions).
		In particular, this shows that, if one places only lower bound fairness constraints, then subset found by Algorithm~\ref{alg:algorithm_for_target_fair} would never violate the fairness constraints.
	\end{remark}

	\vspace{-4mm}
	\setlength{\algomargin}{0.5em}
	\begin{algorithm}[h!]
		\AlgoDontDisplayBlockMarkers\SetAlgoNoEnd
		\caption{Algorithm for \target{}}
		\label{alg:algorithm_for_target_fair}
		\kwInit{A number $n\in\N$, probability matrix $q\in [0,1]^{m\times p}$, utility vector $w\in \R^m$, constraint vectors $L,U\in \R_{\geq 0}^p$.}\vspace{2mm}

		1. {\bf Solve} $x\gets$ {Find a basic feasible solution to linear-programming relaxation}

		\white{.}\hspace{20mm}  of \denoised{} with inputs $(n,q,w,L,U)$. %

		2. {\bf Set }\hspace{0mm} $x^\prime_i\coloneqq\ceil{x_i}\ $ for all $i\in [m]$.\commentalg{Round solution}

		3. {\bf Return} $x^\prime$.
	\end{algorithm}
	\vspace{-4mm}

	\subsection{Proof overview and hardness results}
	\label{sec:theoretical_results:proof_overview}
	In this section, we overview the proof of \cref{thm:algorithm_for_target_fair}.
	The complete proof and an extension of \cref{thm:algorithm_for_target_fair} for multiple protected attributes (i.e., $s\geq 1$) appear in \suppMat s~\ref{sec:proofs} and~\ref{sec:extended_theoretical_results}.

	The proof of \cref{thm:algorithm_for_target_fair} has two broad steps:
	First, we show that solving \denoised{} (even approximately) gives us a ``good'' solution to \target{}, and then we develop an approximation algorithm along with matching hardness results for \denoised{}.
	To prove the former, we bound the difference between the true and the expected {number of candidates from any one group $G_\ell$.}
	\begin{lemma}\label{thm:relation_between_target_and_noisy}
		For all $\delta\negsp \in\negsp (0,1)$ and $x\negsp \in\negsp  [0,1]^m$, s.t., $\sum_{i=1}^m x_i\negsp =\negsp n$:
		\begin{align*}
			\text{$\forall\ \ell\in [p]$,}\quad \abs{\sum\nolimits_{i\in G_\ell} x_i - \sum\nolimits_{i\in [m]}q_{i\ell} x_i } \leq n\delta
		\end{align*}
		holds with probability at least $1-2p\exp\inparen{\nfrac{-\delta^2 n}{3}}$ over the noise in the protected attributes of each item.
	\end{lemma}
	\noindent The proof of this lemma appears in \suppMat~\ref{sec:proof:thm:relation_between_target_and_noisy}.

	Using \cref{thm:relation_between_target_and_noisy}, we can show that any solution that violates the constraints of \denoised{} by a small amount, with high probability, also violates the constraints of \target{} by at most a small amount.
	Let \xtr{} be an optimal {selection} for \target{}.
	Using \cref{thm:relation_between_target_and_noisy}, we can show that \xtr{} is feasible for \denoised{} with high probability.
	It follows any solution $x$ which is optimal for \denoised{} has value at least as large as \xtr{}, i.e., $$\sum\nolimits_{i=1}^m x_i w_i \geq \sum\nolimits_{i=1}^m \xtr_i w_i.$$
	These suffice to show that, solving \denoised{} gives a ``good'' solution for \target{}---which satisfies the claims in the \cref{thm:algorithm_for_target_fair}.

	It remains to solve \denoised{}. Unfortunately, even checking if \denoised{} is {\em feasible} is \np-hard; see \cref{thm:hardness_results_main} (a constant-factor approximation (in utility) to \denoised{} is also \np-hard).
	We overcome this hardness by allowing solutions to violate the constraints of \denoised{} by a small additive amount ($p$).
	Towards this, consider the linear-programming relaxation of \denoised{} (for $s=1$).
	We show that any {\em basic feasible solution} (BFS) of \relaxation{} has a small number of fractional entries (\cref{thm:optimal_solution_with_few_fractional}).
	\begin{align}
		\max\nolimits_{x\in [0,1]^m}\quad  &\sum\nolimits_{i=1}^{m}w_{i}x_{i}\tag{LP-Denoised for $s=1$}\customlabel{prob:relaxed_denoised_fair}{LP-Denoised}\\
		\st\qquad \forall \ \ell \in [p],\qquad & L_\ell - \delta n\leq \sum\nolimits_{i=1}^{m} q_{i\ell}x_{i} \leq U_{\ell}+\delta n,\\
		&\sum\nolimits_{i=1}^{m}x_i = n.
	\end{align}
	\begin{lemma}[{\bf An optimal solution with $p$ fractional entries}]\label{thm:optimal_solution_with_few_fractional}
		Any basic feasible solution $x\in[0,1]^m$ of \relaxation{} has at most $\min(m,p)$ fractional values, i.e.,
		$\sum_{i=1}^{m}\mathbbm{I}[x_i\in(0,1)]\leq \min(m,p)$.
	\end{lemma}
	\noindent The proof follows by specializing well-known properties of BFSs to \relaxation{}.
	{We remark that this result is tight; see \cref{fact:instance_with_p_fractional_entries}.}\\[-1mm]

	\noindent {\em Proof sketch of \cref{thm:algorithm_for_target_fair}.}
	Using \cref{thm:relation_between_target_and_noisy}, we can show that \xtr{} is feasible for \denoised{} with probability at least $1-2p\exp\inparen{\nfrac{-\delta^2 n}{3}}$.
	Assume that this event happens.
	Then, \xtr{} is also feasible for \relaxation{}.
	Consider the basic feasible solution $x$ to \relaxation{} from Step 1 of Algorithm~\ref{alg:algorithm_for_target_fair}.
	Since $x$ is optimal for \relaxation{}, it follows that $x$ {has a value at least as large as \xtr, i.e.,} $$\sum\nolimits_{i=1}^m x_i w_i \geq \sum\nolimits_{i=1}^m \xtr_i w_i.$$
	Further, since $w\geq 0$, the rounded solution $x^\prime$ from Step 2 of Algorithm~\ref{alg:algorithm_for_target_fair} only increases the utility of $x$.
	Thus, $$\sum\nolimits_{i=1}^m x^\prime_i w_i \geq \sum\nolimits_{i=1}^m \xtr_i w_i.$$
	This establishes the first claim in \cref{thm:algorithm_for_target_fair}.

	It follows from \cref{thm:optimal_solution_with_few_fractional} that $x^\prime$ picks at most $p$ more elements than $x$.
	Thus, $x^\prime$ violates Equation~\eqref{eq:denoised_fair:cardinality_constraint}, so Equation~\eqref{eq:target_fair:cardinality_constraint} by at most $p$.
	By the same argument, $x^\prime$ violates the fairness constraints of \denoised{} by at most $p$ (additive).
	Combining this with \cref{thm:relation_between_target_and_noisy}, we can show that, with probability at least $1-2p\exp\inparen{\nfrac{-\delta^2 n}{3}}$,  $x^\prime$ violates the fairness constraints of \target{} by at most $2\delta n+p$ (additive).
	This establishes the last two claims in \cref{thm:algorithm_for_target_fair} {(conditioning on the two events described above).}

	The run time follows since there are polynomial time algorithms to find a basic feasible solution of a linear program.
	Finally, taking a union bound over over the two events completes the proof.

	\subsubsection{Hardness results} Lastly, we present our hardness results; their proofs appear in  \suppMat~\ref{sec:proof:hardness_results_main}.
	\begin{theorem}[{\bf Hardness results---Informal}]\label{thm:hardness_results_main}
		Consider variants of \denoised{} for values of $p$.
		\begin{enumerate}[leftmargin=15pt]
			\item If $p\geq 2$, then deciding if the problem is feasible is \np-hard.
			\item If $p\geq 3$, then the problem is \apx-hard.
			\item If $p=\poly(m)$ and $s>1$, then for every constant $c > 0$, the following {\em violation gap} variant of \denoised{} is \np-hard.
			\begin{itemize}[]
				\item Output YES if the input instance is satisfiable.
				\item Output NO if there is no solution which violates every upper bound constraint at most an additive factor of $c$.
			\end{itemize}
		\end{enumerate}
	\end{theorem}

\newcommand{\customfootnotetext}[2]{{\renewcommand{\thefootnote}{#1}\footnotetext[0]{#2}}}
\section[Empirical results]{Empirical results\textsuperscript{$\star$}}\label{sec:empirical_results}
\customfootnotetext{$\star$}{The code for the simulations is available at \url{https://github.com/AnayMehrotra/Noisy-Fair-Subset-Selection}.}
	We evaluate our approach on utilities and noise derived from both synthetic and real-world data.
	{We consider the following algorithms:}

	\begin{enumerate}[itemsep=0pt, label={-}]\item[]

		\hspace{-10mm} {\bf Baseline.}

		\item \Blind{}: As a baseline, we consider the \Blind{} algorithm which selects $n$ candidates with the highest utility. Note that \Blind{} has the optimal unconstrained utility.\\[-2mm]

		\hspace{-10mm} {\bf Noise aware.}

		\item \FairExpec{} is our proposed approach (see \cref{thm:algorithm_for_target_fair}).

		\item \FairExpecGrp{} is the same as \FairExpec{} but uses the probabilities $\overline{q}_{i\ell}\negsp\coloneqq\negsp \Pr[i\negsp\in\negsp G_\ell \mid \hz_i]$ {from the group-level noise model (\cref{sec:group_level_noise}).}\\[-2mm]

		\hspace{-10mm} {\bf Noise oblivious.}
		\item[]\hspace{-10mm} Impute protected attributes  Bayes-optimally from $q\negsp\in\negsp[0,1]^{m\times p}$ as:
		\begin{align*}
			\forall\ i\in [m],\ \ell\in[p],\quad q_{i\ell}^\prime \coloneqq \begin{cases}
			1 & \text{if } \ell\in \argmax_{j\in [p]}q_{ij},\\
			0 & \text{otherwise}.
		\end{cases}\yesnum\label{eq:threshold}
	\end{align*}
	\hspace{-10mm} If $\argmax_{j}q_{ij}$ is not unique, pick one at random.
	Then, we consider the following noise oblivious algorithms

	\hspace{-10mm} which take the imputed protected attributes $q^\prime$ as input:\\[-2mm]

	\item $\Thresh{}$ solves \target{} defined on $q^{\prime}$.
	This is equivalent to the ranking algorithms of {\cite{celis2018ranking, SinghJ18}} adapted to subset selection.
	\item \MultObj{} is a multi-objective optimization algorithm inspired by \cite{YangS17}'s approach for ranking.
	Let $t\in \Delta^p$ be the target distribution of protected attributes in the selection.
	For example, if the target is equal representation, then $t \coloneqq (\nfrac1p,\dots,\nfrac1p)\in \R^p$.
	Given a constant $\lambda>0$, \MultObj{} solves\footnote{Given two vectors $x,y\in \Delta^p$, $D_{\rm KL}(x,y)$ denotes the Kullback–Leibler divergence of $x$ and $y$ defined as $D_{\rm KL}(x,y)\coloneqq \sum_{i=1}^p x_i\log(\nfrac{x_i}{y_i})$}
	\begin{align*}
		\max_{x\in [0,1]^n\colon \sum_i x_i = n}\ \  w^\top x - \lambda\cdot D_{\rm KL}\inparen{\frac{(q^\prime)^\top x}{n},t}\cdot \frac{w^\top 1_m}{m},
	\end{align*}
	where $1_m\in \R^m$ is the all one vector.
	The first term $w^\top x$ is the value of $x$, $(\nfrac{x}{n})$ is the distribution of noisy protected attributes in $x$, and entire second term is a penalty on $x$ for being far from the target distribution $t$.\footnote{We scale the second term by the average utility $\sum_{i=1}^{m}\nfrac{w_i}{m}$.
	This is not necessary, but ensures that $\lambda$ does not (heavily) depend on the scale of the utility.}
	\end{enumerate}

	\subsection{Setup and metrics}\label{sec:general_setup}
	\subsubsection{Setup}
	We consider one protected attribute ($s=1$) which takes $p$ disjoint values (we use $p=2$ and $p=4$).
	Our simulations either target equal-representation, where $t= (\nfrac1p,\dots,\nfrac1p)\in \R^p$, or proportional representation, where $t= (\nfrac{|G_1|}{m},\dots,\nfrac{|G_p|}{m})$.\\
	In each simulation, we do the following %
	\begin{itemize}[leftmargin=*]
		\item {\FairExpec{}, \FairExpecGrp{}, and \Thresh{}:} Set $L_\ell=0$ and $U_\ell=n(1-\alpha)+n\alpha t_\ell$, and vary $\alpha$ from $0$ to $1$.
		Notice that $\alpha=0$ enforces no constraints on the subset, the constraints become tighter as $\alpha$ increases, and $\alpha=1$ ensures the subset chooses exactly $n=t_\ell$ candidates from the $\ell$-th group.
		\item \MultObj{}: Vary $\lambda$ from $0$ to a large value. %
		Here, $\lambda=0$ enforces no penalty on the objective, the penalty increases as $\lambda$ increases, and $\lambda=\infty$ forces \MultObj{} to satisfy the target distribution exactly (on the noisy attributes). %
	\end{itemize}
	Let $(\alpha_r, \lambda_r)$ be the $r$-th choice of $\alpha$ and $\lambda$.
	For each $(\alpha_r,\lambda_r)$, we draw a set $M$ of $m$ individuals or items from the dataset.
	For each element $i\in M$, we have $q_i\in \Delta^p$ and $w_i\in \R$.
	We give the details of drawing $M$ and fixing $q_i, w_i$ with each simulation.

	\subsubsection{Fairness metric} Given subset $S\in [m]$ and target $t\in[0,1]^p$, let the {\em risk difference} $\cF(S,t)\in [0,1]$ of $S$ for target $t$ be
	\begin{align}
		\cF(S,t)\coloneqq 1-\min_{\ell\in [p]} t_\ell\ \cdot\ \max_{\ell,k\in [p]}\inparen{\frac{|S\cap G_\ell|}{n\cdot t_\ell} - \frac{|S\cap G_k|}{n\cdot t_k}}.%
	\end{align}
	Here, a risk difference 1 is the most fair and 0 is the least fair.
	When the target is proportional representation, $\cF(S,t)$ reduces to the usual definition of risk difference (up to scaling).\footnote{Some works also define risk difference as a measure of unfairness~\cite{calders2010three, ruggieri2014using}, and set it equal to $1-\cF(S,t)$ with $t=(1/p,\dots,1/p)$ (up to scaling).} %
	Let $\cA(w,q)\subseteq [m]$ be the subset selected by algorithm $\cA$ on input $(w,q)$.
	We report $$\cF_\cA\coloneqq \Ex\insquare{\cF(\cA(w,q),t)},$$ where the expectation is over the choices of $(w,q)$.

	\subsubsection{Utility metric}
	{Let $U_{\cA}$ to be the average utility obtained by $\cA$:}
	$$U_{\cA} \coloneqq \Ex\insquare{\sum\nolimits_{i\in \cA(w,q)} w_i},$$
	where the expectation is over the choices of $(w,q)$.
	We report the utility ratio $\cK_{\cA} \in [0,1]$ for different algorithms $\cA$, defined as
	\begin{align*}
		\cK_{\cA}\coloneqq \frac{U_{\cA}}{U_{\Blind{}}}. \tag{Utility ratio}
	\end{align*}
  When the algorithm $\cA$, is not important or clear from context, we drop the subscripts from $\cF_\cA$ and $\cK_\cA$.

	\subsection{Synthetic data with disparate error-rates}\label{sec:toy_simulation}
	In this simulation, we consider the setting where different groups have different noise levels.
	This has been observed in practice, for instance, in commercial image-based gender classifiers~\cite{BuolamwiniG18}.
	\subsubsection{Data}
	We generate a synthetic dataset with one binary protected attribute ($p=2$).
	This attribute partitions the (underlying) population into a minority group {(40\%)} and a majority group {(60\%)}.
	We assume that candidates in both groups have similar potentials, so, sample utilities of all candidates (independently) from $\cU(0,1)$.
	Next, we sample the probabilities $q_i$ from a Gaussian mixture, such that, the resulting population has 40\% minority candidates (in expectation),
	and the imputed attributes $q^\prime$ have a higher \textit{false discovery rate} (FDR) for minority candidates ($\approx$40\%) compared to majority candidates ($\approx$10\%).\footnote{The difference of 30\% in FDRs is comparable to 34\% difference in FDRs between dark-skinned females and light-skinned men observed by~\cite{BuolamwiniG18} for a commercial classifier.}
	Formally, we sample $q_{i}$ as follows:
	$$q_{i0}\sim \frac{7}{11}\cdot  \mathcal{N}_{}(0.6,0.05) + \frac{4}{11}\cdot \mathcal{N}_{}(0.05,0.05)\quad\text{and}\quad q_{i1}\coloneqq 1-q_{i0},$$
	where $\mathcal{N}(\mu,\sigma)$ is the truncated normal distribution on $[0,1]$ with mean $\mu$ and standard deviation $\sigma$.

	\subsubsection{Setup}
	In this simulation, we target equal representation between the majority group and the minority group, and
	fix $m=500$ and $n=100$.

	We report the risk difference ($\cF$) of different algorithms as a function of $\alpha$ (for \FairExpec{}, \FairExpecGrp{}, and \Thresh{}) and as a function of $\lambda$ (for \MultObj{}) in Figure~\ref{fig:toy_simulation:1}.
	\renewcommand{\folder}{./figures/disparate-error}
	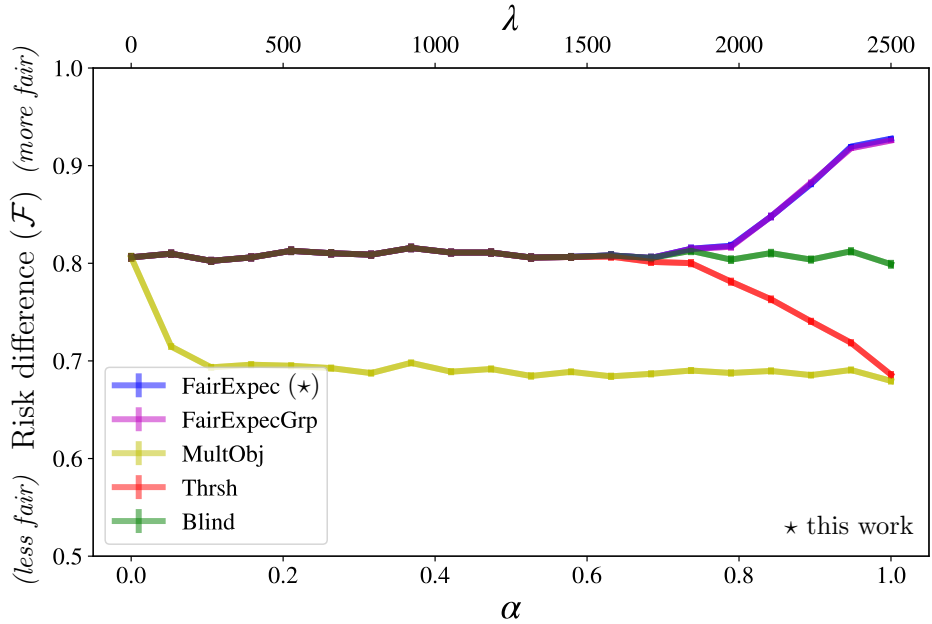
\begin{figure}[t!]
		\vspace{-3mm}
		\centering
		\begin{tikzpicture}
			\node (image) at (0,0) {\includegraphics[width=0.8\linewidth, trim={1cm 0cm 1cm 0.3cm},clip]{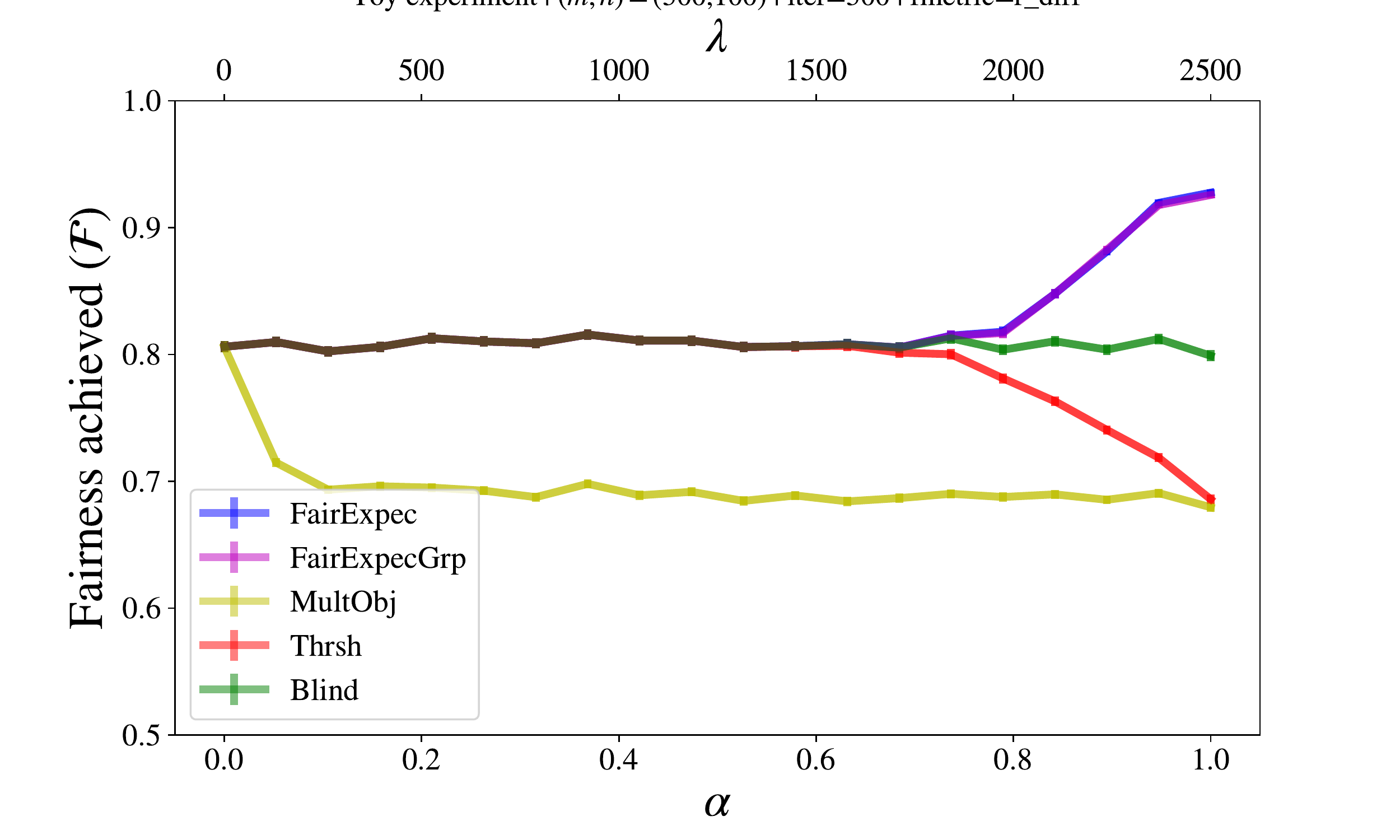}};
			\node[rotate=90, fill=white] at (-4.1*1.53,0) {\Large\white{AAAAAAAAAAAAAAAAAAA}};
			\node[rotate=90] at (-4.1*1.53  ,1.75*1.6) {\textit{(more fair)}};
			\node[rotate=90, fill=white] at (-4.1*1.53  ,0) {\large Risk difference ($\cF$)};
			\node[rotate=90] at (-4.1*1.53  ,-1.75*1.6) {\textit{(less fair)}};
			\node[rotate=0] at (-1.65*1.56,-0.55*1.6) {{$(\star)$}};
			\node[rotate=0] at (2.85*1.64,-1.7*1.6) {{$\star$ this work}};
		\end{tikzpicture}
		\vspace{-4mm}
		\caption{\small
		{\em {Synthetic data with disparate error-rate (\cref{sec:toy_simulation}):}}
    This simulation considers the setting where the minority group (40\% of total) has a higher 30\% higher FDR compared to the majority group.
    The utilities of all candidates are iid from the uniform distribution.
		The target is to ensure equal representation between the majority and minority groups.
		The $y$-axis shows the risk difference $\cF$ of different algorithms, and the $x$-axis shows the constraint parameters ($\alpha$ for \FairExpec{}, \FairExpecGrp{}, and \Thresh{}, and $\lambda$ for \MultObj); $\cF$ values are averaged over 500 trials, and the error bars represent the standard error of the mean.
		We observe that increasing fairness constraints to noise oblivious algorithms (\Thresh{} and \MultObj{}) worsens their risk difference!
		Whereas, the risk difference of noise aware algorithms improves (\FairExpec{} and \FairExpecGrp{}) on increasing fairness constraints.
		\vspace{-3mm}
		}
		\label{fig:toy_simulation:1}
	\end{figure}

	\begin{remark}
		\MultObj{} does not guarantee a particular fairness-level for any fixed $\lambda$.
		Thus, one should consider the limiting value of $\cF_{\MultObj{}}$ in Figure~\ref{fig:toy_simulation:1}.
	\end{remark}

	\subsubsection{Results}
	We observe that without any constraints (i.e., $\alpha=0$ and $\lambda=0$) all algorithms have similar risk difference $\approx 0.81$.
	However, on adding fairness constraints \Thresh{} and \MultObj{} become more {\em unfair}.
	In fact, for the strongest fairness constraint (i.e, $\alpha=1$ and $\lambda=2500$) they have the {\em lowest} risk difference ($<0.7$).
	This is because, the imputed protected attributes have a higher FDR for the minority group (so, the algorithms pick a  higher number of candidates from the majority group).

	In contrast, \FairExpec{} and \FairExpecGrp{} do not use the imputed protect attributes, so increasing fairness constraints to \FairExpec{} and \FairExpecGrp{} improves their risk difference, and for the strongest fairness constraint ($\alpha=1$) they attain the highest risk difference ($>0.92$).
	Finally, since we sample all utilities from the same distribution, it is not surprising that \FairExpec{} and \FairExpecGrp{} perform similarly.

	\subsection{Synthetic data with disparate utilities}\label{sec:causal_simulation}
	In this simulation, we consider the setting where different groups have different distributions of utilities.
	In particular, we assume that the minority group (unfairly) has a lower average utility, when in fact, the distributions of utilities should be the same for both the majority group and the minority group.
	(Contrast this with \cref{sec:toy_simulation} where all utilities are identically drawn).
	Such differences in utility can manifest in the real world for many reasons, including, the implicit biases of the committee evaluating the candidates~\cite{wenneras2001nepotism, bertrand2004emily, uhlmann2005constructed} and structural oppression faced by different groups~\cite{faught2017association}.

	\paragraph{Counterfactually fair approaches.}
	One could also consider counterfactually fair approaches to mitigate bias in selection.
	(We refer the reader to \cite{kusner2017counterfactual} for an overview of counterfactual fairness).
	At a high-level, these approaches try to ``unbias'' the utilities across groups, and then use unbiased utilities in subsequent tasks (say, selection or ranking).
	In this simulation, we also consider counterfactually fair algorithms by \cite{yang2020causal}: \CntrFair{} and \CntrFairRes{}. (They correspond {to non-resolving and resolving algorithms in \cite{yang2020causal}.)}

	{Roughly, they assume that there is a causal model $\cM[\theta]$ (parameterized by $\theta$), such that, given the attributes $(z_{i}, a_i)$ of an individual, their utility is $w_i = \cM[\theta](z_i, a_i)$.}
	{Then, roughly, they fix each individual's protected attributes to $v$ and compute the ``unbiased'' utility as $\hat{w}_i\coloneqq \cM[\theta](v, a_i)$; this represents the utility of the individuals had they had {the same protected attribute $v$.}}

	\subsubsection{Data}
	We consider a synthetic hiring dataset, generated with the code provided by~\cite{yang2020causal}.
	In the data, each candidate $i$ has one protected attribute $z_i\in \zo$
	denoting their race (0 if the candidate is Black and 1 otherwise) and
	two nonprotected attributes $a_{i1}\in \zo$ and $a_{i2}\in \R$:
	$a_{i1}$ is 1 if the candidate has prior work experience, and $a_{i2}$ is denotes their qualifications (the larger the better).\footnote{This interpretation differs from \cite{yang2020causal}, who interpret both $z_i$ and $a_{i1}$ as protected attributes.}
	We sample 2000 candidates independently from a fixed distribution defined in \cite{yang2020causal}, which, is such that, the utility of Black candidates is (unfairly) lower than non-Black candidates.\footnote{The only difference from \cite{yang2020causal} is that we increase the underlying bias (by reducing the mean utilities for Black candidates and candidates without prior experience).
	We do so because, the dataset already had a high risk difference ($>0.9$) without adding any fairness constraints.}
  {For further details, we refer the reader to \suppMat{}~\ref{sec:implementation_details:causal_simulation}.}

	\paragraph{Preprocessing.}
	{We sample a training dataset $D^\prime$ with $m=2000$ candidates.
  Then, using $D^\prime$ we compute an approximation $\bar{\theta}$ of $\theta$, and given a candidate $i$ described by $(w_i, z_i, a_{i1}, a_{i2})$, we compute $q_i\in [0,1]^2.$
  (Note that the candidate $i$ may not be in $D^\prime.$)
  For prefer more details of the preprocessing, {in \suppMat{}~\ref{sec:implementation_details:causal_simulation}.}}
	\renewcommand{\folder}{./figures/disparate-utilities}
	\begin{figure}[t!]
		\begin{center}
			\vspace{-3.0mm}
			\begin{tikzpicture}
				\node (image) at (-0.5,0) {\includegraphics[width=0.8\linewidth, trim={1.8cm 1cm 2cm 1.6cm},clip]{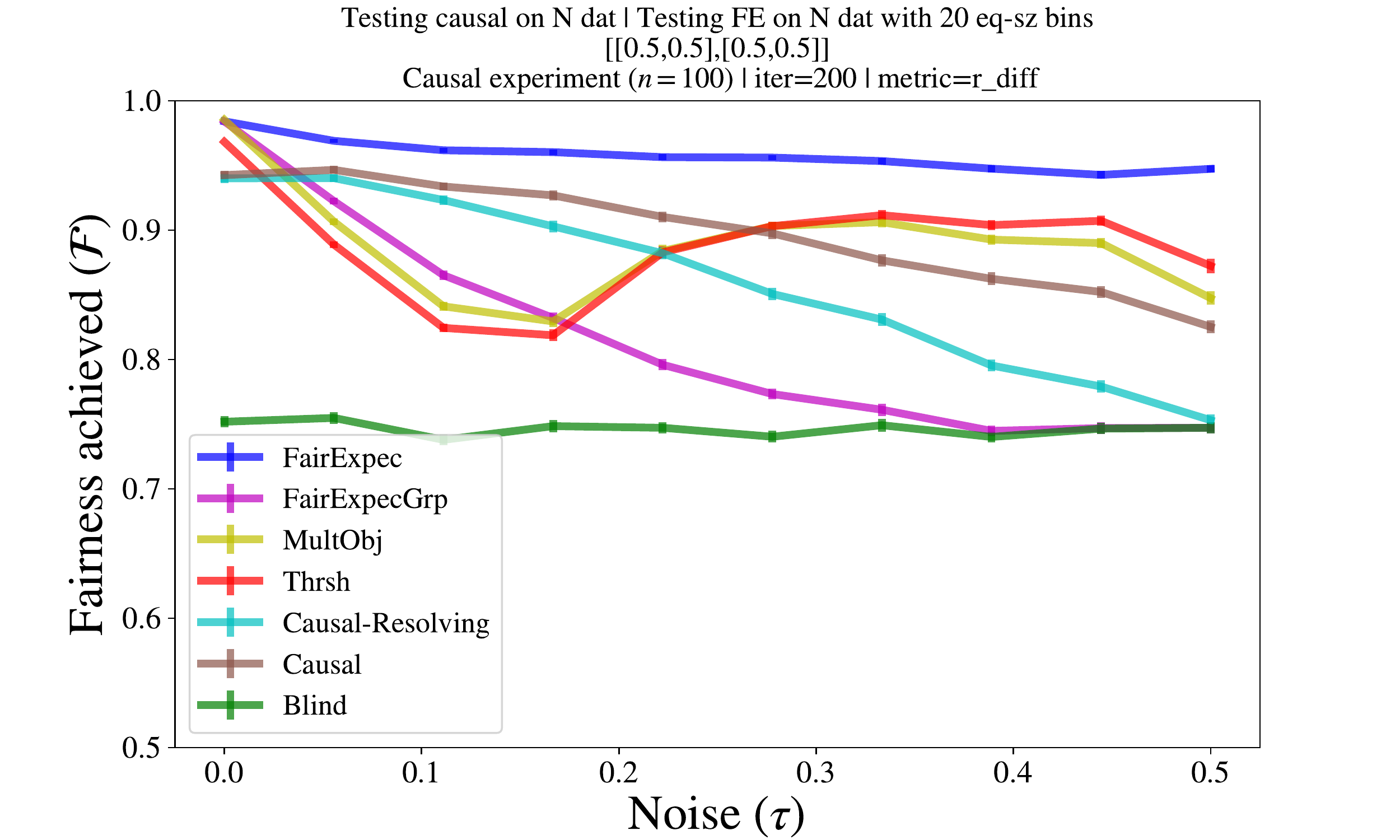}};
				\node[rotate=0, fill=white] at (-0.3*1.70,2.3*1.75) {\large \white{AAAAAAAAAAAAAAAAAAAAA}};
				\node[rotate=90, fill=white] at (-4.4*1.63,-0.5*1.77) {\large \white{AAAAAAAAAAAAAAAAAAAAA}};
				\node[rotate=90] at (-4.6*1.60,1.75*1.77) {\textit{(more fair)}};
				\node[rotate=90, fill=white] at (-4.6*1.60,0*1.77) {\large Risk difference ($\cF$)};
				\node[rotate=90] at (-4.6*1.60,-1.75*1.77) {\textit{(less fair)}};
				\node[rotate=0] at (2.5*1.83,-2.5*1.77) {\textit{(more noise)}};
				\node[rotate=0] at (-0.3*0.8,-2.5*1.77) {\large Noise ($\tau$)};
				\node[rotate=0] at (-3.4*1.60,-2.5*1.77) {\textit{(less noise)}};
				\node[rotate=0] at (-2.1*1.60,-0.125*1.77) {{$(\star)$}};
				\node[rotate=0] at (2.3*1.83,-1.7*1.77) {{$\star$ this work}};
				\node[rotate=0,fill=white] at (-2.78*1.59,-1.45*1.75) {{CntrFair}};
				\node[rotate=0,fill=white] at (-2.33*1.58,-1.175*1.75) {{\white{CntrFairRe}}};
				\node[rotate=0,fill=white] at (-2.63*1.56,-1.175*1.75) {{CntrFairRes}};
			\end{tikzpicture}
		\end{center}
		\vspace{-4mm}
		\caption{\small\small
		{\em {Synthetic data with disparate utilities (\cref{sec:causal_simulation}):}}
		{This simulation considers the setting where the utilities of a minority group have a lower average than the majority group, and both groups have an identical amount of noise.
    The target is to ensure proportional representation.}
		{The $y$-axis shows the risk difference $\cF$ of different algorithms, and the $x$-axis shows the amount of noise added $\tau\in[0,\nfrac12]$; $\cF$ values are averaged over 200 trials, and the error bars represent the standard error of the mean.}
		We observe that the risk difference of all algorithms becomes poorer with noise ($\tau>0$) than without it ($\tau=0$).
		Here, \FairExpec{} has the highest risk difference for all values of noise.
		Finally, unlike \cref{sec:toy_simulation}, \FairExpecGrp{} has a lower fairness than \FairExpec{} (since, in this simulation, {different groups have different distributions of utilities).}
		\vspace{-3mm}
		}
		\label{fig:causal_simulation}
	\end{figure}
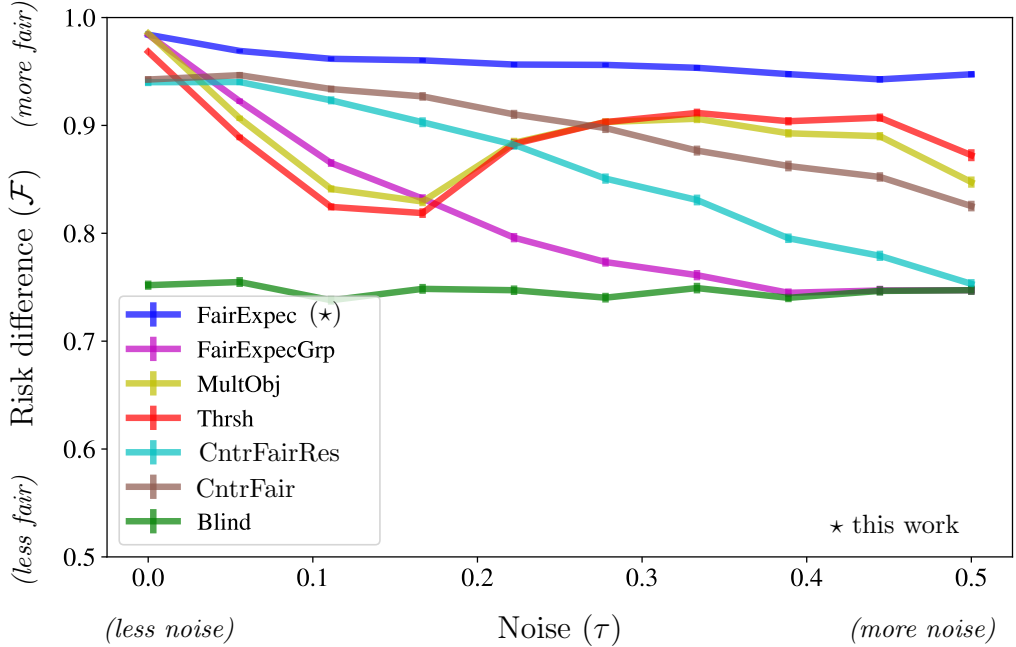

	\paragraph{Adding noise.}
	The dataset does not have noise to begin with.
	Given a noise level $\tau\in [0,\nfrac12]$, we generate noisy race $\widehat{z}_i$ of candidate $i$ by flipping their {race $z_i$ (independently) with probability $\tau$.}

	\subsubsection{Setup}
	We target proportional representation of race and vary $\tau$ over $[0,0.5]$.
	For each noise level $\tau$, we sample a new instance $D$ of the dataset with $m=2000$ and add $\tau$-noise to it.
	We fix $n=100$ and the strongest constraints $\alpha=1$ and $\lambda=500$ for the algorithms.\footnote{We choose $\lambda=500$ as \MultObj{}'s fairness $\cF_{\MultObj}$ converges before $\lambda=500$.}
	Here, \CntrFair{} and \CntrFairRes{} use $\cM(\bar{\theta})$ (calculated in preprocessing), and \FairExpec{}, \FairExpecGrp{}, \MultObj{}, and \Thresh{} use $q$ (or the imputed attributes $q^\prime$; both calculated in preprocessing). %

	We report the risk difference ($\cF$) as a function of the noise-level ($\tau$) in Figure~\ref{fig:causal_simulation}.
	We also report the selection rates from each group in \suppMat{}~\ref{sec:extended_empirical_results:causal_simulation}.

	\subsubsection{Results}
	We observe that all algorithms have the highest fairness when there is no noise ($\tau=0$).
	Here, they have a similar risk difference (lying between $0.94$ to $0.98$).
	As the noise increases, we observe that the risk difference of \FairExpecGrp{} and \CntrFair{} approaches $\cF_{\Blind{}}=0.74$, and risk difference of \MultObj{}, \Thresh{}, and \CntrFairRes{} approaches a value between $[0.82,0.87]$.
	In contrast, \FairExpec{} has a better risk difference, $\cF>0.94$, throughout. %
	The risk difference of \MultObj{} and \Thresh{} improves with $\tau$ at some values of $\tau$ --- we give a possible explanation in Remark~\ref{rem:explanation_higher_fairness_of_thresh_mult_obj}.
	Notice at $\tau=0.5$, for all candidates $i\in [m]$, the noisy label $\wh{z}_i\in \zo$ is chosen uniformly at random and provides no information about $z_i$.
	\CntrFair{} and \CntrFairRes{} use $\wh{z}_i$ to compute the counterfactual utilities, so, perform poorly at $\tau\approx 0.5$.
	Further, \FairExpecGrp{} uses the probabilities $\overline{q}_i$ which depend on $\widehat{z}_i$, but not on $w_i$.
	Since the utility of candidates of different races has a different distribution, $\overline{q}_i$ can be skewed (see Section~\ref{sec:group_level_noise}).

	We note that the utility of all algorithms decreases on adding noise.
	In particular, while \FairExpec{} is able to satisfy the fairness constraints with noise, its utility decreases on adding noise; see \suppMat~\ref{sec:extended_empirical_results:causal_simulation} {for a plot of utility ratio ($\cK$) vs noise ($\tau$).}

	\begin{remark}\label{rem:explanation_higher_fairness_of_thresh_mult_obj}
		{The risk difference of \Thresh{} and \MultObj{} is non-monotonic in the noise.
		This might be because the false discovery rate (FDR) of $q^\prime$ for Black candidates is non-monotonic.
		Specifically, the FDR first increases with $\tau$ (roughly, for $\tau \leq 0.2$), and then decreases.
		The decrease in FDR after $\tau=0.2$ comes at the cost of fewer total positives (i.e., $q^\prime$ identifies fewer total Black candidates).
		The total number of total positives drop below $\nfrac{n}{2}$ for higher values of $\tau$.
		Correspondingly, the $\cF$ of \Thresh{} and \MultObj{} first decreases as FDR reduces, then increases as FDR increases until the number of total positives is larger than, roughly, $\nfrac{n}{2}$, and finally, decreases as the number of total positives drops below $\nfrac{n}{2}$.}
	\end{remark}

	\begin{remark}%
		We do not consider counterfactual approaches in \cref{sec:toy_simulation} because there, the utilities are already unbiased, and so, \CntrFair{} and \CntrFairRes{} reduce to \Blind{}.
		Further, \CntrFair{} and \CntrFairRes{} only ensure proportional representation.
		We find that the datasets considered in Sections~\ref{sec:candidate_selection} and  \ref{sec:image_selection} are already fair when the target is proportional representation; in both cases, $\cF_\Blind{}>0.9$.
		Therefore, we omit these algorithms from those simulations.
	\end{remark}

	\subsection{Real-world data for candidate selection}\label{sec:candidate_selection}
	In this simulation, we consider the problem of selecting candidates under noisy information about their race.
  Similar to what has been used in applications (e.g.,~\cite{elliott2009UsingCencusSurnameList}), we use a candidate's last name to predict their race.
	We consider a candidate's utility as their ``previous-salary,'' which we model using the race-aggregated income dataset~\cite{census_income_dataset}. This dataset provides the income distribution of families from different races (elaborated below).
  This problem could be relevant in the context of an online hiring platform, which would like to display a race-balanced set of candidates, but only has noisy information about each candidate's race~\cite{linkedin_inferred_gender}.

	\subsubsection{Data}
	\paragraph{US 2010 Census dataset~\cite{census_data_10}.} %
	The dataset contains 151,671 distinct last names which occurred at least a 100 times in the US 2010 census (23,656 last names occur at least a 1000 times).
	For each last name $i$, the dataset has its total occurrences per 100k people $c(i)\in \Z$, and a vector $\hat{f}=(f_1(i),\dots,f_6(i))\in [0,1]^6$ representing the fraction of individuals who are White, Black, Asian and Pacific Islander (API), American Indian and Alaskan Native only (AIAN), multiracial, or Hispanic respectively.

	We do not use `AIAN' and `two or more races' categories (i.e., $f_4$ and $f_5$) as they do not occur in the income dataset.
	Then, for each last name $i$, we define the probability vector $q_i$ as the normalized version of the vector $(f_1(i),f_2(i),f_3(i),f_6(i))$.

	\renewcommand{\folder}{./figures/candidate-subset-selection}
	\begin{figure}[t!]
		\vspace{-3mm}
		\begin{center}
			\begin{tikzpicture}
				\node (image) at (-0.5,0) {\includegraphics[width=0.8\linewidth, trim={1.7cm 1.2cm 0.5cm 1.7cm},clip]{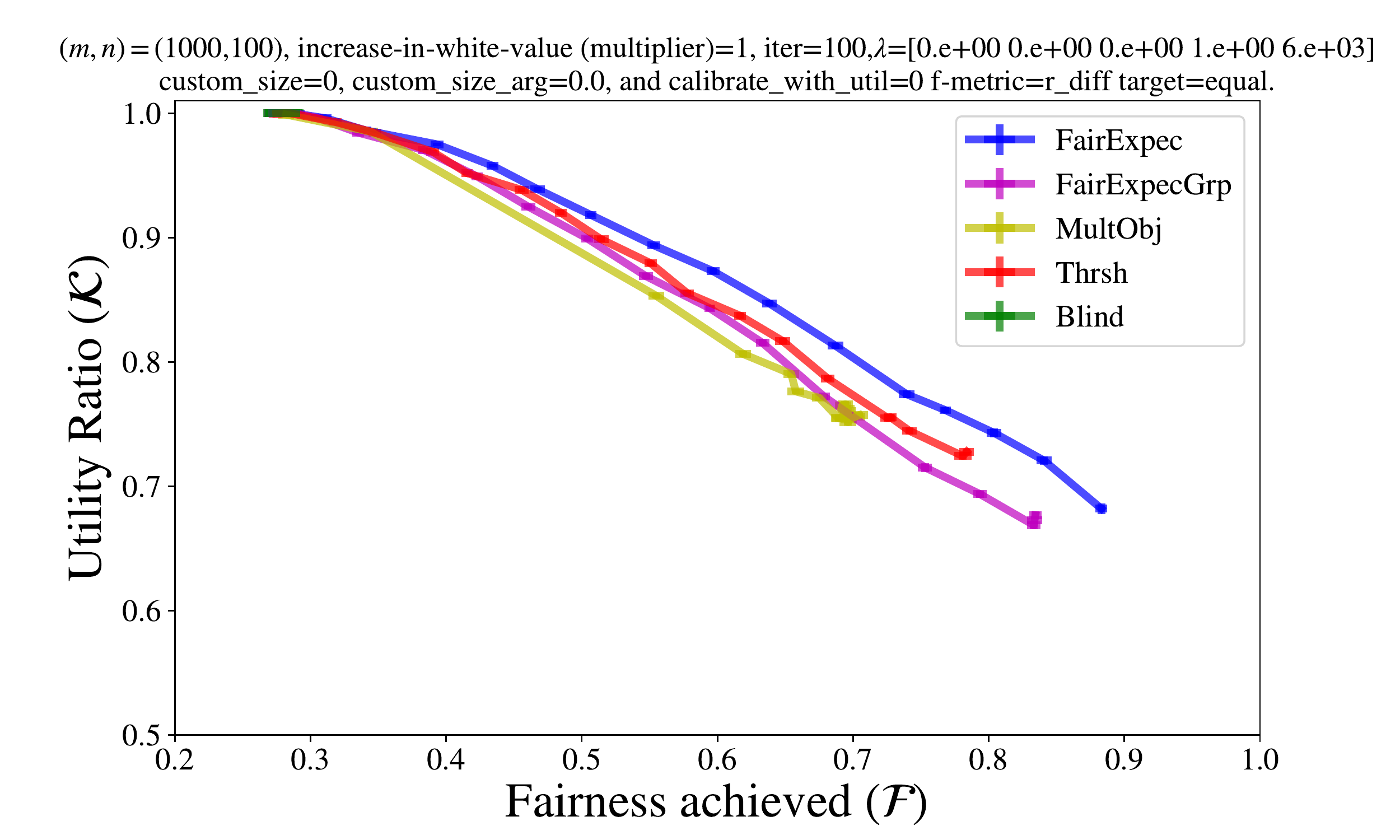}};
				\node[rotate=0, fill=white] at (-0.45*1.56,2.28*1.64) {\large \white{AAAAAAAAAAAAAAAAAAAAAAAAAAAAAAAAAAAA}};
				\node[rotate=90, fill=white] at (-4.6*1.56,-0.5*1.68) {\large \white{AAAAAAAAAAAAAAAAAAAAA}};
				\node[rotate=90, fill=white] at (-4.6*1.56,0.15) {\large Utility ratio ($\cK$)};
				\node[rotate=0] at (2.3*1.76 ,-2.4*1.68) {\textit{(more fair)}};
				\node[rotate=0] at (-0.3*1.56,-2.4*1.68) {\large Risk difference ($\cF$)};
				\node[rotate=0] at (-3.6*1.56,-2.4*1.68) {\textit{(less fair)}};
				\node[rotate=0] at (2.525*1.76, 1.85*1.64 ) {{$(\star)$}};
				\node[rotate=0] at (2.2*1.76,-1.7*1.64 ) {{$\star$ this work}};
			\end{tikzpicture}
		\end{center}
		\vspace{-4mm}
		\caption{\small
		{\em {Real-world data for candidate selection (\cref{sec:candidate_selection}):}}
		{This simulation considers race as the protected attribute which takes $p=4$ values.
    The simulation uses last-name as a proxy to derive noisy information of race and draws the utility of each candidate from a fixed distribution depending on their race.
    The target is to ensure equal representation across race.}
		The $y$-axis shows the utility ratio $\cK$ of different algorithms, and the $x$-axis shows the risk difference of different algorithms; both $\cK$ and $\cF$ values are averaged over 100 trials, and the error bars represent the standard error of the mean.
		We observe that \FairExpec{} reaches the highest risk difference ($\cF=0.89$), and has a better tradeoff between utility and risk difference compared to other algorithms.
		}
		\vspace{-3mm}
		\label{fig:candidate_selection:1}
	\end{figure}
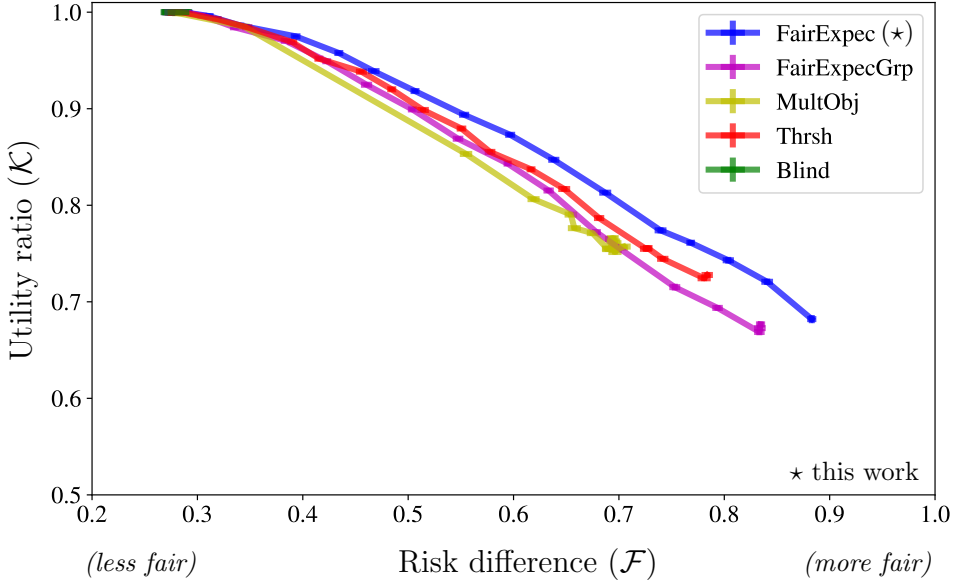

	\newcommand{\negspb}{\hspace{0mm}}
	\paragraph{Income dataset~\cite{census_income_dataset}.}
	We use family income data aggregated by race~\cite{census_income_dataset}.
	This was compiled by the US Census Bureau from the Current Population Survey 2018~\cite{current_population_survey}.
	The dataset provides income data of 83,508,000 families.
	It {has} four races (White, Black, Asian, and Hispanic), 12 age categories, and 41 income categories.%
	\footnote{The age categories are: 15 to 24, 25 to 30, 30 to 35, \dots. %
	The income categories \negspb{} are: $[0,5000),$ $[5000,10^4), $ \dots, $[1.95,2\cdot 10^5),$\negspb{} and\negspb{} $(2\cdot 10^5,\infty)$\negspb{} in\negspb{} USD\negspb{} per\negspb{} annum.
	}
	For each set of race, age, and income categories, the dataset has the number of families whose reference person (see definition \href{https://www.census.gov/programs-surveys/cps/technical-documentation/subject-definitions.html}{here}) belongs to these categories.

	For each race $r$, we consider the discrete distribution $\cD_r$ of incomes of families with race $r$ derived from the income dataset~\cite{census_income_dataset}; see Figure~\ref{fig:stats_of_dr} in supplementary material for some statistics of $\cD_r$.

	\subsubsection{Setup}\label{sec:candidate_selection:ground_population}
	We consider race as a protected attribute with four labels ($p=4$) and target equal representation based on race.
	Let $m=1000$ and $n=100$.
	For each choice of $\alpha$ and $\lambda$, we draw a set $M$ of $m$ last names uniformly from the entire population with replacement: The $i$-th last name is drawn with probability proportional to $c(i)$.
	For each last name $i\in M$, we sample a ground-truth race $r_i$ (unknown to the algorithms) according to the distribution $q_i$, and then sample the income $w_i \sim \cD_{r_i}$.

  We report the utility ratio ($\cK$) as a function of the risk difference ($\cF$) for different algorithms in Figure~\ref{fig:candidate_selection:1}.
	We also report $\cF$ as a function of $\alpha$ (for \FairExpec{}, \FairExpecGrp{}, and \Thresh{}) and as a function of $\lambda$ (for \MultObj{})
	in \suppMat~\ref{sec:extended_empirical_results:candidate_selection}.

	\subsubsection{Results}
	\FairExpec{} reaches the highest risk difference of $0.89$, followed by \FairExpecGrp{}, which reaches a risk difference of $0.84$.
	\Thresh{} reaches $\cF=0.79$, \MultObj{} reaches $\cF=0.70$, and \Blind{} has $\cF = 0.28$.
	We do not expect the algorithms to outperform the unconstrained utility (i.e., that of \Blind{}).
	We observe that \FairExpec{} has a better Pareto-tradeoff compared to other algorithms, i.e., for any desired level of risk difference, it has a better utility ratio ($\cK$) than other algorithms.
	In contrast, while \FairExpecGrp{} also has a high maximum risk difference, it is not Pareto-optimal.
	All algorithms lose a large fraction of the utility (up to 33\%).
	This is because the unconstrained and constrained optimal are very different:
	without any constraints, we would roughly select 7\% candidates from some races.
	However, ensuring equal representation requires selecting roughly 4 times as many candidates from these races.
	When the difference between unconstrained and constrained optimal is smaller, we expect to lose a smaller fraction of the utility.

	\subsection{Real-world data for image search}\label{sec:image_selection}
	In this simulation, we consider the problem of selecting images under noisy information about the gender of the person depicted in the image.
	We derive noisy information about the gender of the person depicted in an image using a CNN-based gender classifier and use this information to select a gender-balanced set of images.
	This could be relevant in mitigating gender bias in image search engines, which have been observed to over-represent the stereotypical gender in job-related searches~\cite{KayMM15, celis2019imagesummarization}.
	{Here, one could first select a balanced subset of images to display on each page, and then order this subset in decreasing order of utility from top to bottom of each page.}

	\subsubsection{Data}
	We use a recent image dataset named the Occupations Dataset by \cite{celis2019imagesummarization}.
	The dataset contains top 100 Google Image Search results (from December 2019) for 96 occupations related queries.
	For each image, the dataset has a gender (coded as men, women, or other), skin-tone (coded as dark or light), and the image's position in the search result (an integer in $[100]$).
	We present aggregate statistics from the dataset in Figure~\ref{fig:occupations_dataset_stats}.
	\begin{figure}[t!]
		\vspace{3mm}
		\begin{center}
			\begin{tabular}{|l|lll|l|}\hline
				& Male & Female & NA   & Total \\\hline
				Dark  & 568  & 318    & 106  & 992\\
				Light & 2635 & 1987   & 386  & 5008\\
				NA    & 198  & 119    & 3283 & 3600\\\hline
				Total    & 3401  & 2424    & 3775 & 9600 \\\hline
			\end{tabular}
		\end{center}
		\vspace{-5mm}
		\caption{ \small {\bf\em Statistics of the Occupations dataset~\cite{celis2019imagesummarization}.}
		\vspace{-3.0mm}
		}
		\label{fig:occupations_dataset_stats}
	\end{figure}

	\paragraph{Gender classifier.}
	We use an off-the-shelf face-detector~\cite{cnn_model} to extract faces of people from the images,
	and then use a CNN-based classifier~\cite{imdb_wiki_code} to predict the (supposed) gender of people from their faces.
	For each image $i$, the classifier outputs a prediction $f_i\in [0,1]$ (resp. $1-f_i$) which is the (uncalibrated) likelihood that the image is of a man (resp. women).\footnote{While there could be richer and nonbinary gender categories, many commercial classifiers, and the classifier by~\cite{imdb_wiki_code} categorizes images as either male or female.}
	{We calibrate this score as described next.}
  {As a preprocessing step, we remove all images which either have a gender label of NA and for which the face-detector did not detect any face.\footnote{Note that we do not check if the detected faces are correct. This introduces some error, which is also expected in practice.}
  Then, we calibrate outputs of the classifier ($f_i$) on all remaining images.
  This calibration is only done once and is used to compute the noise-information ($q_i$) for the entire simulation.
  For more details of the preprocessing used, we refer the reader to \suppMat{}~\ref{sec:implementation_details:image_selection}.}

	\paragraph{Selecting occupations.} %
	Next, we infer the occupations which have considerable gender stereotype.
	Towards this, we fix a threshold $\zeta\in [0,1]$ and partition the occupations into three sets:
	\begin{itemize}[leftmargin=10pt]
		\item $\stf(\zeta)$: occupations for which at least $\zeta$-fraction of images were labeled to appear to depict women,
		\item $\stf(\zeta)$: occupations for which at least $\zeta$-fraction of images were labeled to appear to depict men, and
		\item all other occupations.
	\end{itemize}
	\noindent We fix $\zeta=0.8$.
	This gives us $|\stf(\zeta)|=12$ and $|\stm(\zeta)|=29$, and 1,877 images with occupations in $\stf(\zeta)\cup \stm(\zeta)$ (for a list of the occupations {see Table~\ref{table:occupations} in \suppMat~\ref{sec:extended_empirical_results:image_selection}).}
	\begin{remark}\label{rem:image_selection:skewed_probability}
		Note that we calibrate $q$ on all occupations, and only consider images in a subset of occupations (less than half).
		This means that $q$s may not be an unbiased estimate of the protected attributes---which is a hard case for \FairExpec{}.
	\end{remark}

	\renewcommand{\folder}{./figures/image-subset-selection}
	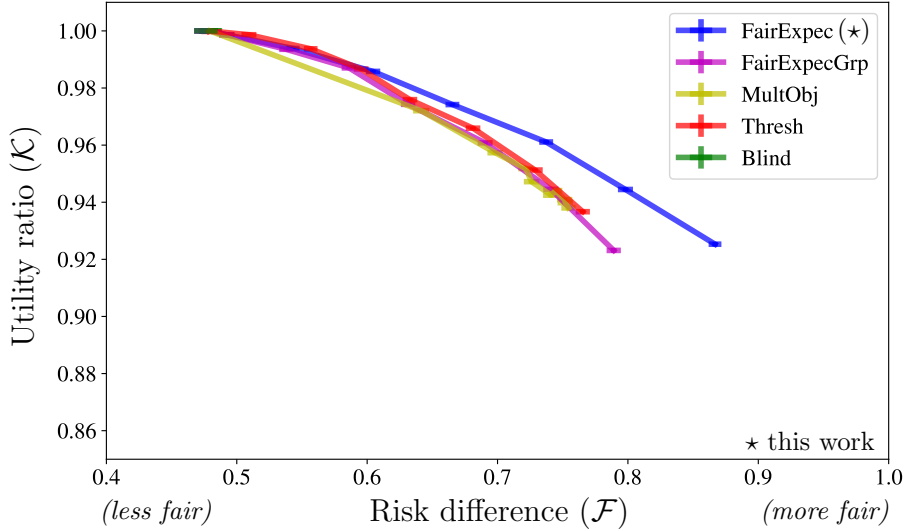
\begin{figure}[t!]
		\begin{center}
			\vspace{-3mm}
			\begin{tikzpicture}
				\node (image) at (-0.5,0) {\includegraphics[width=0.8\linewidth, trim={-0.1cm 0cm 0.5cm  1.7cm},clip]{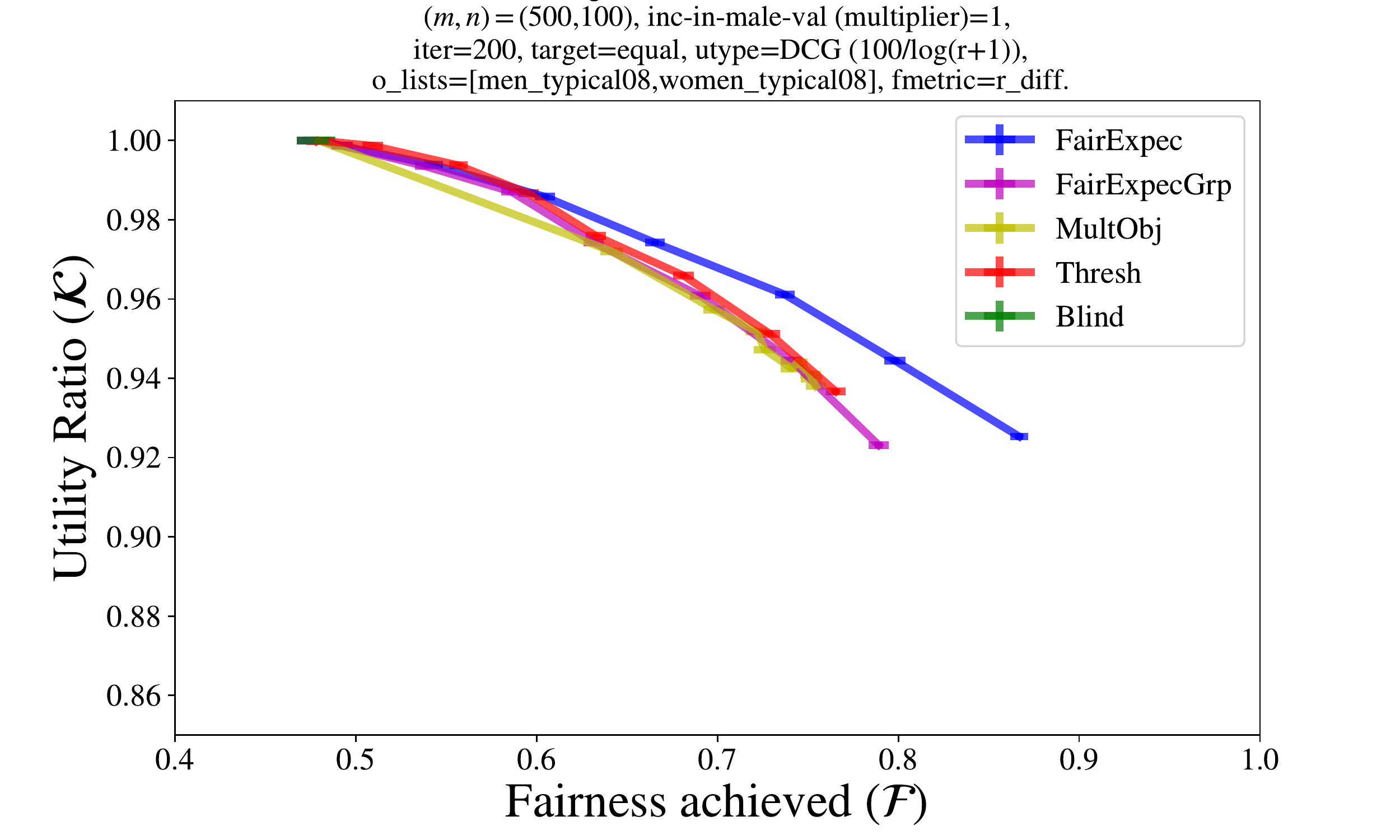}};
				\node[rotate=0, fill=white] at (-0.3*1.5,2.425*1.564) {\large \white{AAAAAAAAAAAAAAAAAAAAA}};
				\node[rotate=90, fill=white] at (-4.3*1.48,-0.5*1.6) {\large \white{AAAAAAAAAAAAAAAAAAAAA}};
				\node[rotate=90, fill=white] at (-4.3*1.5,0.25*1.6) {\large Utility ratio ($\cK$)};
				\node[rotate=0, fill=white] at (-0.3*1.5,-2.13*1.54) {\large \white{AAAAAAAAAAAAAAAAAAAAA}};
				\node[rotate=0] at (2.55*1.64,-2.13*1.54) {\textit{(more fair)}};
				\node[rotate=0] at (-0.3*0.5,-2.13*1.54) {\large Risk difference ($\cF$)};
				\node[rotate=0] at (-3.25*1.45,-2.13*1.54) {\textit{(less fair)}};
				\node[rotate=0] at (2.75*1.66, 1.99*1.57) {{$(\star)$}};
				\node[rotate=0] at (2.40*1.66,-1.45*1.6) {{$\star$ this work}};
			\end{tikzpicture}
		\end{center}
		\vspace{-4mm}
		\caption{\small \textbf{\em Real-world data for image search  (\cref{sec:image_selection})}:
		{This simulation considers gender as the protected attribute and uses a CNN-based classifier to derive noisy information about the gender of the person depicted in the image.
    The target is to ensure equal representation equal between genders.}
		The $y$-axis shows the utility ratio $\cK$ of different algorithms, and the $x$-axis shows the risk difference of different algorithms; $\cK$ and $\cF$ values are averaged over 200 trials, and the error bars represent the standard error of the mean.
		We observe that \FairExpec{} reaches the highest risk difference ($\cF\negsp =\negsp 0.86$), and has a better tradeoff between utility and fairness compared to other algorithms.
		}
		\vspace{-3mm}
		\label{fig:image_selection}
	\end{figure}

	\subsubsection{Setup}
	In this simulation, we consider gender as the protected attribute with two values ($p=2$), and fix $m=500$ and $n=100$.

	{We say a particular gender is {\em stereotypical} for a given occupation if the majority of images of this occupation are labeled to appear to depict a person with this gender}. For example, men are stereotypical for occupations in $\stm(0.8)$ and women are stereotypical for occupations in $\stf(0.8)$.
  {We call an image stereotypical if the dataset labels the person depicted in the image to appear to be of the stereotypical gender for its occupation.}
	We call an image {\em anti-stereotypical} if it is not stereotypical.
	We would like to ensure equal representation between {stereotypical and anti-stereotypical images.} %

	{For each choice of $\alpha$ and $\lambda$,} we draw a subset $M$ of $m$ images uniformly from all images with occupation in $\stf(\zeta)\cup \stm(\zeta)$.
	For each image $i\in M$, let its rank be $r_i\in [100]$.
	We compute $q_i$ as discussed earlier, and set its utility $w_i$ to
	$w_i\coloneqq (\log{(1+r_i)})^{-1}$.

	We report the utility ratio ($\cK$) as a function of the risk difference ($\cF$) for different algorithms in Figure~\ref{fig:image_selection}.
	We also report $\cF$ as a function of $\alpha$ (for \FairExpec{}, \FairExpecGrp{}, and \Thresh{}) and as a function of $\lambda$ (for \MultObj{})
	in \suppMat~\ref{sec:extended_empirical_results:image_selection}.
	\begin{remark}
		Since the underlying application in this simulation is ranking image results, we also considered Normalized DCG~\cite{DCG} as the utility metric.
    We observed similar results for this.
		For completeness, we present the plot in \suppMat~\ref{sec:extended_empirical_results:image_selection}.
		Furthermore, we also tried other functions for utilities $w_i$, including $100-r_i$ and $\nfrac{100}{r_i}$, and observed similar results.
	\end{remark}
	\subsubsection{Results}
	The risk difference of \Blind{} (i.e., the risk difference without any interventions) is $\cF=0.48$.
	All algorithms reach a better risk difference than \Blind{}.
	Among them, \FairExpec{} reaches the highest risk difference of $\cF=0.86$. While the next best algorithm \FairExpecGrp{} has $\cF=0.79$.
	Since the algorithms satisfy fairness constraints, we do not expect them to have a higher utility than \Blind{}; hence, it is not surprising that $\cK\leq 1$.
	Among the algorithms we consider, we observe that \FairExpec{} has the Pareto-optimal tradeoff between utility and fairness: it has a higher utility ratio compared to other algorithms for a given value of risk difference.

	\subsection{Additional empirical results}\label{rem:other_metrics}
	We present additional empirical results with selection lift in \suppMat{}~\ref{sec:extended_empirical_results}.
	We observe that, indeed, \FairExpec{} is able to mitigate discrimination with respect to selection lift and reaches the most-fair selection lift in all experiments.
	Further, it has the Pareto-optimal tradeoff between utility and fairness compared to other algorithms in all but one case:
	in the simulation from Section~\ref{sec:candidate_selection}, while \FairExpec{} has a lower utility than \MultObj{} for some levels of selection lift.
	We believe this is because the constraint region induced by threshold of selection lift is ``different'' from the constraint region of \target{}.
	One could correct this, e.g., by using \cite[Theorem 3.1]{celis2019classification} to reduce the constraint from selection lift to that of multiple lower and upper bound constraints.\footnote{This reduction uses multiple lower bound and upper bound constraints to provably approximate the constraints of selection lift.}

	\begin{remark}[{\bf Risk difference on varying $\nfrac{n}{m}$}]\label{rem:varying_n}
		Different applications could require selecting different fractions of results from the set of available items.
		For example, an image search engine might select a small fraction of available results, whereas a job platform can select a larger fraction.
		Towards analyzing the robustness of our approach to the fraction of items selected, we fixed $m$ and varied $n$ in simulations.
		We find that on increasing $n$ (holding $m$ fixed) the difference in algorithms' fairness remains roughly the same.
		See Figure~\ref{fig:candidate_selection:varying_n} for a plot for the simulation in \cref{sec:candidate_selection}; we present the plots for simulations from \cref{sec:toy_simulation,sec:image_selection}~in \suppMat~\ref{sec:extended_empirical_results:varying_n}.
	\end{remark}
  \renewcommand{\folder}{./figures/candidate-subset-selection/varying-n}
	\begin{figure}[t!]
		\begin{center}
			\vspace{-3mm}
			\begin{tikzpicture}
				\node (image) at (-0.5,0) {\includegraphics[width=0.8\linewidth, trim={0.5cm 0cm 1.5cm  1.6cm},clip]{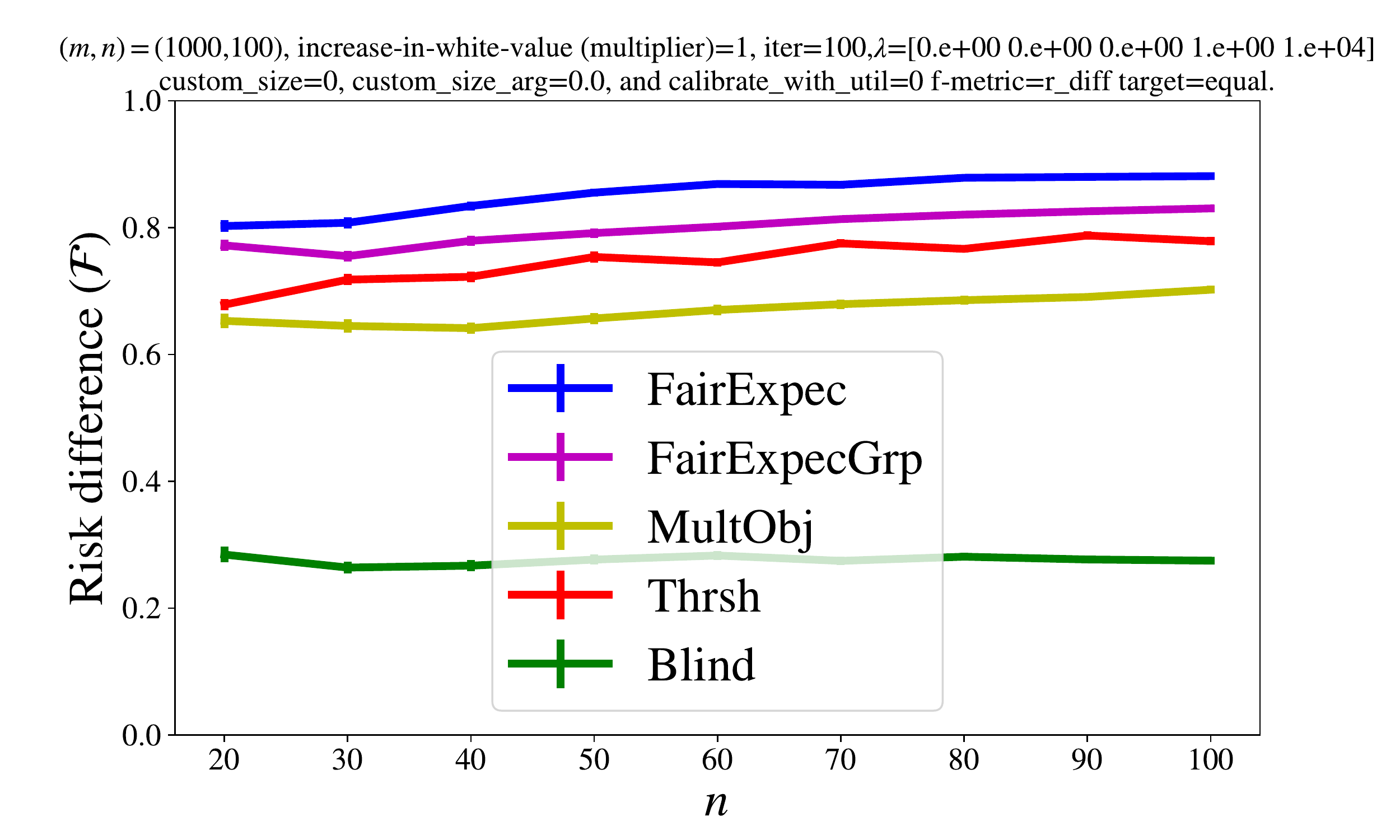}};
				\node[rotate=0, fill=white] at (0,2.60*1.55) {\large \white{AAAAAAAAAAAAAAAAAAAAAAAAAAAAAAAAAAAA}};
				\node[rotate=90, fill=white] at (-4.4*1.47,-0.5*1.6) {\large \white{AAAAAAAAAAAAAAAAAAAAA}};
				\node[rotate=90] at (-4.4*1.5,1.75*1.8) {\textit{(more fair)}};
				\node[rotate=90, fill=white] at (-4.4*1.5,0.30*1.6) {\large Risk difference ($\cF$)};
				\node[rotate=90] at (-4.4*1.5,-1.75*1.3) {\textit{(less fair)}};
				\node[rotate=0, fill=white] at (-0.3*1.5,-2.25*1.6) {\large \white{AAAAAAAAAAAAAAAAAAAAA}};
				\node[rotate=0] at (0,-2.25*1.6) {\large Selection size ($n$)};
				\node[rotate=0] at (0.9*1.8, 0.5*1.6) {\large{$(\star)$}};
				\node[rotate=0] at (2.65*1.65,-1.55*1.6) {{$\star$ this work}};
			\end{tikzpicture}
		\end{center}
		\vspace{-4mm}
		\caption{ \small {\bf\em Risk difference on varying $\nfrac{n}{m}$ (Remark~\ref{rem:varying_n}):}
		Towards analyzing the robustness of our approach to the fraction of items selected ($\nfrac{n}{m}$), we fix $m=1000$ and vary $n$ from $20$ to $100$ in Simulation~\ref{sec:candidate_selection}.
		The $y$-axis shows the risk difference $\cF$ of different algorithms, and the $x$-axis shows $n$; $\cF$ values are averaged over 100 trials, and the error bars represent the standard error of the mean.
		We find that increasing on $n$ the difference in the fairness of different algorithms remains roughly the same.
		}
		\vspace{-3mm}
		\label{fig:candidate_selection:varying_n}
	\end{figure}
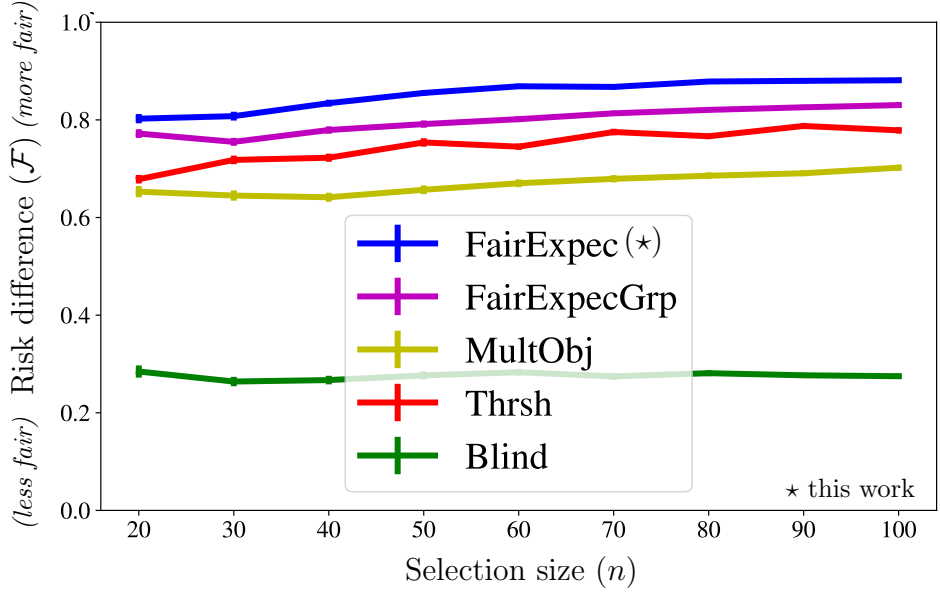

	\section{Proofs}\label{sec:proofs}
		In this section, we present the proof of \cref{thm:algorithm_for_target_fair} (\suppMat{}~\ref{sec:proof:thm:algorithm_for_target_fair}), and present formal statements of our hardness results for \denoised{} and their proofs (\suppMat{}~\ref{sec:proof:hardness_results_main}).

		\subsection{Proof of Theorem~\ref{thm:algorithm_for_target_fair}}\label{sec:proof:thm:algorithm_for_target_fair}
		\subsubsection{Proof of Lemma~\ref{thm:relation_between_target_and_noisy_gen}}\label{sec:proof:thm:relation_between_target_and_noisy}
		\begin{proof}[Proof of \cref{thm:relation_between_target_and_noisy_gen}]
			For all $i\in [m]$, define $Z_i\in \zo$ to be the indicator random variable that $(i\in G_\ell)$.
			Notice that
			\begin{align}
				\sum_{i\in G_\ell} x_i = \sum_{i\in [m]} x_i Z_i.\label{eq:counting_with_indicator_rv}
			\end{align}
			From \cref{def:individual_flipping_noise} we have $$\Ex[Z_i]=q_{i\ell}.$$
			Using linearity of expectation we get
			$$\sum_{i=1}^m x_i q_{i\ell} \stackrel{}{=} \sum_{i=1}^m \Ex[Z_i].$$
			Now, we have
			\begin{align*}
				\Pr\insquare{\sum_{i=1}^m x_i Z_i > \sum_{i=1}^m x_i q_{i\ell} + n\delta}&\stackrel{}{=} \Pr\insquare{\sum_{i=1}^m x_i Z_i > \Ex\sinsquare{\sum_{i=1}^m Z_i} +n\delta}\\
				&\leq \exp\inparen{-\frac{(\delta n)^2}{3} \frac{1}{\Ex\insquare{\sum\nolimits_{i=1}^m x_i Z_i}}}\tag{Additive Chernoff bound~\cite{motwani1995randomized}}\\
				&\leq \exp\inparen{-\frac{(\delta n)^2}{3} \frac{1}{\Ex\insquare{\sum\nolimits_{i=1}^m x_i}}}\tag{$\forall i\in [m],\ Z_i\leq 1$}\\
				&\stackrel{}{\leq} \exp\inparen{-\frac{\delta^2 n}{3}}.\hspace{-1mm}\tag{$\sum_{i=1}^m x_i = n$}
			\end{align*}
			By a similar argument, we have
			\begin{align*}
				\Pr\insquare{\sum\nolimits_{i=1}^m x_i Z_i < \sum_{i=1}^m x_i q_{i\ell} - n\delta}\leq \exp\inparen{-\frac{\delta^2 n}{2}}.
			\end{align*}
			We get the required result by taking a union bound over all $\ell\in [p]$ and using Equation~\eqref{eq:counting_with_indicator_rv}.
		\end{proof}
		\subsubsection{Proof of Lemma~\ref{thm:optimal_solution_with_few_fractional}}
		\label{sec:proof:thm:optimal_solution_with_few_fractional}
		\begin{proof}
			Assume $p<m$, otherwise the result is trivial.
			Recall that \relaxation{} has $m$ variables, $p$ pairs of upper and lower bound constraints, one cardinality constraint ($\sum_{i=1}^{m}x_i=n$), and $2m$ bounding-box constraints of the form: $(x_i\geq 0)$ and $(x_i\leq 1)$ for each $i\in [m]$.
			If for some $\ell\in [p]$, $L_{\ell} = U_{\ell}$, replace the two inequality constraints by the equivalent equality constraint.

			For any linear program, there exists a basic feasible solution~\cite[Theorem 2.6]{papadimitriou1998combinatorial}.
			If $x$ is a basic feasible solution, it must satisfy $m$ linearly-independent  equalities.
			Notice from each of the at most $p$ pairs of upper and lower bound constraints\footnote{If the lower bound and the upper bound are the same, we discard one of the (repeated) inequalities.} at most one of the upper bound or the lower bound can be satisfied with equality $x$.
			Thus, including the cardinality constraint ($\sum_{i=1}^{m}x_i=n$) (but, excluding the bounding-box constraints, $x$ satisfies at most $(p+1)$ constraints with equality.
			Further, note that if any of bounding-box constraints are satisfied with equality, then the corresponding index in $x$ is integral.
			This suffices to show that $x$ has at most $(p+1)$ fractional entries.

			Assume that exactly $(p+1)$ non-bounding-box constraints are satisfied with equality.
			By the previous discussion, one of them is the cardinality constraint, and the others must correspond to distinct values of $\ell\in [p]$.
			However, for each $i\in[m]$ it holds that $\sum_{\ell\in[p]}q_{i\ell}=1$.
			Further, in this case, the upper bounds or lower bound corresponding to the $p$ tight constraints must sum to $n$ (otherwise, the $(p+1)$ non-bounding-box constraints cannot hold together).
			But, in this case, all the $p$ upper or lower bound constraints imply the cardinality constraint.
			Thus, at most $p$ of the $(p+1)$ non-bounding-box constraints are linearly independent.
			It follows that any basic feasible solution has at most $p$ fractional constraints.
		\end{proof}
		\noindent The next fact shows that  \cref{thm:optimal_solution_with_few_fractional} is tight.
		\begin{fact}\label{fact:instance_with_p_fractional_entries}
			There exists an instance of \denoised{} such that the {\em unique} optimal solution to \relaxation{} has $p$ fractional entries.
		\end{fact}
		\begin{proof}
			Fix $\eps>0$ and $s=1$.
			Let $p\coloneqq p_1$ and drop the superscripts on the variables.
			Let $n=p$, $m=p+1$, and for all $\ell\in[p]$, $L_\ell=0$, and $U_\ell=1$.
			Let $e_i$ be the $i$-th standard basis vector on $\R^p$.
			Define
			\begin{align*}
				w_i=\begin{cases}
				1&\text{if }i<m,\\
				2&\text{if }i=m
			\end{cases}\quad\text{and }
			q_i=\begin{cases}
			e_i&\text{if }i<m,\\
			\frac{1}{p}\mathbbm{1}_p&\text{if }i=m.
		\end{cases}
		\end{align*}
		Notice that $q_i\in \Delta^p$ for each $i\in [m]$.
		It is easy to see that any optimal solution $x$ has $x_m=1$.
		Further, $x$ must pick $n$ items ($\sum_{i=1}^{m}x_i=n$) and $\sum_{\ell\in [p]}U_\ell=n$.
		So, each upper bound constraint must be tight.
		This suffice to show the unique optimal is:
		\begin{align*}
			\xdn_i =\begin{cases}
			1-\frac{1}{p}&\text{if }1\leq i<m,\\
			1&\text{if }i=m.
		\end{cases}
		\end{align*}
		Note that $\xdn$ has $p$ fractional values.
		\end{proof}
		\subsubsection{Proof of Theorem~\ref{thm:algorithm_for_target_fair}}
		\begin{proof}[Proof of \cref{thm:algorithm_for_target_fair}]
			We claim that Algorithm~\ref{alg:algorithm_for_target_fair} satisfies the claims in the theorem.\\

			\paragraph{\bf Run time.}
			Step 1 finds a basic feasible solution for
			\relaxation{}.
			Since this is a linear program, one can do so in time polynomial in the bit-complexity of the input, say, using the ellipsoid method~(see, e.g.,\cite{grotschel2012geometric}).
			Step 2 is clearly linear time.
			Thus, the bound on the running time follows.\\

			\paragraph{\bf Correctness.}
			Since $x$ (in Step 1) is feasible for \relaxation{}, it satisfies
			\begin{align*}
				\text{For all $\ell\in [p]$,}\qquad & L_\ell - \delta n\leq \sum\nolimits_{i=1}^{m} q_{i\ell}x_{i} \leq U_{\ell}+\delta n,\\
				&\sum\nolimits_{i=1}^{m}x_i = n.
			\end{align*}
			Further, since $x$ is a basic feasible solution of \relaxation{}, $x$ has at most $p$ fractional values by \cref{thm:optimal_solution_with_few_fractional}.
			Thus, $x^\prime$ obtained by rounding up each nonzero coordinate to 1  (in Step 2), satisfies (Recall that $q_{i\ell}\in [0,1]$)
			\begin{align*}
				\text{For all $\ell\in [p]$,}\qquad & L_\ell - \delta n\leq \sum\nolimits_{i=1}^{m} q_{i\ell}x_{i}^\prime \leq U_{\ell}+\delta n+p,\yesnum\label{eq:eq1_inmainproof}\\
				&n \leq \sum\nolimits_{i=1}^{m}x_i^\prime \leq n+p.\yesnum\label{eq:eq2_inmainproof}
			\end{align*}
			\cref{thm:relation_between_target_and_noisy} says that with probability at least $1-2p\exp\inparen{\nfrac{-\delta^2 n}{3}}$ (over the noise in the protected attributes)
			\begin{align*}
				\text{For all $\ell\in [p]$,}\quad \abs{\sum\nolimits_{i\in G_\ell} x^\prime_i - \sum\nolimits_{i=1}^{m} q_{i\ell} x^\prime_i } \leq n\delta.
				\yesnum\label{eq:eq3_inmainproof}
			\end{align*}
			Now, combining Equation~\eqref{eq:eq1_inmainproof} and Equation~\eqref{eq:eq3_inmainproof}, we get that
			\begin{align*}
				\text{For all $\ell\in [p]$,}\qquad & L_\ell - 2\delta n\leq \sum\nolimits_{i=1}^{m} q_{i\ell}x_{i}^\prime \leq U_{\ell}+2\delta n+p \yesnum\label{eq:eq4_inmainproof}
			\end{align*}
			holds with probability at least $1-2p\exp\inparen{\nfrac{-\delta^2 n}{3}}$ (over the noise in the protected attributes).
			Equation~\eqref{eq:eq4_inmainproof} implies that $x^\prime$ violates Equation~\eqref{eq:target_fair:fairness_constraint} by at most $(p+2\delta n)$ (additive)
			and Equation~\eqref{eq:eq2_inmainproof} implies $x^\prime$ violates Equation~\eqref{eq:target_fair:cardinality_constraint} by at most $p$ (additive).
			Let $\cE_1$ be the event that Equation~\eqref{eq:eq4_inmainproof} and Equation~\eqref{eq:eq2_inmainproof} hold.
			From the above discussion, we conclude that $$\Pr[\cE_1]\geq 1-2p\exp\inparen{\nfrac{-\delta^2 n}{3}}.$$

			\noindent Let \xtr{} be an optimal solution of \target{}. It remains to show that $x^\prime$ has a value higher than \xtr{}.
			Let $\cE_2$ be the event that \xtr{} is feasible for \target{}.
			Then, conditioning on $\cE_2$ we that
			\begin{align*}
				\sum\nolimits_{i=1}^{m} w_i \xtr_i &\leq \sum\nolimits_{i=1}^{m} w_i x_i  \tag{$x$ is optimal for \target{}}\\
				&\leq \sum\nolimits_{i=1}^{m} w_i x^{\prime}_i. \tag{$x^\prime\geq x$ and $w\geq 0$}
			\end{align*}
			Thus, conditioning on $\cE_2$ gives us that $x^\prime$ has a value higher than \xtr{}.
			From \cref{thm:relation_between_target_and_noisy} we get that $$\Pr[\cE_2]\geq 1-2p\exp\inparen{\nfrac{-\delta^2 n}{3}}.$$

			\noindent Taking the union bound over $\cE_1$ and $\cE_2$ we get the required result.
		\end{proof}

		\subsection{Hardness results}\label{sec:proof:hardness_results_main}
		In this section, we present the formal statements of our hardness results for \denoised{} and their proofs.

		\subsubsection{Proof of Theorem~\ref{thm:hardness_results_main} (1)}\label{sec:proof:thm:nphard_feasibility_ccsp}
		\begin{theorem}\label{thm:feasibility_is_hard}
			For all $p\geq 2$, deciding if \denoised{} is feasible is \np-hard.
		\end{theorem}
		\begin{proof}
			We present a reduction from the subset-sum problem which is known to be \np-hard~\cite{cormen2009introduction}.\\

			\begin{mdframed}[style=FrameBox2]
				\noindent {\em Subset sum problem.}
				Given a set of $m$ integers $A=\{a_1,\dots,a_n\} \in \Z_{\geq 0}^n$ in binary, and a target sum $t\in \Z_{\geq 0}$.
				The subset-sum problem is: Is there a subset of $A$ such that the items of the subset sum to $t$?
		\end{mdframed}\white{.}\\[-5.0mm]

			\noindent Without loss of generality assume that $a_i\leq t$ for all $i\in [n]$ and $t>0$.
			If not, we can remove items which are larger than $t$. Since $a_i\geq 0$, we can solve the subset-problem for $t=0$ in linear-time.

			\paragraph{Reduction.}
			Given an instance of the subset sum problem corresponding instance \target{} as follows:
			Let the number of items to select be $n$ and let there be $m \coloneqq 2n$ items.
			Consider one attribute (i.e., $s\coloneqq 1$) with two disjoint values $p\coloneqq 2$.
			For each item $i\in [m]$, let %
			\begin{align*}
				q_{i} = \begin{cases}
				\insquare{\frac{a_i}{t}, 1-\frac{a_i}{t}} & \text{if } i\in [n], \\
				\insquare{0, 1} & \text{if } i\in [n+1,m].
			\end{cases}
		\end{align*}
		Define the constraints to be $U_{1}=1$ and $U_{2}=n-1$.
		In this construction, for each $i\in[n]$, the integer $a_i\in A$ corresponds to the $i$-th item.
		Notice that the input to \target{} is polynomial-sized in the input of the subset-sum problem.
		Further, the values $q_i$, and constraints $U$ can be calculated in polynomial time.

		It remains to show that $A$ has a subset $S\subseteq A$ which sums to $U$ if and only if \target{} is feasible.\\

		\noindent {\em ($\implies$).}
		If $A$ has a subset $S\subseteq A$ such that $\sum_{a\in S} a = t$,
		define the solution $x$ of \target{} by choosing $|S|\leq n$ items corresponding to integers in $S$, and any $n-|S|\leq n$ items from $[m]\backslash [n]$.
		It holds
		\begin{align*}
			\hspace{-10mm} &\sum_{i\in [m]}x_{i}q_{i1}
			= \inparen{\sum_{i=1}^{n}x_{i}q_{i1} +\sum_{i\in [m]\backslash[n]}x_{i}q_{i1}}
			= \inparen{\sum_{a\in S}\frac{a}{t}+ \sum_{i=1}^{n-|S|}0}
			= 1 \
			\leq\  U_{1},
		\end{align*}
		\begin{align*}
			\sum_{i\in [m]}x_{i}q_{i2} &=
			\inparen{\sum_{i=1}^{n}x_{i}q_{i2} +\sum_{i\in [m]\backslash[n]}x_{i}q_{i2}}
			= \inparen{\sum_{a\in S}\bigg(1-\frac{a}{t}\bigg)+ \sum_{i=1}^{n-|S|}1}
			= n-1
			\ \leq\ U_{2}.
		\end{align*}
		Thus, \target{} is feasible if the subset-sum instance is a `YES' instance.\\

		\noindent{\em ($\Longleftarrow$).}
		For the opposite direction, notice that for any subset $x$, $\sum_{i\in [m]}x_{i}(q_{i1}+q_{i2})=n$, i.e., the sum of $q_{i\ell}$ for all items and both properties is fixed.
		Since $U_{1}+U_{2}=n$, the constraints of both properties must hold with equality:
		\begin{align}
			\sum_{i\in [m]}x_{i}q_{i1}= 1 \text{ and } \sum_{i\in [m]}x_{i}q_{i2}= n-1.\label{eq:111}
		\end{align}
		Define a solution $S\subseteq A$ which has items of $A$ corresponding to items $i\in [n]$ present in $x$.
		Since $q_{i1}=0$ for any $i\in [m]\backslash[n]$, it follows from \eqref{eq:111} that $\sum_{a\in S}\nfrac{a}{t}=1$ or $\sum_{a\in S}a=t$.
		Thus, the subset sum instance is a `YES' instance if \target{} is feasible.

		Combining both cases gives us that the \target{} instance is feasible iff the subset sum instance is a `YES' instance.
		Thus, we can use an algorithm to check the feasibility of \target{} to solve an arbitrary subset sum problem.
		\end{proof}
		\subsubsection{Proof of Theorem~\ref{thm:hardness_results_main} (2)}\label{sec:proof:thm:nphard_approx_ccsp}
		\begin{theorem}\label{thm:problem_is_apx_hard}
			For all $p\geq 3$, \denoised{} is \apx-hard.
		\end{theorem}
		\begin{proof}
			We present a reduction from the $d$-dimensional knapsack problem which is known to be \apx-hard for $d\geq 2$~\cite{kellerer2004knapsack}.\\

			\begin{mdframed}[style=FrameBox2]
				\noindent{\em $d$-dimensional knapsack problem (\dkp{}).}
				Given a vectors $v,c\in \R_{> 0}^k$, a matrix $w\in \R_{\geq 0}^{k\times d}$ the $d$-dimensional knapsack (\dkp{}) is the problem to solve the following integer linear program.
				\begin{align}\label{eq:ILP-d-kp}
					\max_{x\in \zo^k}\quad &\sum_{i\in [k]}v_{i}x_{i}\\
					\st\quad \forall\ \ell \in [d],\quad &\sum_{i\in [k]} w_{i\ell} x_{i} \leq c_{\ell}, \tagnum{Capacity}\\
					&x \in  \{0,1\}^k. \tagnum{Integrality}
				\end{align}
		\end{mdframed}\white{.}\\[-5.0mm]

			\noindent Without loss of generality we assume $w_{i\ell}\leq c_\ell$ for all $i\in [k]$ and $\ell\in[d]$.
			If not, we can remove such items in linear time.

			\paragraph{Reduction.}
			Given an instance of \dkp{} define an instance of \target{} with $s\coloneqq 1$ as follows:
			\begin{enumerate}[leftmargin=14pt] %
				\item Set $p\coloneqq d+1$, $n\coloneqq k$, $m\coloneqq 2k$, and $U_{p}\coloneqq m$. For all $\ell\in [p-1]$ set
				\begin{align*}
					U_\ell &\coloneqq \frac{c_\ell}{\sum_{k\in[d]} c_k}.
				\end{align*}
				\item For each item $i\in[k]$ set $w_i\coloneqq v_i$, and for all $\ell\in [p-1]$ set
				$$q_{i\ell}\coloneqq \frac{w_{i\ell}}{\sum_{j\in[d]}c_j}.$$
				\noindent Finally set, $q_{ip}=1-\sum_{i\in[p-1]}q_{i\ell}$.
				\item Add $k$ additional {\em dummy} items $(k+1,k+2,\dots, 2k)$. For all $i\in [m]\backslash[n]$, set $q_{ip}=1$ and $w_i\coloneqq 0$.
				Also, for all $i\in [m]\backslash[n]$ and $\ell\in [p-1]$ set $q_{i\ell}=0$.
			\end{enumerate}

			\smallskip

			\noindent Recall that $\Delta^p$ is the $p$-dimensional simplex.
			Note that for all $i\in [m]$, $q_{i}\in \Delta^p$.
			Now, it is easy to see that the construction is a valid instance of the \target{}.

			Let $x^\star\in \zo^k$ be a solution \dkp{}, such that, $\sum_{i=1}^{m}x^\star=r\leq l$.
			Define a solution $y^\star\in \zo^m$ to \target{} by selecting any $(k-r)$ dummy items.
			Alternatively, given a solution given a $y^\star\in \zo^m$ to \target{} define a solution $x^\star\in \zo^k$ to \dkp{}, by omitting the dummy items.\\

			\noindent {\em ($\iff$).}
			Consider a solution $x^\star$ to \dkp{}.
			Then for all $\ell\in [d]$, $x^\star$ satisfies
			\begin{align*}
				\sum_{i\in [k]} w_{i\ell} x_{i}^\star \leq c_{\ell}
				&\iff \sum_{i\in [k]} \inparen{\frac{w_{i\ell}}{\sum_{j\in[d]}c_j}} x_{i}^\star \leq \inparen{\frac{c_{\ell}}{\sum_{j\in[d]}c_j}}\\
				&\iff \sum_{i\in [k]} q_{i\ell} x_{i}^\star \leq U_{\ell}\\
				&\iff \sum_{i\in [m]} q_{i\ell} y_{i}^\star \leq U_{\ell}.  \tag{Using $q_{i\ell}=0$ for $i\leq [m]\backslash [n], \ell\in [p-1]$}
			\end{align*}
			Further, by construction $\sum_{i=1}^{m} y_{i}^\star=n$.
			Thus $x^\star$ is feasible for \dkp{} iff $y^\star$ is feasible for the corresponding instance of \target{}.
			Finally, it holds
			\begin{align*}
				\sum_{i\in [k]}v_{i}x_{i}^\star &= \sum_{i\in [k]}w_{i}x_{i}^\star
				=\sum_{i\in [m]} w_{i}y_{i}^\star. \tag{Using $w_{i}=0$ for $i\in [m]\backslash [n]$}
			\end{align*}
			This shows that the reduction is approximation preserving.
			We conclude that \target{} is \apx-hard for all $p=(d+1)\geq 3$.

		\end{proof}
		\subsubsection{Proof of Theorem~\ref{thm:hardness_results_main} (3)}\label{sec:proof:thm:hardness_for_p_poly_m}
		\begin{theorem}\label{thm:additive_hardness_result}
			For every constant $c > 0$, the following {\em violation gap} variant of \denoised{} is \np-hard.
			\begin{itemize}[]
				\item Output YES if the input instance is satisfiable.
				\item Output NO if there is no solution which violates every upper bound constraint at most an additive factor of $c$.
			\end{itemize}
		\end{theorem}
		\noindent Our reduction below is similar to the one used in \cite[Theorem 10.4]{celis2018ranking}.
		\begin{proof}
			We use the inapproximability of the independent set \cite{hastad1996clique, zuckerman2006linear} to prove the theorem.
			The inapproximability result states that:
			Given a graph $G(V,E)$ it is \np-hard to approximate the size of the maximum independent set to within a multiplicative factor of $|V|^{1-\eps}$ for every constant $\eps>0$.

			Fix any constant $c>0$.
			Then our reduction is as follows:

			\noindent {\em Reduction.} Consider an instance of the independent set problem, i.e., given a graph $G(V,E)$ and a number $n\in \N$ to check whether $G$ has an independent set of size at least $n$.
			Construct an instance of \denoised{} with $m=|V|$ candidates and the same $n$.
			For each clique $C$ in $G$ of size at least $(c+2)$  add a protected attribute $C$ which takes two values, and set an upper bound on the number of items with value 1 for attribute $C$: $U_{1}^{\sexp{k}}=1$.
			Formally, for each attribute $C$, set $L_{1}^{\sexp{C}}=0$ and $U_{1}^{\sexp{C}}=1$ and  $L_{0}^{\sexp{C}}=0$ and $U_{0}^{\sexp{C}}=n$.
			Further, for each vertex $v\in C$, set $q_{v, 1}^{\sexp{C}} = 1$ and $q_{v,0}^{\sexp{C}} = 0$.
			Note that there are at most $\binom{m}{c+2} = m^{c+2} = \poly(m)$ protected attributes.

			We claim the following:
			\begin{itemize}
				\item If the \denoised{} is not feasible then $G$ does not have an independent set of size $n$.
				\item If \denoised{} has a solution which violates the upper bound constraints by at most $c$ (additive), then $G$ has an independent set of size at least $n^{\nfrac{1}{(c+1)}}$.
			\end{itemize}
			These claims prove the theorem: if there was an algorithm $\mathcal{A}$ to solve \denoised{} in polynomial time, then we can use it to approximate independent set within $|V|^{1-\nfrac{1}{(c+1)}}$ factor in polynomial time,\footnote{The approximation algorithm returns the same answer as the $\mathcal{A}$.} which is \np-hard.

			It remains to establish the claim.
			Notice that if $G$ has an independent set $S$ of size $n$, then this gives a feasible solution to \denoised{} by picking the items corresponding to elements in $S$.
			This establishes the first claim.

			Given a subset $S$ which violates the constraints of \denoised{} by at most $c$ (additive), consider the subgraph $G_S$ of $G$ induced by $S$.
			$G_S$ does not contain any $(c+2)$-cliques.
			By \cite[Lemma 4.3]{chekuri2004multidimensional}, $G_S$ has an independent set of size at least $n^{\nfrac{1}{(c+1)}}$, so also $G$.
			This establishes the second claim.
		\end{proof}

	\section{Theoretical results with multiple protected attributes}\label{sec:extended_theoretical_results}
		In this section, we extend our theoretical results to multiple nonbinary (and intersectional) protected attributes (i.e., $s\geq 1$ and $p \geq 1$).
		Like \cref{sec:theoretical_results}, our main algorithmic result is an approximation algorithm for \target{}.
		\begin{theorem}[{\bf An approximation algorithm for \target{} when $s\geq 1$}]\label{thm:algorithm_for_target_fair_gen}
			There is an algorithm (Algorithm~\ref{alg:algorithm_for_target_fair_gen}) that given an instance of \target{} and $q$ from \cref{def:individual_flipping_noise},
			outputs a selection $x\in \zo^m$,
			such that, with probability at least $1-4ps\exp\inparen{\nfrac{-\delta^2 n}{3}}$ over the noise in the protected attributes of each item, the selection $x$
			\begin{enumerate}[leftmargin=15pt, label=\arabic*.]
				\item has a value at least as \ul{high} as the optimal value of \target{},
				\item violates the cardinality constraint~\eqref{eq:target_fair:cardinality_constraint} (additively) by at most
				$$1+\sum\nolimits_{k=1}^s (p_k-1),$$%
				\item and violates the fairness constraints~\eqref{eq:target_fair:fairness_constraint} (additively) by at most
				$$1+2\delta n+\sum\nolimits_{k=1}^s (p_k-1).$$
			\end{enumerate}
	    The algorithm runs in polynomial time in the bit complexity of the input.
		\end{theorem}
		\begin{remark}
			Note that in the special case that $s=1$, we have
	    \begin{align*}
	      1+\sum\nolimits_{k=1}^s (p_k-1) &=p_1,\\
	      1+2\delta n+\sum\nolimits_{k=1}^s (p_k-1) &=1+2\delta n+p_1.
	    \end{align*}
	    Substituting these values in \cref{thm:algorithm_for_target_fair_gen}, we can recover the statement of \cref{thm:algorithm_for_target_fair}.
		\end{remark}
		Compared to \cref{thm:algorithm_for_target_fair}, the high-probability bound and approximation factors are weaker by (roughly) a factor of $s$.
		Note that, in many real-world applications, $p$ and $s$ are small constants compared to the number of items selected $n$.
		Here, \cref{thm:algorithm_for_target_fair_gen} implies that $x$ violates the fairness constraints in Equation~\eqref{eq:target_fair:fairness_constraint} by a multiplicative factor of at most $(1+2\delta+o(1))$
		and the constraint in Equation~\eqref{eq:target_fair:cardinality_constraint} by a multiplicative factor of at most $(1+o(1))$ with high probability. %

		Algorithm~\ref{alg:algorithm_for_target_fair_gen} is equivalent to Algorithm~\ref{alg:algorithm_for_target_fair};
		the only difference is that Algorithm~\ref{alg:algorithm_for_target_fair_gen} solves the \denoised{} for $s\geq 1$.

		\begin{remark}[{\bf Analogous to \cref{rem:only_lower_bounds}}]
			We can strengthen \cref{thm:algorithm_for_target_fair_gen} to guarantee that  Algorithm~\ref{alg:algorithm_for_target_fair_gen} finds an $x\in\zo^m$ which does not violate the lower bound fairness constraint (left inequality in Equation~\eqref{eq:target_fair:fairness_constraint})
			and violates the upper bound fairness constraints by at most $(s\cdot(p-1)+2\delta n)$ (without changing other conditions).
			In particular, this shows that, if one places only lower bound fairness constraints, then subset found by Algorithm~\ref{alg:algorithm_for_target_fair_gen} would never violate the fairness constraints.
		\end{remark}

		\subsection*{Proof of Theorem~\ref{thm:algorithm_for_target_fair_gen}}
		The proof of \cref{thm:algorithm_for_target_fair_gen} follows a similar template as the proof of \cref{thm:algorithm_for_target_fair}.
		Instead of repeating the entire proof, we highlight how the proof of \cref{thm:algorithm_for_target_fair_gen} differs from the proof of \cref{thm:algorithm_for_target_fair}.

		We first present a generalization of \cref{thm:relation_between_target_and_noisy}.
		\begin{lemma}\label{thm:relation_between_target_and_noisy_gen}
			For all $\delta\negsp \in\negsp (0,1)$ and $x\negsp \in\negsp  [0,1]^m$, s.t., $\sum_{i=1}^m x_i\negsp =\negsp n$:
			\begin{align*}
				\text{$\forall\ k\in [s],\ \ell\in [p_k]$,}\quad \abs{\sum\nolimits_{i\in G_\ell^{\sexp{k}}} x_i - \sum\nolimits_{i\in [m]}q_{i\ell}^{\sexp{k}} x_i } \leq n\delta
			\end{align*}
			holds with probability at least $1-2ps\exp\inparen{\nfrac{-\delta^2 n}{3}}$ over the noise in the protected attributes of each item.
		\end{lemma}
		\begin{proof}
			The result follows by applying \cref{thm:relation_between_target_and_noisy} for each attribute $k\in [s]$ and taking the union bound over the events.
		\end{proof}
		\begin{remark}[{\bf Hardness results}]
			Note that the hardness results from \cref{thm:feasibility_is_hard} and \cref{thm:problem_is_apx_hard} also apply when $s\geq 1$.
			Like in \cref{sec:theoretical_results}, we overcome this hardness by allowing solutions to violate the constraints of \denoised{} by an additive amount.
		\end{remark}

		\noindent Consider the following linear-programming relaxation of \denoised{}
		\begin{align*}
			\max_{x\in [0,1]^m}\quad  &\sum\nolimits_{i=1}^{m}w_{i}x_{i}\tagnum{Relaxed-Denoised}\customlabel{prob:relaxed_denoised_fair_gen}{LP-Denoised}\\
			\st,\qquad & L_\ell^{\sexp{k}}\negsp - \delta n\leq \sum\nolimits_{i=1}^{m} q_{i\ell}^{\sexp{k}} x_{i} \leq U_{\ell}^{\sexp{k}}\negsp+\delta n, \ \  \forall \ k\in [s], \ell \in [p_k],\\
			&\sum\nolimits_{i=1}^{m}x_i = n.
		\end{align*}
		We can prove the following lemma.
		\begin{lemma}[{\bf An optimal solution with $p$ fractional entries}]\label{thm:optimal_solution_with_few_fractional_gen}
			Any basic feasible solution $x\in[0,1]^m$ of \ref{prob:relaxed_denoised_fair_gen} has at most
			$$ \min\inparen{m,1+\sum\nolimits_{k=1}^s (p_k-1)}$$
			fractional values, i.e.,
			$$\sum\nolimits_{i=1}^{m}\mathbbm{I}[x_i \in (0,1)]\leq \min\inparen{m,1+\sum\nolimits_{k=1}^s (p_k-1)}.$$
		\end{lemma}
		\begin{proof}
			In the proof of \cref{thm:optimal_solution_with_few_fractional}, we show that any basic feasible solution satisfies at most $p_1$ linearly-independent non-bounding-box constraints with equality.
			Toward this, we show that any set of $p_1$ linear-independent lower bound or upper bound constraints which are satisfied with equality, must contain the cardinality constraint ($\sum_{i=1}^m x_i=n$) in their span.

			\cref{thm:optimal_solution_with_few_fractional_gen} follows by using the same argument for all protected attributes and noticing that for each protected attribute $k\in [s]$, any set of $p_k$ linear-independent lower bound or upper bound constraints which are satisfied with equality, must contain the cardinality constraint ($\sum_{i=1}^m x_i=n$) in their span.
		\end{proof}

		\begin{remark}
			Note that in the special case that $s=1$, we have
			$$1+\sum\nolimits_{k=1}^s (p_k-1)=p_1,$$
			and we recover \cref{thm:optimal_solution_with_few_fractional} from \cref{thm:optimal_solution_with_few_fractional_gen}.
			Further, using \cref{fact:instance_with_p_fractional_entries}, it follows that \cref{thm:optimal_solution_with_few_fractional_gen} is tight.
		\end{remark}

		\begin{proof}[Proof of \cref{thm:algorithm_for_target_fair_gen}]
			The run time follows since there are polynomial time algorithms to find a basic feasible solution of a linear program.

			The rest of proof follows a similar template to that of \cref{thm:algorithm_for_target_fair}, and follows by using the generalizations of the lemmas proved above and replacing the additive error of $p_1$ (in $x^\prime$ due to rounding) with
			$$1+\sum\nolimits_{k=1}^s (p_k-1).$$

			\noindent Let \xtr{} be an optimal solution for \target{}.
			Using \cref{thm:relation_between_target_and_noisy_gen} it can be shown that \xtr{} is feasible for \denoised{} with probability at least $1-2ps\exp\inparen{\nfrac{-\delta^2 n}{3}}$.
			Assuming this event happens, it holds that the rounded solution $x^\prime$ (from Step 2 of Algorithm~\ref{alg:algorithm_for_target_fair_gen}) has a value at least as large as \xtr. (This follows from the argument as in the proof of \cref{thm:algorithm_for_target_fair}).

			Further, from \cref{thm:optimal_solution_with_few_fractional} it follows that $x^\prime$ picks at most
			$$1+\sum\nolimits_{k=1}^s (p_k-1)$$
			more elements than $x$.
			Thus, $x^\prime$ violates Equation~\eqref{eq:denoised_fair:cardinality_constraint}, and so Equation~\eqref{eq:target_fair:cardinality_constraint} by at most $1+\sum\nolimits_{k=1}^s (p_k-1)$ (additively).
			Further, $x^\prime$ violates the fairness constraints of \denoised{} by at most $1+\sum\nolimits_{k=1}^s (p_k-1)$ (additively).
			Combining this with \cref{thm:relation_between_target_and_noisy_gen}, we get that with probability at least $1-2ps\exp\inparen{\nfrac{-\delta^2 n}{3}}$,
			$x^\prime$ violates the constraints of \target{} by at most the bound claimed in \cref{thm:algorithm_for_target_fair_gen}.
			Finally, taking a union bound over above two events completes the proof.
		\end{proof}

		\vspace{-4mm}
		\setlength{\algomargin}{0.5em}
		\begin{algorithm}[h!]
			\AlgoDontDisplayBlockMarkers\SetAlgoNoEnd
			\caption{Algorithm for \target{}}
			\label{alg:algorithm_for_target_fair_gen}
			\kwInit{Numbers $n,s\in\N$, utility vector $w\in \R^m$, and for each $k\in [s]$: a probability matrix $q^{\sexp{k}}\in [0,1]^{m\times p_k}$ and constraint vectors $L^{\sexp{k}},\ U^{\sexp{k}}\in \R_{\geq 0}^{p_k}$.}\vspace{2mm}

			1. {\bf Solve} $x\gets$ {Find a basic feasible solution to linear-programming relaxation}

			\white{.}\hspace{20mm}  of \denoised{} with inputs $(n,s,q,w,L,U)$.

			2. {\bf Set } $x^\prime_i\coloneqq\ceil{x_i}\ $ for all $i\in [m]$.\commentalg{Round solution}

			3. {\bf Return} $x^\prime$.
		\end{algorithm}
		\vspace{-4mm}

\section{Limitations and future work}\label{sec:lim_future_work}
	We consider the natural setting where the utility of a subset is the sum of utilities of its items.
	A useful extension to this work could consider submodular objectives, which are relevant when the goal is to select a subset which summarizes a collection of items~\cite{Lin11}.

	Apart from this, some works also study other variants of subset selection, for example, diverse (and fair) subset selection in the online settings~\cite{StoyanovichYJ18}.
	Studying and mitigating bias in the presence of noise under this variant is an interesting direction for future work.

	Further, our approach assumes access to probabilistic information about the true protected attributes.
	If this information is itself is skewed or incorrect, then our approach can have a poor performance.
	Although, empirical results on real-world data suggest that our approach can improve fairness even when there is some skew in the noise information (see, e.g., \cref{rem:image_selection:skewed_probability}).
	Still, determining probabilistic information more reliably is an important problem, and
	{recent works have made some progress toward this goal~\cite{johnson2018multcalibration, jung2020multicalibration}.}

	Furthermore, while we focus on the subset selection problem, our results can also extend to the ranking problem (where after selecting a subset, it must be ordered) by satisfying the fairness constraints in the top-$k$ positions, for a small number choices of $k$, say $k_1\leq k_2\leq\dots\leq k_g$;\footnote{To do so: first, pick $k_1$ items and place them top-$k_1$ positions, then from those remaining pick $(k_2\negsp -\negsp k_1)$ items and {place them in the next $(k_2\negsp -\negsp k_1)$ positions, and so on.}}
  {this reduces the high-probability guarantee from $1-4p\exp{(\nfrac{-\delta^2 n}{3})}$ to $1-4gp\exp{(\nfrac{-\delta^2 k_1}{3})}$.\footnote{{This follows by using the union bound. However, as the events involved are correlated, it may be possible to get a stronger bound using a more sophisticated analysis.}}}
	This is particularly relevant in the setting where results are displayed one page at a time.
	Satisfying the constraints for a larger number of positions with high probability might require stronger information about noise, and is an interesting direction for future work.

	Empirically, we could report fairness on several other metrics, e.g., selection lift or extended difference~\cite{calders2010three, Hajian2014generalization, HajianPrivacy2015}.
	We focus on risk difference as it is closer to our approach. %
	Nevertheless, \target{} can mitigate discrimination with respect to selection lift (and other metrics) as well (e.g., see \cite{CKSDKV18, celis2019classification}).
	Empirically evaluating this would be an important direction for future work.

	{Finally, we note that bias is a systematic issue and this work only addresses one aspect of it.
	Indeed, any such approach is limited by how people and the broader system uses the subset presented to them; e.g., a recruiter on an online hiring platform might deliberately reject minority candidates even when presented with a representative candidate subset.
	Thus, it is important to complement our approach with other necessary tools to mitigate bias and counter discrimination.}

\section{Conclusion}
	We consider the problem of mitigating bias in subset selection when the protected attributes are noisy, or are missing for some or all of the entries and must be {imputed} using proxy variables.
	We note that accounting for real-world noise in algorithms is important to mitigate bias, and not accounting for noise can have unintended adverse effects, e.g., adding fairness constraints in a noise oblivious fashion can even decrease fairness when the protected attributes are noisy (Section~\ref{sec:toy_simulation}).

	We consider a model of noise where we have probabilistic information about the protected attributes, and develop a framework to mitigate bias, which given this information, can satisfy one from a large class of fairness constraints with at most a small multiplicative error with high probability (\cref{sec:theoretical_results}).
	In our empirical study, we observe that this approach achieves a high level of fairness on standard fairness metrics (e.g., risk difference), even when the probabilistic information about protected attributes is skewed (\cref{rem:image_selection:skewed_probability} and \cref{sec:lim_future_work}), and
	this approach has a better tradeoff between utility and fairness compared to several prior approaches (\cref{sec:empirical_results}).

\vspace{-0.5mm}
\section*{Acknowledgements}
	This research was supported in part by a J.P. Morgan Faculty Award.
	We would like to thank Nisheeth K. Vishnoi for several useful discussions on the problem and approach.
  \vspace{-2mm}

\newpage
\bibliographystyle{plain}
\bibliography{bib-v1.bib}

\appendix

\addtocontents{toc}{\protect\setcounter{tocdepth}{1}}

\newpage

\section{Extended empirical results}\label{sec:extended_empirical_results}
	\subsection{Implementation details}\label{sec:implementation_details}
	\subsubsection{Optimization libraries.}
		We used CVXPY~\cite{agrawal2018rewriting} and Gurobi~\cite{gurobi} to implement the algorithms.

	\subsubsection{Rounding scheme.}
		The rounding scheme from our theoretical results (Step 2 of Algorithm~\ref{alg:algorithm_for_target_fair}) crucially uses the fact that the intermediate solution found has a small number of fractional entries.
		Since the solutions of \MultObj{} can have a large number of fractional entries, the same rounding scheme is not useful.
		Instead, one can use randomized rounding~\cite{motwani1995randomized}, which, given a solution $x\in [0,1]^m$, outputs a subset $S\subseteq[m]$ of size $n$ with probability $\prod_{i\in S}x_i$.

		In our simulations, we use the same rounding scheme --- randomized rounding~\cite{motwani1995randomized} --- for \MultObj{}, \FairExpec{}, and \FairExpecGrp{}.
		We also ran a version of our experiments where, we used the rounding scheme from our theoretical results (Step 2 of Algorithm~\ref{alg:algorithm_for_target_fair}) for \FairExpec{}, and \FairExpecGrp{}, and randomized rounding for \MultObj{}.
		We observed similar results in these simulations.

	\subsubsection{Choice of $m$ and $n$.}
	We set $m=500$ and $n=100$ in all experiments,
	except in \cref{sec:causal_simulation} --- where we set $m=2000$ following \cite{yang2020causal} ---
	and in \cref{sec:candidate_selection} --- where increase $m$ to $1000$.
	We increase $m$ to $1000$ in \cref{sec:candidate_selection} to ensure that the sample of candidates has at least $\nfrac{n}{4}$ candidates from each of the four races.
	(Any algorithm requires this to be able to satisfy the fairness constraints.)

  \subsubsection{Additional details of the simulation from \cref{sec:causal_simulation}}\label{sec:implementation_details:causal_simulation}
    \paragraph{Additional details of the data.}
      {The dataset has 37\% Black candidates, 37\% candidates without prior experience, and 13.7\% Black candidates without prior experience.
    	The utilities are such that the Black candidates ($z_i=0$) have a lower expected utility than non-Black candidates, candidates without prior experience ($a_{i1}=0$) have a lower expected utility than those with prior work experience, and Black candidates without prior experience {($z_i=0,a_{i1}=0$) have the lowest expected utility.}}

    \paragraph{Preprocessing}
      {We sample a training dataset $D^\prime$ with $m=2000$ candidates.
    	Then, we use the code by \cite{yang2020causal} to get an approximation $\bar{\theta}$ of $\theta$ from $D^\prime$.
    	(\CntrFair{} and \CntrFairRes{} use $\cM(\bar{\theta})$ to generate counterfactual utilities).
    	To generate the estimate of $q$, we partition candidates into 20 equally sized bins, $(b_1,\dots,b_{20})$, by their utility. (Each bin has 100 candidates).
    	Given a candidate $i$ described by $(w_i, z_i, a_{i1}, a_{i2})$ such that $w_i\in b_k$, we define $q_i$ as follows:
    	\begin{align*}
    		q_{i0} = \Pr\nolimits_{j}[z_j=0\mid w_j\in b_k]\quad\text{ and }\quad  q_{i1}=1-q_{i0},\yesnum\label{eq:candidate_selection:q}
    	\end{align*}
    	where the probability is over the candidates $j$ in $D^\prime$.
    	Note that the candidate $i$ may not be in $D^\prime$.}

  \subsubsection{Implementation details of simulation from \cref{sec:image_selection}}\label{sec:implementation_details:image_selection}

  	\paragraph{Specifics of cropping.}
  	For each image where the face-detector found a face, we cropped the image to the region ``around'' the detected face.
  	More precisely, we expand the bounding box returned by the face-detector by 40\% (on each side) and then crop the image to this expanded bounding box.
  	If the expanded bounding box exceeded the image dimensions, we filled in the empty region by copying the last pixel layer.

  	\paragraph{Face-detector.} We referenced the tutorial by~\cite{cnn_tutorial} while writing the code for our simulations.

    \paragraph{Preprocessing}

    {As a preprocessing step, we remove all images with a gender label NA, this leaves us with 5,825 images.
  	Next, we use a face-detector to extract faces from the images  (5,825) and remove the ones where the detector does not detect any face.%
  	\footnote{Note that we do not check if the detected faces are correct. This introduces some error, which is also expected in practice.}
  	This gives us 4,494 images (77\% of 5,825).}
  	{We calibrate the classifier by binning its values into $b=20$ bins.
  	For each bin $j\in [20]$, we calculate the classifier's accuracy $a(j)\in [0,1]$.
  	This initial calibration is done only once.
  	Given image $i$, we compute the classifier's output $f_i\in[0,1]$.
  	Let $f_i$ fall in bin $j$.
  	Then, we set the noise-information $q_i$ of image $i$ to $q_i\coloneqq [a(j), 1-a(j)]$}
    More formally, we ran the image classifier on images to get a set of predictions $F\coloneqq \{f_i\}_i$.
  	Then, we binned the values in $F$ into $b=20$ equally sized bins (over $[0,1]$).
  	For all $j\in[b]$, let $B_j\subseteq [m]$ be the set of images in the $j$-th bin.
  	Further, let $G_{m}$ and $G_f$ be the set of images containing men and women, respectively.
  	Then, taking an uncalibrated $f_i$ as input, we output calibrated $q_{i, \ell}$ as follows: Let $f_i$ fall in the $j$-th bin, then for all $\ell\in \{m,f\}$ output
  	\begin{align*}
  		q_{i, \ell}\coloneqq \frac{\abs{B_j\cap G_\ell}}{\abs{B_j}}.\yesnum\label{eq:image_classifier_simulation:calibration}
  	\end{align*}

	\subsection{Optimizing for Selection Lift}
	Given subset $S\in [m]$ and target $t\in[0,1]^p$, define the selection lift $\cF_L(S,t)\in [0,1]$ as
	\begin{align}
		\cF_L(S,t)\coloneqq \min_{\ell,k\in [p]}\inparen{\frac{|S\cap G_\ell|}{n\cdot t_\ell}\cdot \frac{n\cdot t_k}{|S\cap G_k|}}.\tagnum{Selection lift}\customlabel{eq:selection_lift_def}{\theequation}
	\end{align}
	Note that when the target is proportional representation, then the above definition reduces to the usual definition of selection lift~\cite{calders2010three,Hajian2014generalization,HajianPrivacy2015}.
	Let $\cA(w,q)\subseteq [m]$ be the subset outputted by algorithm $\cA$ on input $(w,q)$.
	We report $$\cF_{L,\cA}\coloneqq \Ex\insquare{\cF_L(\cA(w,q),t)},$$ where the expectation is over the choices of $(w,q)$.
	We drop the subscript $\cA$ when the algorithm is not important or clear from context.

	\paragraph{Discussion.} Let selection lift of a selection $x\in \zo^m$, be $\cF_L(x)\in [0,1]$.
	Fix the desired level $\sigma\in [0,1]$ of selection lift, and let the set of selections $x\in \zo^m$ which have $\cF_L(x)\geq \sigma$ selection lift be $R(\sigma)$.
	Ideally, we would like to find
	\begin{align}\label{eq:ideal_problem}
		\argmax\nolimits_{x\in R(\sigma)} w^\top x.
	\end{align}
	{However, the feasible region $R^\prime(\sigma)$, considered by \FairExpec{} is a strict subset of $R(\sigma)$.
  Thus, \FairExpec{} outputs}
	\begin{align}\label{eq:lu_approximation_to_ideal_problem}
		\argmax\nolimits_{x\in R^\prime(\sigma)} w^\top x.
	\end{align}
  {We believe this is why \FairExpec{} is not Pareto-optimal in the simulation from Section~\ref{sec:candidate_selection} for $\cF_L(\cdot)$.}
	However, a reduction from Program~\eqref{eq:ideal_problem} to (multiple instances of) Program~\eqref{eq:lu_approximation_to_ideal_problem} should ensure that \FairExpec{} is Pareto-optimal for $\cF_L(\cdot)$ (see, e.g., \cite[Theorem 3.1]{celis2019classification}).
  \newpage
	\subsection{Additional simulations varying $\nfrac{n}{m}$}\label{sec:extended_empirical_results:varying_n}
	\bigskip
	\renewcommand{\folder}{./figures/image-subset-selection/varying-n}
	\begin{figure}[h!]
		\begin{center}
			\vspace{-3mm}
			\begin{tikzpicture}
				\node (image) at (-0.5,0) {\includegraphics[width=0.6\linewidth, trim={0.5cm 0cm 1.5cm  1.8cm},clip]{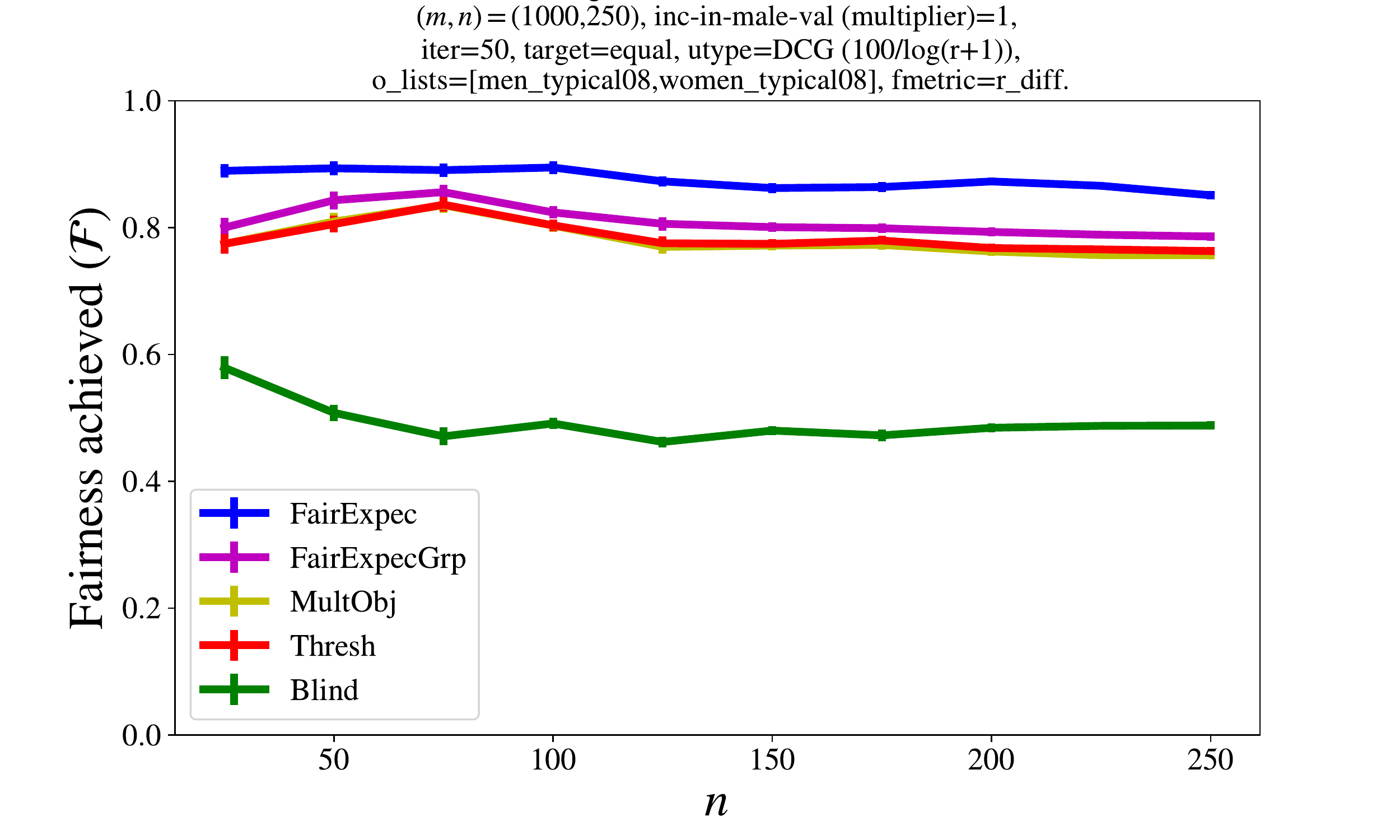}};
				\node[rotate=90, fill=white] at (-4.4*1.1400,-0.5*1.2) {\small\white{AAAAAAAAAAAAAAAAAAAAA}};
				\node[rotate=90] at (-4.4*1.1400,1.75*1.2) {\scriptsize\textit{(more fair)}};
				\node[rotate=90, fill=none] at (-4.4*1.1400,0*1.2) {\small Risk difference ($\cF$)};
				\node[rotate=90] at (-4.4*1.1400,-1.75*1.2) {\scriptsize\textit{(less fair)}};
				\node[rotate=0, fill=white] at (-0.3*1.1400,-2.25*1.2) {\small \white{AAAAAAAAAAAAAAAAAAAAA}};
				\node[rotate=0] at (-0.1*1.1400,-2.25*1.2) {\small Selection size ($n$)};
			\end{tikzpicture}
		\end{center}
		\vspace{-5mm}
		\caption{ \small {\bf\em Risk difference on varying $\nfrac{n}{m}$ (\cref{sec:image_selection}):}
		To analyze the robustness of our approach to the fraction of items selected ($\nfrac{n}{m}$), we fix $m=1000$ and vary $n$ from $25$ to $250$ in Simulation~\ref{sec:image_selection}.
		The $y$-axis shows the risk different $\cF$ of different algorithms, and the $x$-axis shows $n$; $\cF$ values are averaged over 200 trials, and the error bars represent the std. error of the mean.
		We find that on increasing $n$ the difference between the $\cF$ of different algorithms remains roughly the same.
		We refer the reader to Remark~\ref{rem:varying_n} for the details.
		}
	\end{figure}
	\bigskip
	\renewcommand{\folder}{./figures/disparate-error/varying-n}
	\begin{figure}[h!]
		\begin{center}
			\vspace{-3mm}
			\begin{tikzpicture}
				\node (image) at (-0.5,0) {\includegraphics[width=0.6\linewidth, trim={0.5cm 0cm 1.5cm  1.8cm},clip]{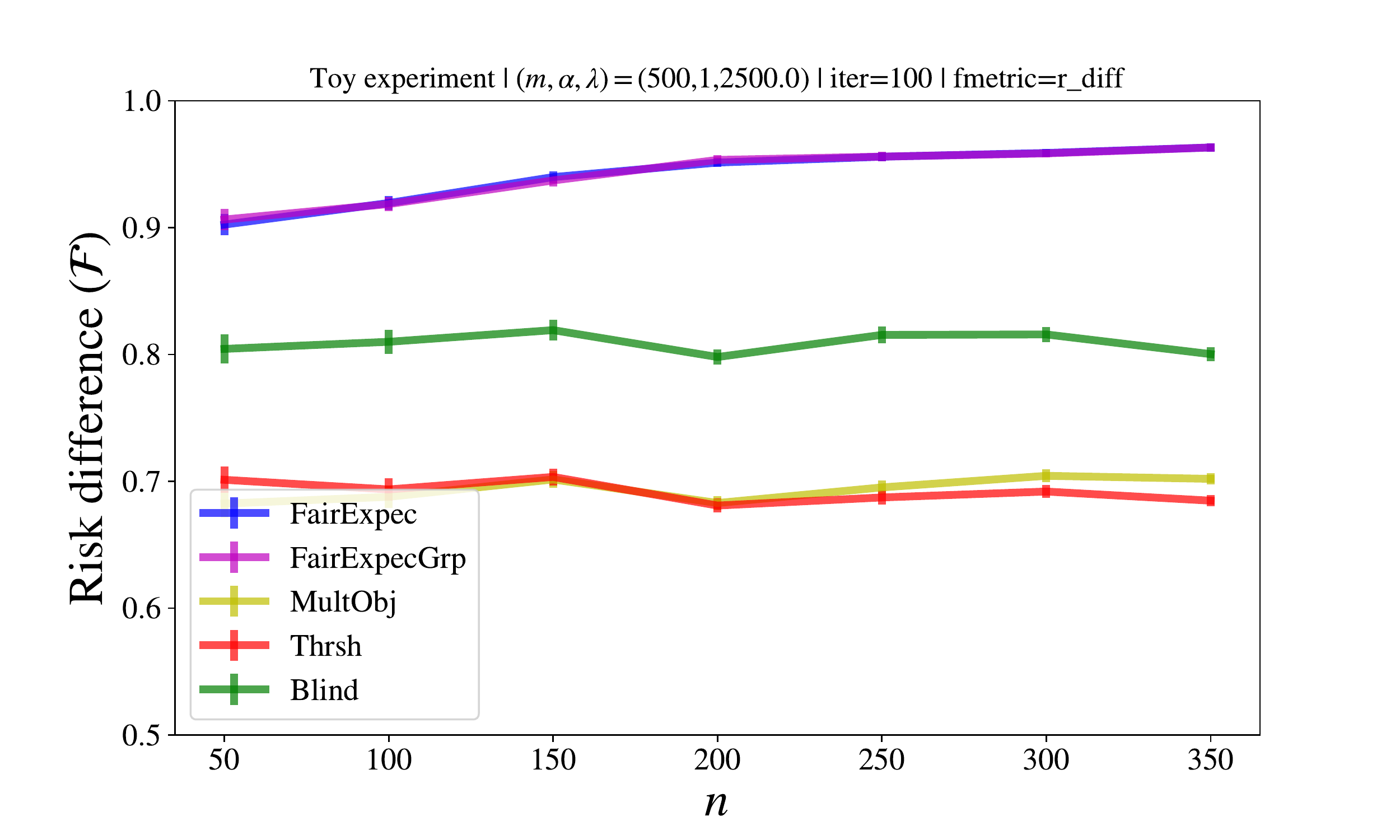}};
				\node[rotate=90, fill=white] at (-4.4*1.1400,-0.5*1.2) {\small\white{AAAAAAAAAAAAAAAAAAAAA}};
				\node[rotate=90] at (-4.4*1.1400,1.75*1.2) {\scriptsize\textit{(more fair)}};
				\node[rotate=90, fill=none] at (-4.4*1.1400,0*1.2) {\small Risk difference ($\cF$)};
				\node[rotate=90] at (-4.4*1.1400,-1.75*1.2) {\scriptsize\textit{(less fair)}};
				\node[rotate=0, fill=white] at (-0.3*1.1400,-2.25*1.2) {\small \white{AAAAAAAAAAAAAAAAAAAAA}};
				\node[rotate=0] at (-0.1*1.1400,-2.25*1.2) {\small Selection size ($n$)};
			\end{tikzpicture}
		\end{center}
		\vspace{-5mm}
		\caption{ \small {\bf\em Risk difference on varying $\nfrac{n}{m}$ (\cref{sec:toy_simulation}):}
		To analyze the robustness of our approach to the fraction of items selected ($\nfrac{n}{m}$), we fix $m\negsp \negsp =\negsp \negsp 1000$ and vary $n$ from $50$ to $350$ in Simulation~\ref{sec:toy_simulation}.
		{The $y$-axis shows the risk different $\cF$ of different algorithms, and the $x$-axis shows $n$; $\cF$ values are averaged over 100 trials, and the error bars represent the std. error of the mean.}
		We find that on increasing $n$, \FairExpec{} and \FairExpecGrp{} become slightly fairer compared to others.
		We refer the reader to Remark~\ref{rem:varying_n} for the details.
		}
		\vspace{-20mm}
		\label{fig:toy_simulation:varying_n}
	\end{figure}
	\newpage

	\subsection{Additional plots for Section~\ref{sec:toy_simulation}}\label{sec:extended_empirical_results:toy_simulation}
	\renewcommand{\folder}{./figures/disparate-error}
	\begin{figure}[h!]
		\vspace{-5.5mm}
		\begin{center}
		{\begin{tikzpicture}
		\node (image) at (0,0) {\includegraphics[width=0.6\linewidth, trim={1.7cm 0.3cm 0.5cm 0.26cm},clip]{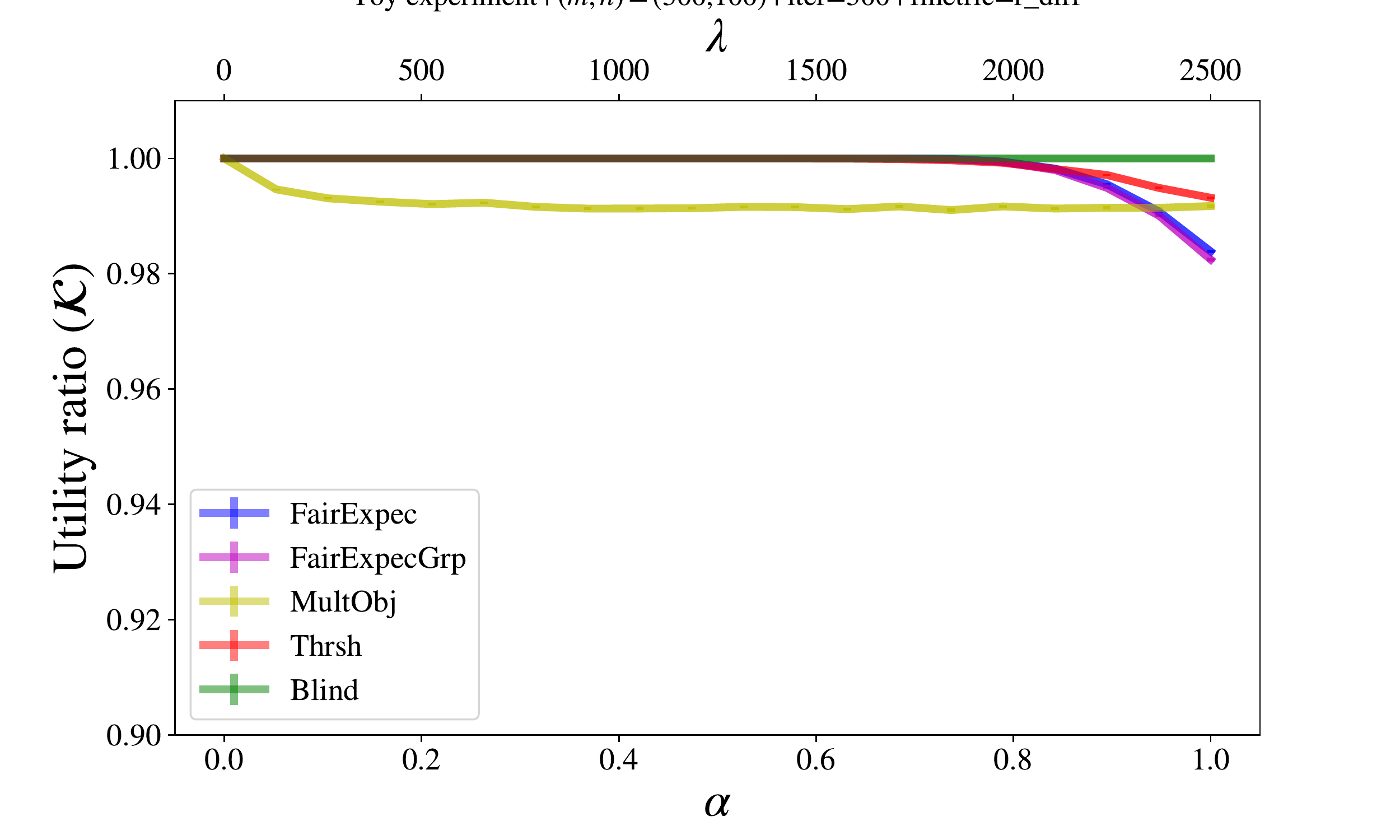}};
		\node[rotate=90, fill=white] at (-4.1*1.25,0) {\small \white{AAAAAAAAAAAAAAAA}};
		\node[rotate=90, fill=white] at (-4.1*1.25,0) {\small Utility Ratio ($\cK$)};
	\end{tikzpicture}}
	\end{center}
	\vspace{-4mm}
	\caption{\small
	{\em {Additional plot for \cref{sec:toy_simulation}:}}
  This simulation considers the setting where the minority group (40\% of total) has a higher 30\% higher FDR compared to the majority group.
  The utilities of all candidates are iid from the uniform distribution.
  The target is to ensure equal representation between the majority and minority groups.
	We refer the reader to \cref{sec:toy_simulation} for the details of the simulation.
	The $y$-axis shows the utility ratio $\cK$ of different algorithms, and the $x$-axis shows the constraint parameters ($\alpha$ for \FairExpec{}, \FairExpecGrp{}, and \Thresh{}, and $\lambda$ for \MultObj); $\cK$ averaged over 500 trials, and the error bars represent the standard error of the mean.
	We observe that \FairExpec{} and \FairExpecGrp{} lose a larger fraction of the utility.
	Although, note that this because they have a higher level of fairness; see Figure~\ref{fig:toy_simulation:1}.
	}
	\vspace{-4.0mm}
	\label{fig:toy_simulation:1_additional}
	\end{figure}

	\begin{remark}[{\bf Reading the double $x$-axis}]
		We note that both subfigures in Figure~\ref{fig:toy_simulation:1_additional} and Figure~\ref{fig:toy_simulation:selection_lift}  have a double $x$-axis: one for  $\alpha$ and one for $\lambda$.
		Therefore, one should note compared the $\cF$ or $\cK$ of \MultObj{} with other algorithms at a particular $x$-coordinate.
		It could be useful to consider the limiting their values at the largest $\alpha$ and $\lambda$, respectively.
	\end{remark}
	\renewcommand{\folder}{./figures/disparate-error}
	\begin{figure}[h!]
		\vspace{-5.5mm}
		\begin{center}
		{ \begin{tikzpicture}
		\node (image) at (0,0) {\includegraphics[width=0.6\linewidth, trim={1.7cm 0.3cm 0.5cm 0.26cm},clip]{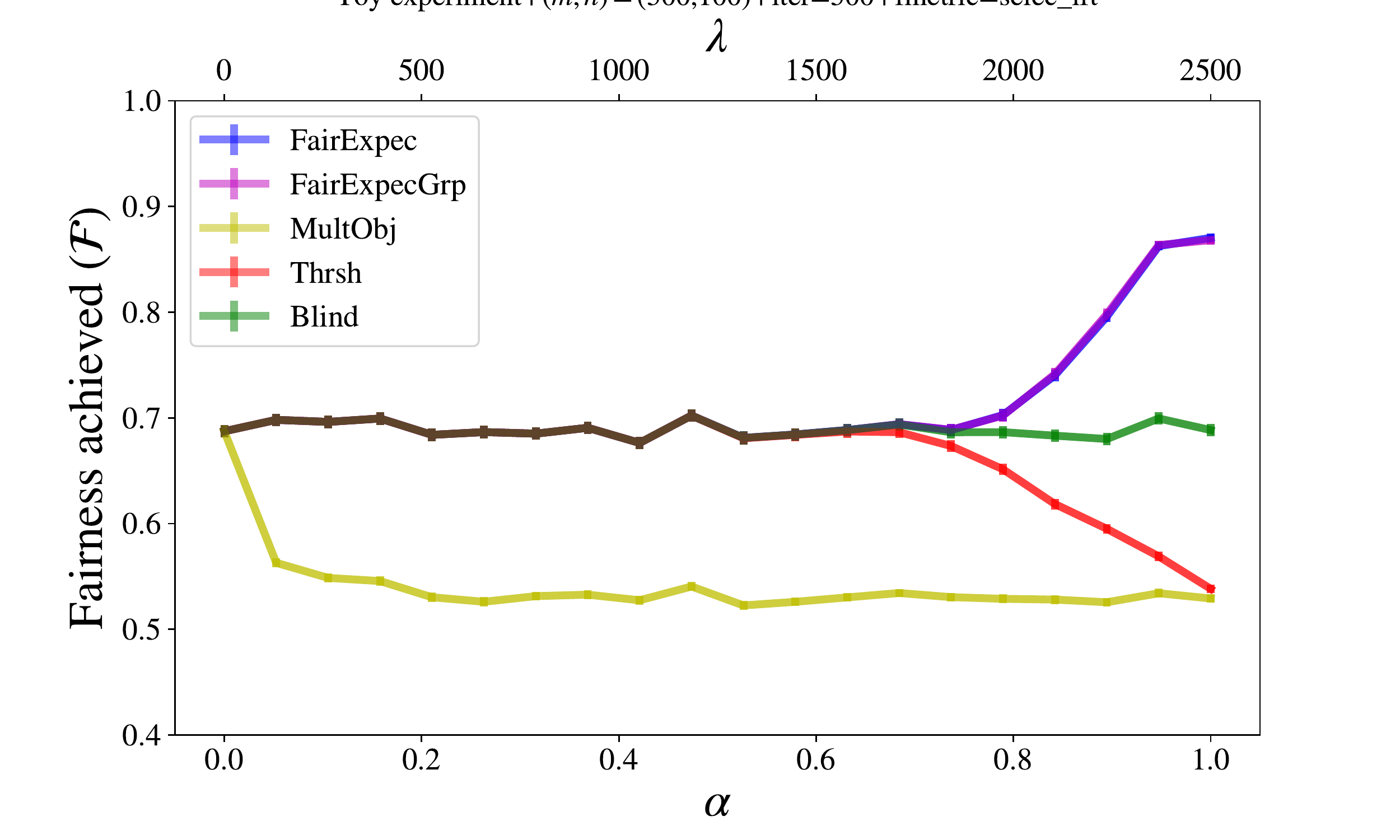}};
		\node[rotate=90, fill=white] at (-4.1*1.23,0) {\small \white{AAAAAAAAAAAAAAAA}};
		\node[rotate=90] at (-4.1*1.23,1.75*1.25) {\footnotesize\textit{(more fair)}};
		\node[rotate=90, fill=white] at (-4.1*1.23,0) {\small Selection lift ($\cF_L$)};
		\node[rotate=90] at (-4.1*1.23,-1.75*1.22) {\footnotesize\textit{(less fair)}};
	\end{tikzpicture}}
		\vspace{-3mm}
	\end{center}
	\vspace{-4mm}
	\caption{\small
	{\em {Additional plot for \cref{sec:toy_simulation} with selection lift:}}
  This simulation considers the setting where the minority group (40\% of total) has a higher 30\% higher FDR compared to the majority group.
  The utilities of all candidates are iid from the uniform distribution.
  The target is to ensure equal representation between the majority and minority groups.
	We refer the reader to \cref{sec:toy_simulation} for the details of the simulation.
	The $y$-axis shows the selection lift $\cF_L$ of different algorithms, and the $x$-axis shows the constraint parameters ($\alpha$ for \FairExpec{}, \FairExpecGrp{}, and \Thresh{}, and $\lambda$ for \MultObj); $\cF_L$ values are averaged over 500 trials, and the error bars represent the standard error of the mean.
	We observe similar results as with risk difference ($\cF$).
	For a definition of selection lift see Equation~\eqref{eq:selection_lift_def}.
	}
	\vspace{-10mm}
	\label{fig:toy_simulation:selection_lift}
	\end{figure}

	\newpage

	\renewcommand{\folder}{./figures/disparate-utilities}

	\subsection{Additional plots for Section~\ref{sec:causal_simulation}}\label{sec:extended_empirical_results:causal_simulation}
	\begin{figure}[h!]
		\vspace{-3mm}
		\centering
		\begin{tikzpicture}
			\node (image) at (0,0) {\includegraphics[width=0.6\linewidth, trim={1cm 1cm 1cm 1.67cm},clip]{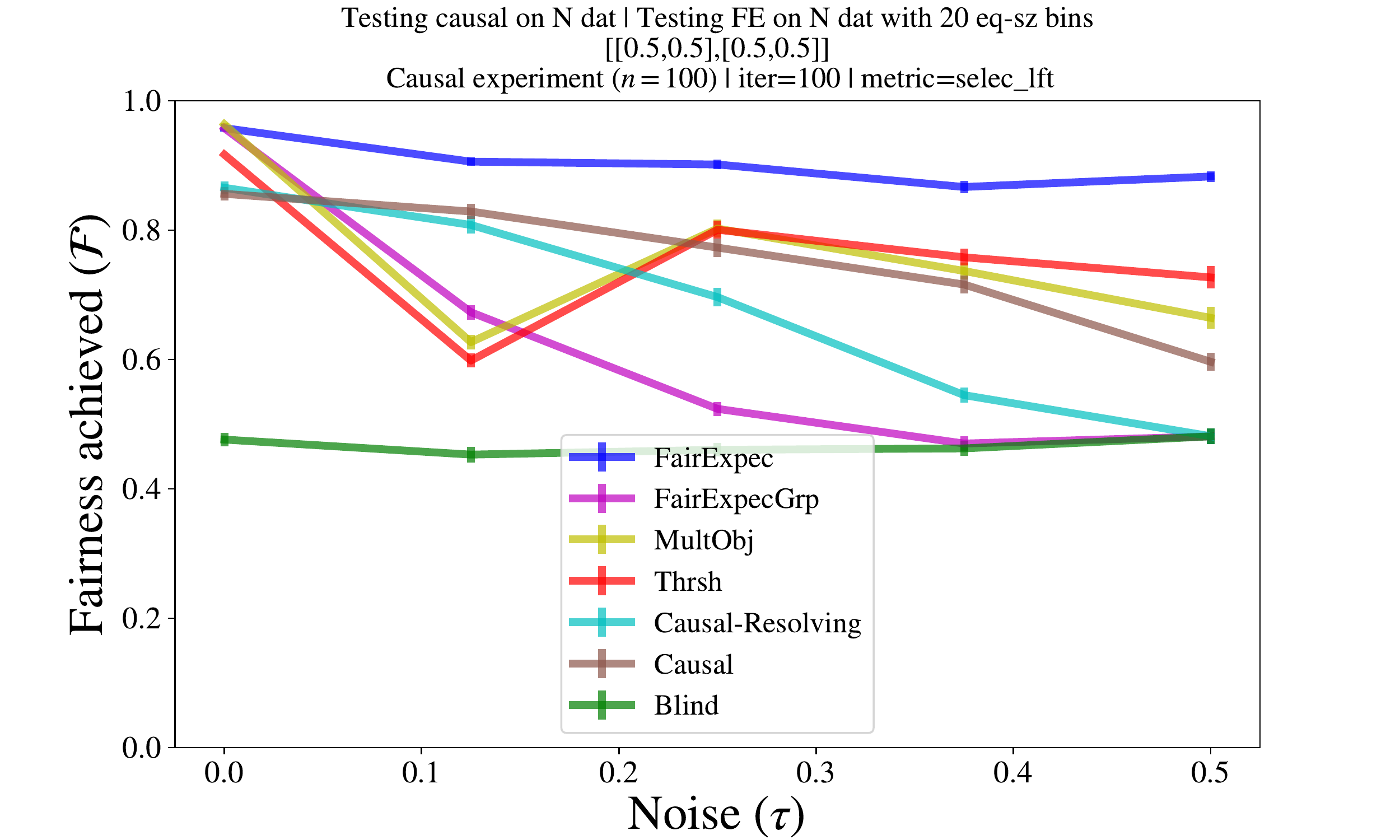}};
			\node[rotate=90, fill=white] at (-4.1*1.1400,-0.5*1.2) {\small\white{AAAAAAAAAAAAAAAAAAAAA}};
			\node[rotate=90] at (-4.1*1.1400,1.75*1.2) {\scriptsize\textit{(more fair)}};
			\node[rotate=90, fill=none] at (-4.1*1.1400,0*1.2) {\small Selection lift ($\cF_L$)};
			\node[rotate=90] at (-4.1*1.1400,-1.75*1.2) {\scriptsize\textit{(less fair)}};
			\node[rotate=0,fill=white] at (2.55*1.14 - 2.48*1.14 + 0.090, -0.22 - 0.05-1.45) {\scriptsize{CntrFair}};
			\node[rotate=0,fill=white] at (2.55*1.14 - 2.33*1.14 + 0.13, -0.18-0.05-1.175) {\scriptsize{CntrFairRes}};
		\end{tikzpicture}
		\vspace{-4mm}
		\caption{\small
		{\em {Additional results for \cref{sec:causal_simulation} with selection lift:}}
		The $y$-axis shows the selection lift $\cF_L$ of different algorithms, and the $x$-axis shows the amount of noise added $\tau\in[0,\nfrac12]$; $\cF_L$ values are averaged over 100 trials, and the error bars represent the standard error of the mean.
		We refer the reader to \cref{sec:causal_simulation} for the details of the simulation.
		We observe similar results as with risk difference ($\cF$).
		For a definition of selection lift see Equation~\eqref{eq:selection_lift_def}.
		\vspace{-3mm}
		}
	\end{figure}
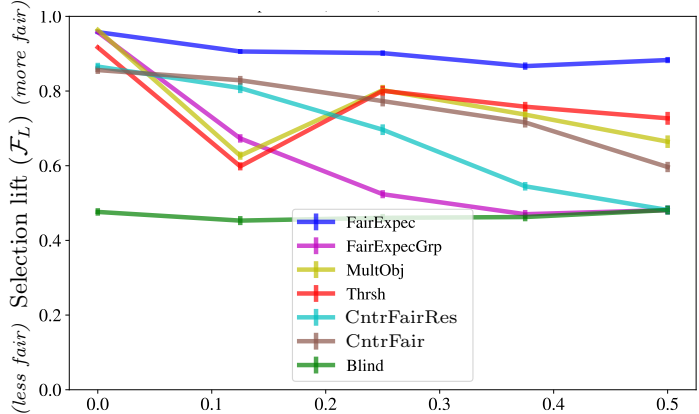

	\begin{figure}[h!]
		\centering
		\begin{tikzpicture}
			\node (image) at (0,0) {\includegraphics[width=0.6\linewidth, trim={1.7cm 1cm 0.5cm 1.6cm},clip]{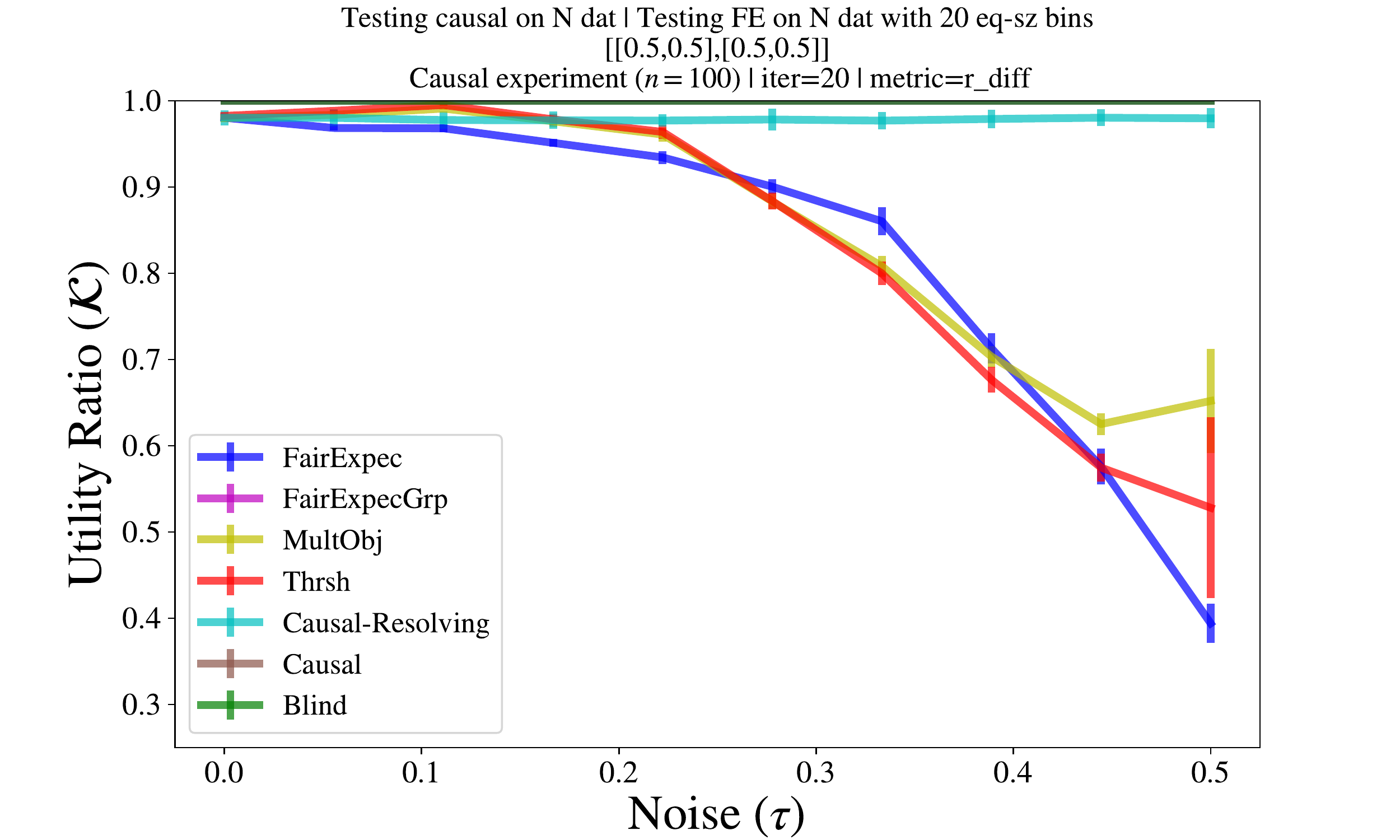}};
			\node[rotate=90, fill=white] at (-4.1*1.23,0) {\small \white{AAAAAAAAAAAAAAAA}};
			\node[rotate=90, fill=white] at (-4.1*1.23,0) {\small Utility Ratio ($\cK$)};
			\node[rotate=0] at (2.3*1.25,-2.5*1.2) {\footnotesize\textit{(more noise)}};
			\node[rotate=0] at (-0.1*1.2,-2.5*1.2) {\small Noise ($\tau$)};
			\node[rotate=0] at (-3*1.2,-2.5*1.2) {\footnotesize\textit{(less noise)}};
			\node[rotate=0,fill=white] at (-2.48,-1.45) {\scriptsize{CntrFair}};
			\node[rotate=0,fill=white] at (-2.1,-1.175) {\scriptsize{\white{CntrFairRe}}};
			\node[rotate=0,fill=white] at (-2.33,-1.175) {\scriptsize{CntrFairRes}};
		\end{tikzpicture}
		\vspace{-3mm}
		\caption{\small
		{\em {Additional results for \cref{sec:causal_simulation}:}}
		This simulation considers the setting where the utilities of a minority group have a lower average than the majority group, and both groups have identical amount of noise.
		{The $y$-axis shows the utility ratio $\cK$ of different algorithms, and the $x$-axis shows the amount of noise added $\tau\in[0,\nfrac12]$; $\cK$ values are averaged over 20 trials, and the error bars represent the standard error of the mean.
		We refer the reader to \cref{sec:causal_simulation} for details and discussion.}
		We note that the utility ratio of all algorithms decreases on adding noise.
		Finally, we note that utilities in this experiment can be negative; we observe that \CntrFair{} has a negative utility ratio.
		The caveat is that \CntrFair{} and \CntrFairRes{}, try to maximize the counterfactual utilities.
		\vspace{0mm}
		\label{fig:causal_simulation:value_plot}
		}
	\end{figure}
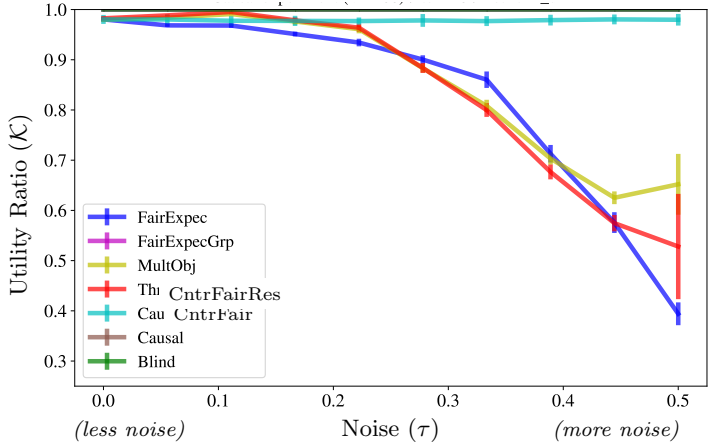

	Given subset $S\in [m]$ and a $\ell\in [p]$, we define the selection rate of $S$ with respect to group $G_\ell$ as
	\begin{align*}
		SR(S,\ell)\coloneqq \frac{|S\cap G_\ell|}{n} \cdot \frac{m}{|G_\ell|}\tagnum{Selection rate}\customlabel{eq:selection_rate_def}{\theequation}
	\end{align*}
	Let $\cA(w,q)\subseteq [m]$ be the subset outputted by algorithm $\cA$ on input $(w,q)$.
	We report $$SR_{\cA,\ell}\coloneqq \Ex\insquare{SR(\cA(w,q),\ell)},$$ where the expectation is over the choices of $(w,q)$.
	We drop the subscript $\cA$ when the algorithm is not important or clear from context.

	\begin{figure}[h!]
		\vspace{-3mm}
		\centering
		\begin{tikzpicture}
			\node (image) at (0,0) {\includegraphics[width=0.6\linewidth, trim={1cm 1cm 1cm 1cm},clip]{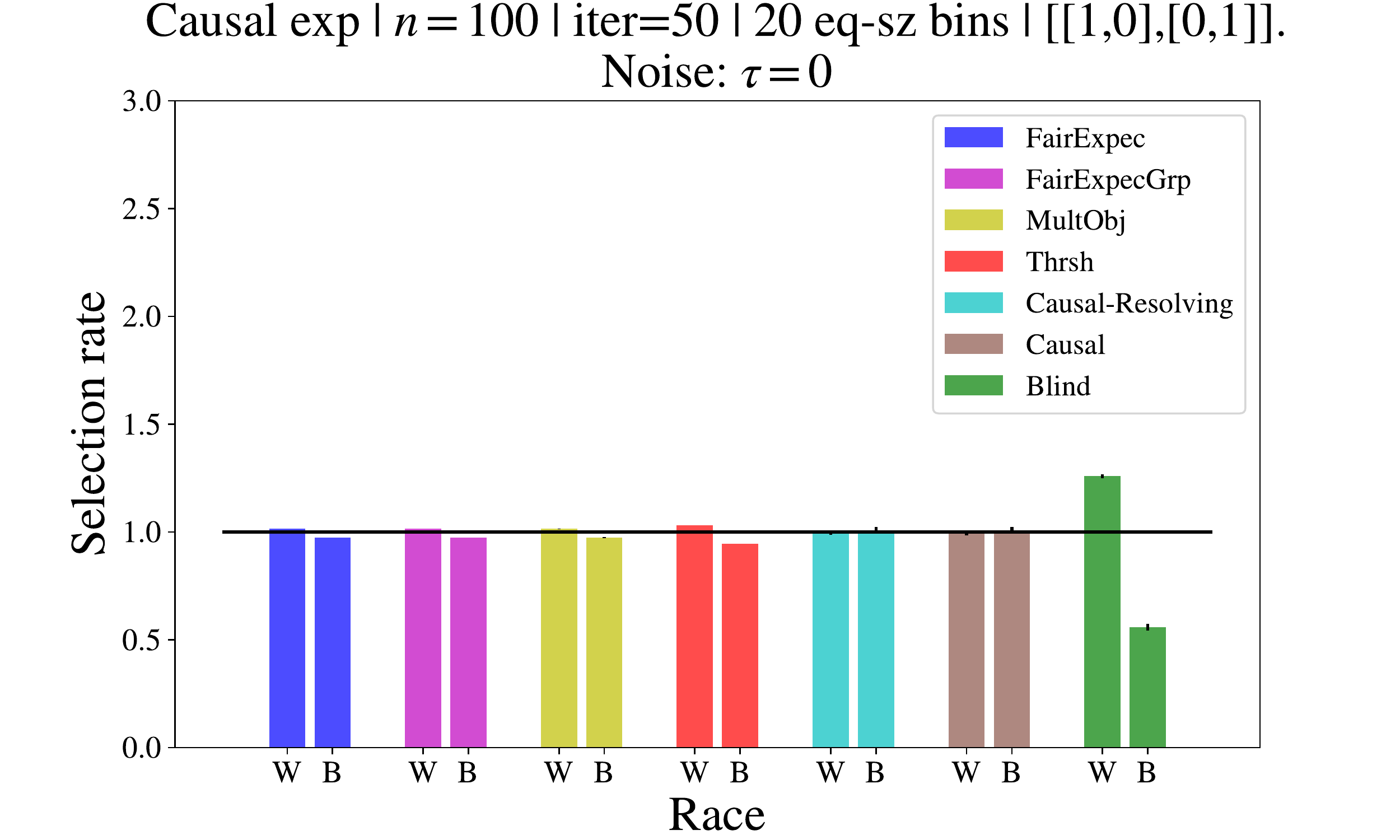}};
			\node[rotate=90, fill=white] at (-4.1*1.1400,-0.5*1.2) {\small\white{AAAAAAAAAAAAAAAAAAAAA}};
			\node[rotate=90, fill=none] at (-4.1*1.2,0*1.2) {\small Selection rate};
			\node[rotate=0] at (0.2,-2.8*1.1) {\small Race {\em (W: White, B: Black)}};
			\node[rotate=0,fill=white] at (0.52 + 2.44 + 2.55 - 2.48, 0.09 + 2.1 - 0.15-1.45) {\scriptsize{CntrFair}};
			\node[rotate=0,fill=white] at (0.58 + 2.44 + 2.55 - 2.33, 0.12 + 2.1 - 0.15-1.175) {\scriptsize{CntrFairRes}};
			\end{tikzpicture}
			\par
			\begin{tikzpicture}
				\node (image) at (0,0) {\includegraphics[width=0.6\linewidth, trim={1cm 1cm 1cm 1cm},clip]{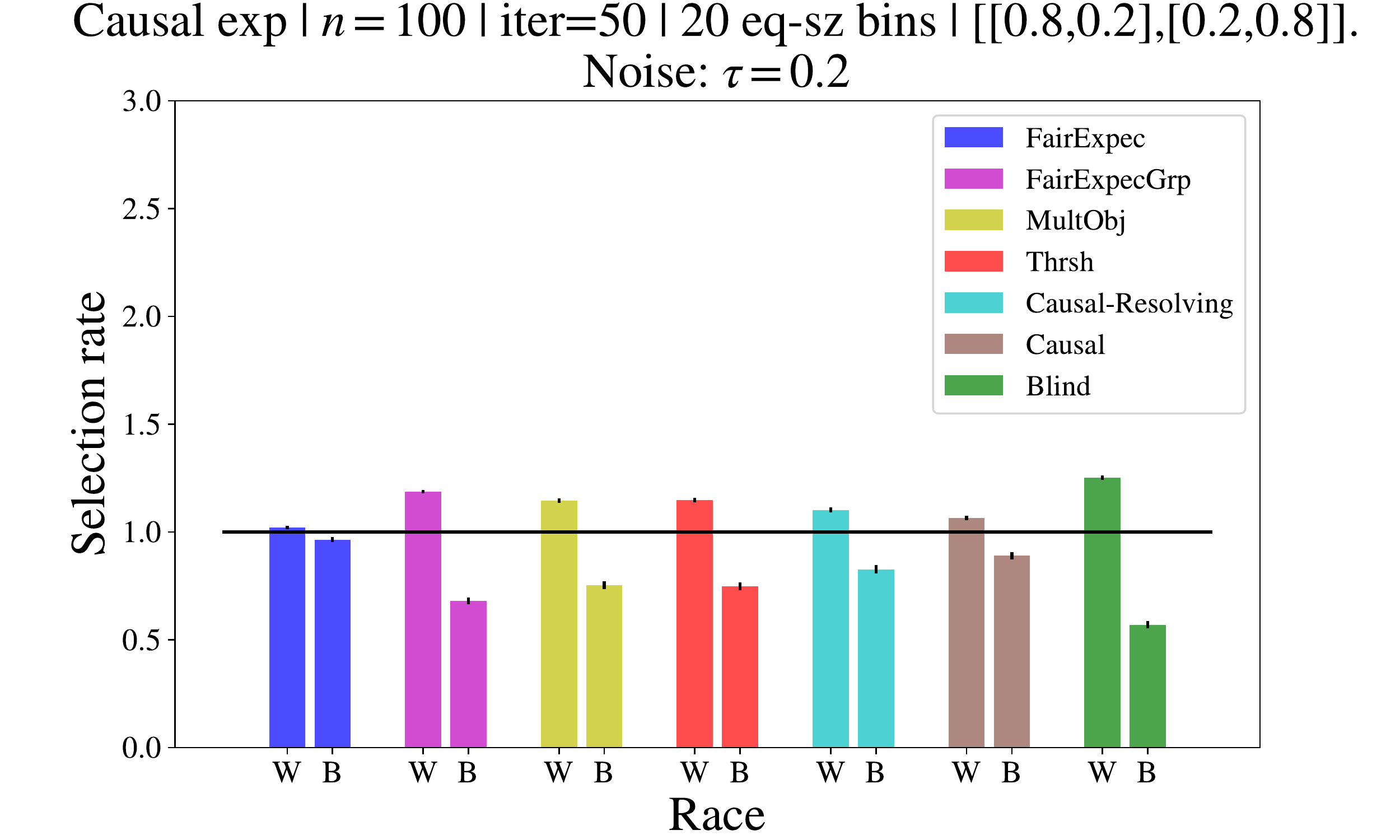}};
				\node[rotate=90, fill=white] at (-4.1*1.1400,-0.5*1.2) {\small\white{AAAAAAAAAAAAAAAAAAAAA}};
				\node[rotate=90, fill=none] at (-4.1*1.2,0*1.2) {\small Selection rate};
				\node[rotate=0] at (0.2,-2.8*1.1) {\small Race {\em (W: White, B: Black)}};
				\node[rotate=0,fill=white] at (0.52 + 2.44 + 2.55 - 2.48, 0.09 + 2.1 - 0.15-1.45) {\scriptsize{CntrFair}};
				\node[rotate=0,fill=white] at (0.58 + 2.44 + 2.55 - 2.33, 0.12 + 2.1 - 0.15-1.175) {\scriptsize{CntrFairRes}};
			\end{tikzpicture}
			\begin{tikzpicture}
				\node (image) at (0,0) {\includegraphics[width=0.6\linewidth, trim={1cm 1cm 1cm 1cm},clip]{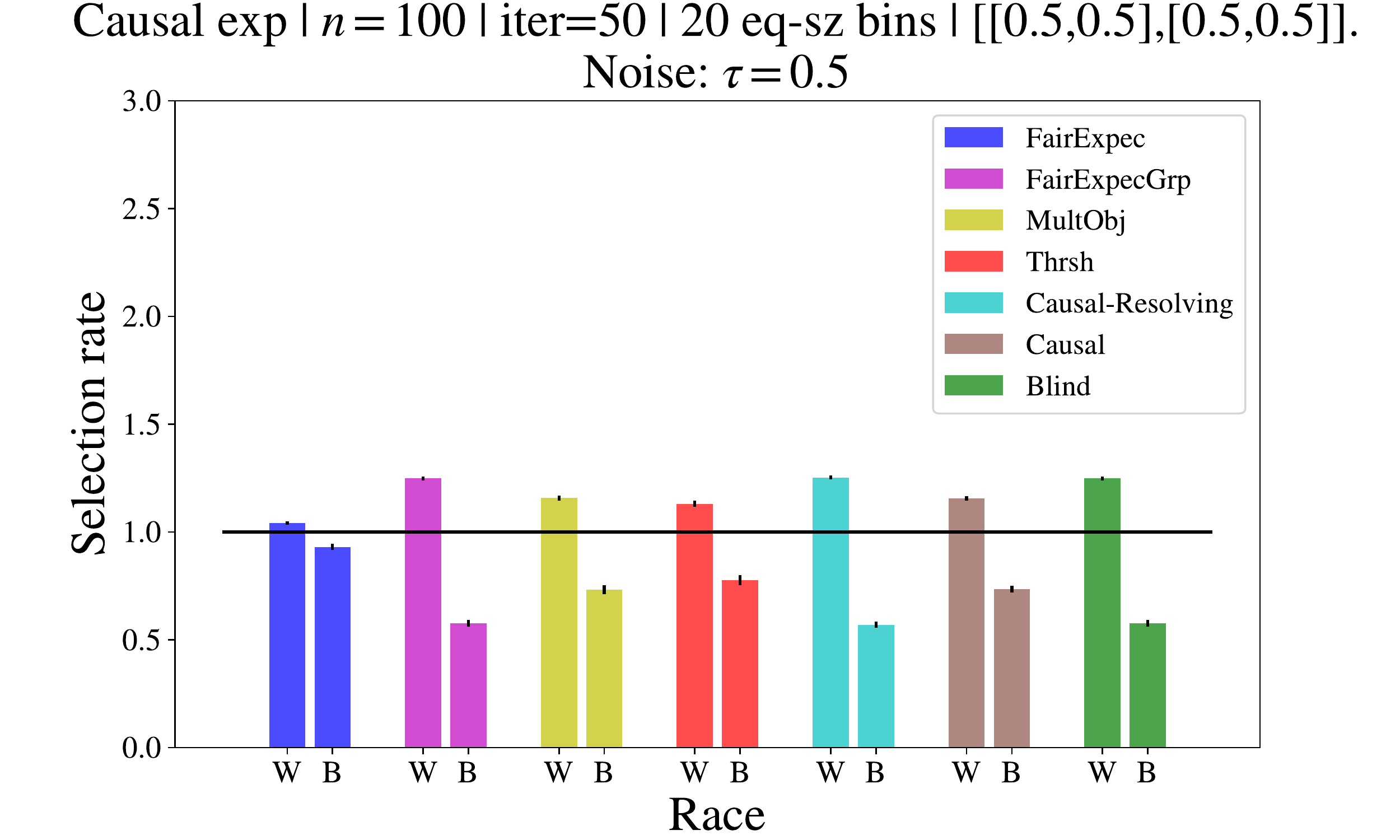}};
				\node[rotate=90, fill=white] at (-4.1*1.1400,-0.5*1.2) {\small\white{AAAAAAAAAAAAAAAAAAAAA}};
				\node[rotate=90, fill=none] at (-4.1*1.2,0*1.2) {\small Selection rate};
				\node[rotate=0] at (0.2,-2.8*1.1) {\small Race {\em (W: White, B: Black)}};
				\node[rotate=0,fill=white] at (0.52 + 2.44 + 2.55 - 2.48, 0.09 + 2.1 - 0.15-1.45) {\scriptsize{CntrFair}};
				\node[rotate=0,fill=white] at (0.58 + 2.44 + 2.55 - 2.33, 0.12 + 2.1 - 0.15-1.175) {\scriptsize{CntrFairRes}};
			\end{tikzpicture}
			\caption{\small
			{\em {Additional plots for \cref{sec:causal_simulation}:}}
			This simulation considers the setting where the utilities of a minority group have a lower average than the majority group, and both groups have identical amount of noise.
			The three plots show three values of noise $\tau\in \inbrace{0,0.2,0.5}$.
			{The $y$-axis shows the selection rate (see Equation~\eqref{eq:selection_rate_def}) of different algorithms; Selection rate values are averaged over 50 trials, and the error bars represent the standard error of the mean.}
			We refer the reader to \cref{sec:causal_simulation} for details and discussion.
			We observe that \FairExpec{} is the closest having a selection rate of 1 across noise levels.
			We extended the code by \cite{yang2020causal} to generate this plot.
			\vspace{-3mm}
			}
		\end{figure}
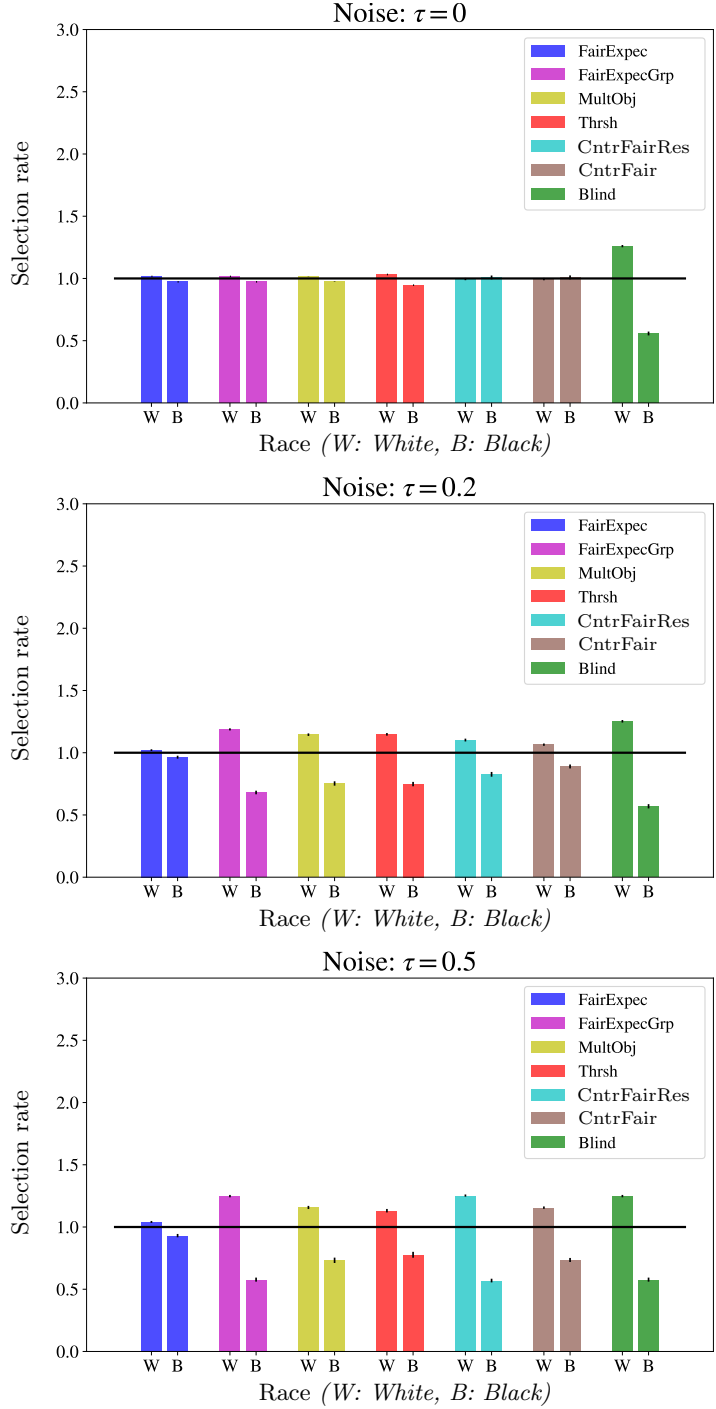

		\newpage\white{.}\newpage

		\subsection{Additional results for Section~\ref{sec:candidate_selection}}\label{sec:extended_empirical_results:candidate_selection}
		\renewcommand{\folder}{./figures/candidate-subset-selection}
		\begin{figure}[h!]
			\vspace{-3mm}
			\begin{center}
				\subfigure[The $y$-axis shows the utility ratio $\cK$ of different algorithms, and the $x$-axis shows the constraint parameters ($\alpha$ for \FairExpec{}, \FairExpecGrp{}, and \Thresh{}, and $\lambda$ for \MultObj); $\cK$ values are averaged over 100 trials, and the error bars represent the standard error of the mean.]
				{\begin{tikzpicture}
				\node (image) at (0,0) {\includegraphics[width=0.6\linewidth, trim={1.7cm 0.3cm 0.5cm 0.26cm},clip]{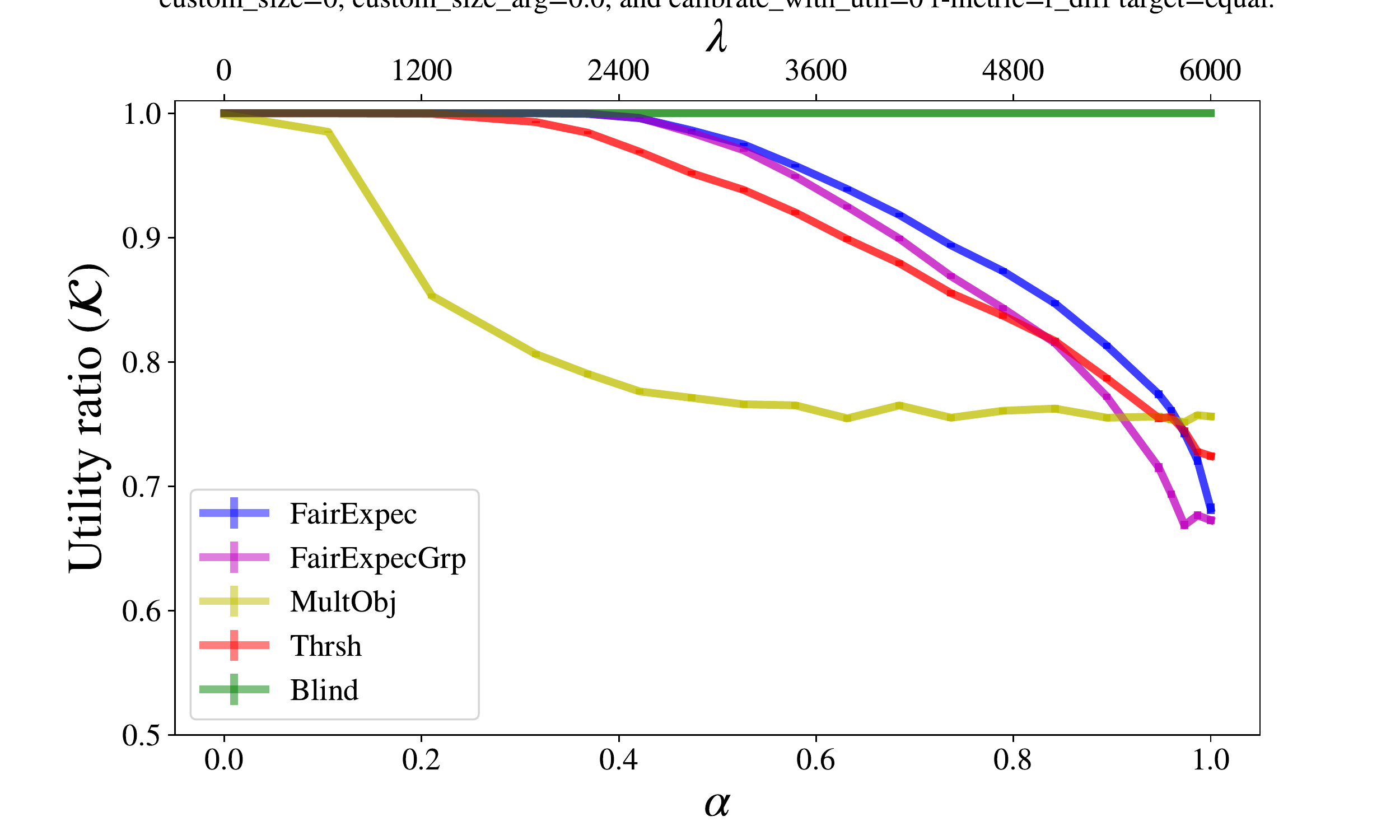}};
				\node[rotate=90, fill=white] at (-4.1*1.23,0) {\small \white{AAAAAAAAAAAAAAAA}};
				\node[rotate=90, fill=white] at (-4.1*1.23,0) {\small Utility Ratio ($\cK$)};
				\end{tikzpicture}}\par
				\subfigure[The $y$-axis shows the risk difference $\cF$ of different algorithms, and the $x$-axis shows the constraint parameters ($\alpha$ for \FairExpec{}, \FairExpecGrp{}, and \Thresh{}, and $\lambda$ for \MultObj); $\cF$ values are averaged over 100 trials, and the error bars represent the standard error of the mean.]
				{ \begin{tikzpicture}
				\node (image) at (0,0) {\includegraphics[width=0.6\linewidth, trim={1.7cm 0.3cm 0.5cm 0.26cm},clip]{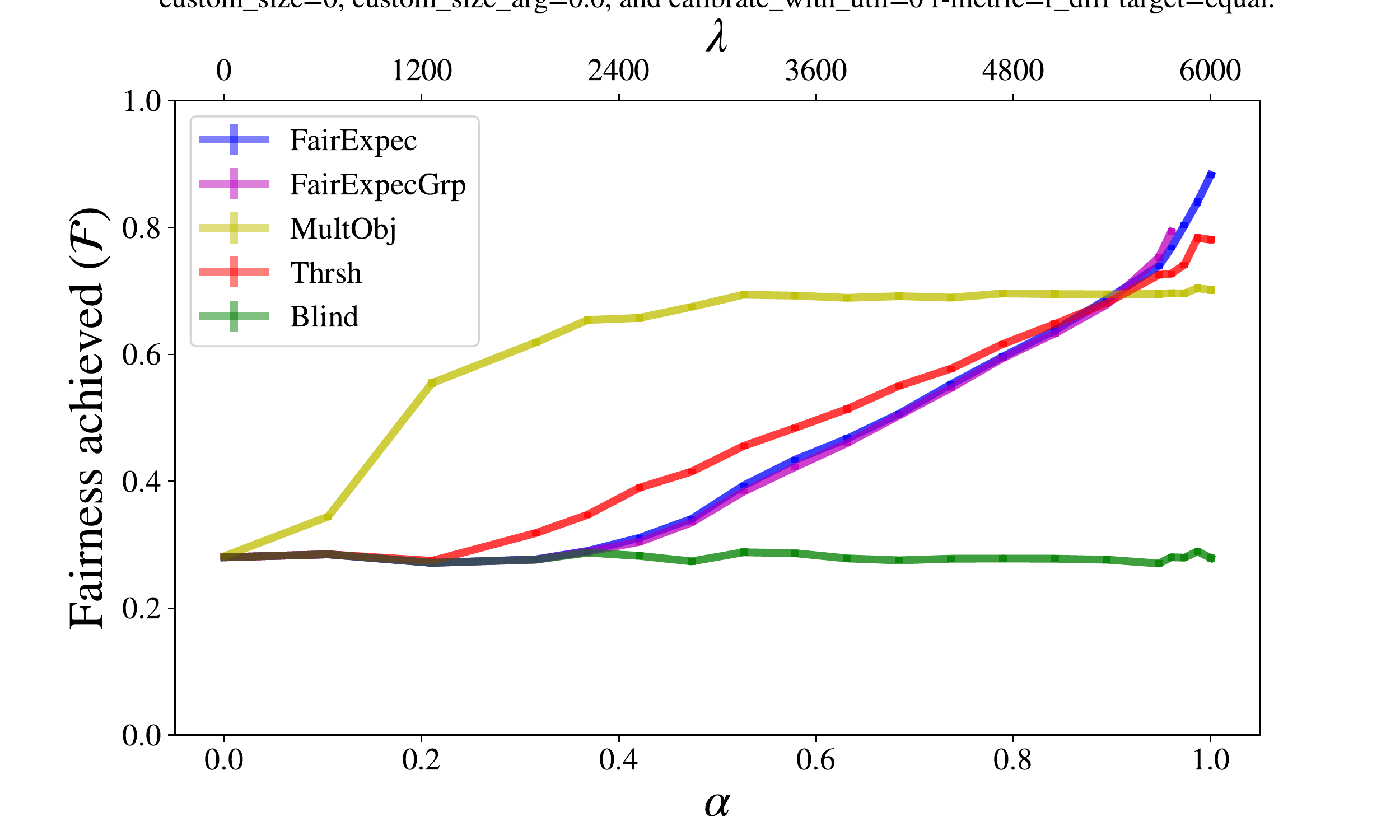}};
				\node[rotate=90, fill=white] at (-4.4*1.1400,-0.5*1.2) {\small\white{AAAAAAAAAAAAAAAAAAAAA}};
				\node[rotate=90] at (-4.4*1.1400,1.75*1.2) {\scriptsize\textit{(more fair)}};
				\node[rotate=90, fill=none] at (-4.4*1.1400,0*1.2) {\small Risk difference ($\cF$)};
				\node[rotate=90] at (-4.4*1.1400,-1.75*1.2) {\scriptsize\textit{(less fair)}};
			\end{tikzpicture}}
		\end{center}
		\caption{\small
		{\em {Additional plot for candidate selection (\cref{sec:candidate_selection}):}}
		This simulation considers race as the protected attribute which takes $p=4$ values, where an individual's utility is drawn from different distributions depending on their race.
		The simulation uses last name as a proxy to derive noisy information of race.
		We refer the reader to \cref{sec:candidate_selection} for the details of the simulation.
		}
		\label{fig:candidate_selection:1_additional}
	\end{figure}
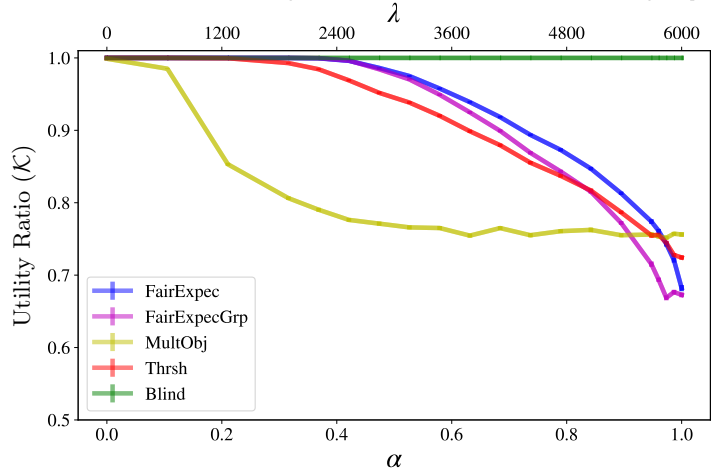
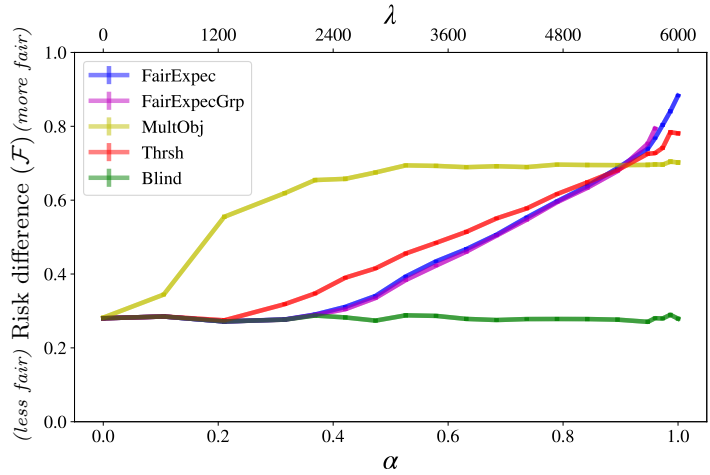

	\vspace{-5mm}

	\begin{remark}[{\bf Reading the double $x$-axis}]
		We note that both subfigures in Figure~\ref{fig:candidate_selection:1_additional} have a double $x$-axis: one for $\alpha$ and one for $\lambda$.
		Therefore, one should note compared the $\cF$ or $\cK$ of \MultObj{} with other algorithms at a particular $x$-coordinate.
		It could be useful to consider the limiting their values at the largest $\alpha$ and $\lambda$, respectively.
	\end{remark}

	\begin{figure}[h!]
		\centering
		\begin{tabular}{|l|l|}\hline
			Race & Mean (USD per annum) \\\hline
			White & 100,169 \\
			Black & 70,504\\
			Asian and Pacific Islander (API) & 118,421\\
			Hispanic & 71,565 \\\hline
		\end{tabular}
		\caption{Statistics of $\cD_r$ for different races $r$; $\cD_r$ is the discrete distributions of incomes of families with race $r$ derived from the income dataset~\cite{census_income_dataset}.}
		\label{fig:stats_of_dr}
		\vspace{0mm}
	\end{figure}

	\newpage
	\renewcommand{\folder}{./figures/candidate-subset-selection}
	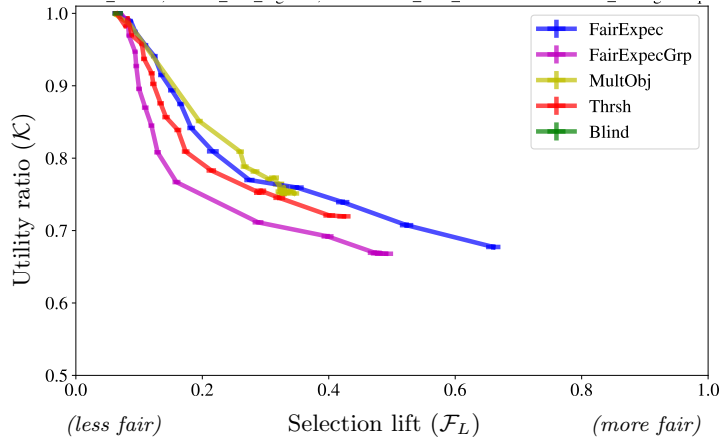
\begin{figure}[h!]
		\vspace{-5mm}
		\begin{center}
			\begin{tikzpicture}
				\node (image) at (-0.5,0) {\includegraphics[width=0.6\linewidth, trim={1.7cm 1.2cm 0.5cm 1.7cm},clip]{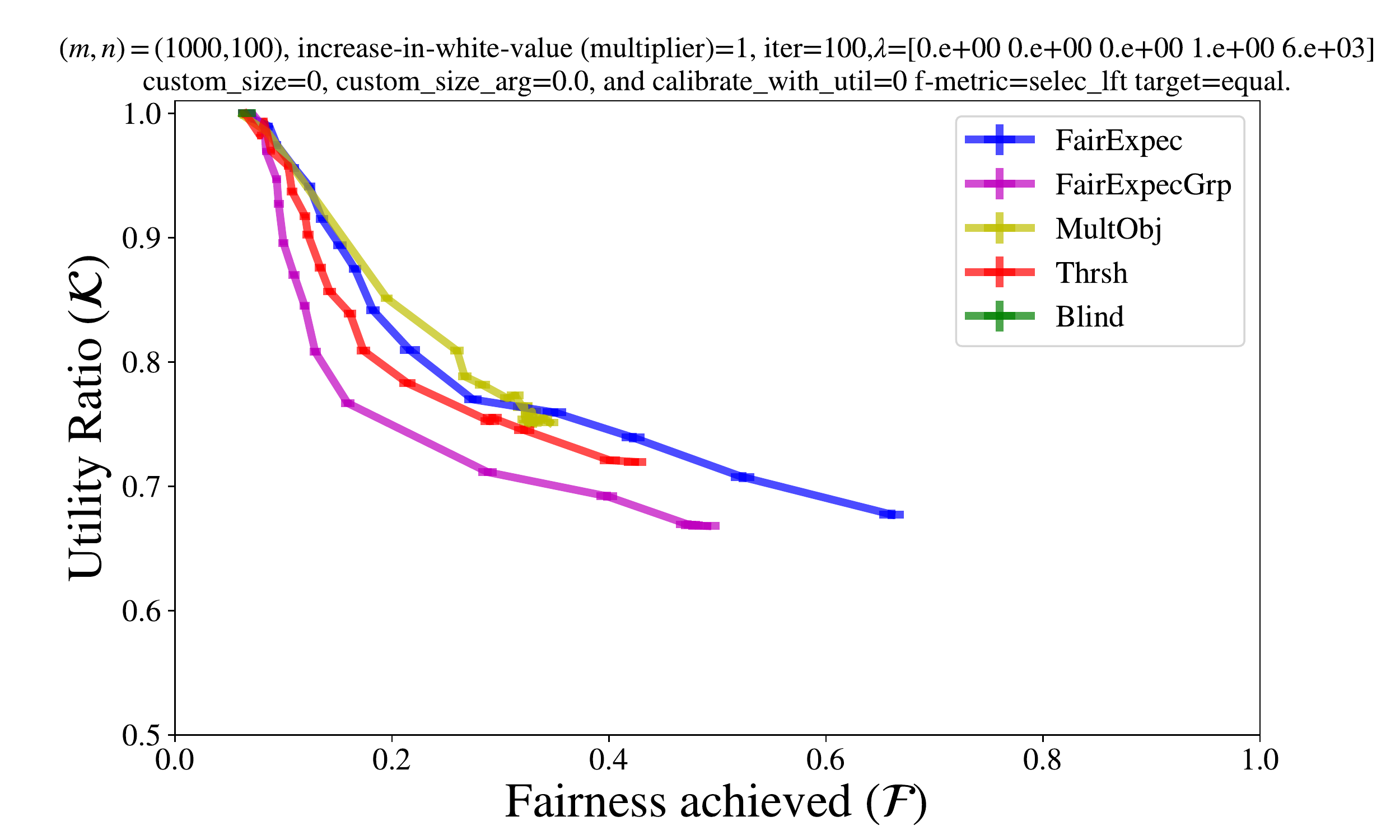}};
				\node[rotate=0, fill=white] at (-0.3*1.2,2.3*1.25) {\small \white{.........................}};
				\node[rotate=90, fill=white] at (-4.6*1.2,-0.5*1.25) {\small \white{AAAAAAAAAAAAAAAAAAAAA}};
				\node[rotate=90, fill=white] at (-4.6*1.2,0*1.25) {\small Utility ratio ($\cK$)};
				\node[rotate=0] at (2.3*1.2,-2.4*1.25) {\footnotesize\textit{(more fair)}};
				\node[rotate=0] at (-0.6*1.2,-2.4*1.25) {\small Selection lift ($\cF_L$)};
				\node[rotate=0] at (-3.6*1.2,-2.4*1.25) {\footnotesize\textit{(less fair)}};
			\end{tikzpicture}
			\vspace{-7mm}
		\end{center}
		\caption{\small
		{\em {Additional results for \cref{sec:candidate_selection} with selection lift:}}
		The $y$-axis shows the selection lift $\cF_L$ of different algorithms, and the $x$-axis shows the risk difference of different algorithms; both $\cK$ and $\cF_L$ values are averaged over 200 trials, and the error bars represent the standard error of the mean.
		We refer the reader to \cref{sec:candidate_selection} for the details of the simulation.
		We observe that \FairExpec{} has lower utility ratio than \MultObj{} for some values of selection lift; see \cref{rem:other_metrics} for a discussion.
		For a definition of selection lift see Equation~\eqref{eq:selection_lift_def}.
		}
		\vspace{-5mm}
		\label{fig:candidate_selection:additional2}
	\end{figure}

	\subsection{Additional results for Section~\ref{sec:image_selection}}\label{sec:extended_empirical_results:image_selection}
	\renewcommand{\folder}{./figures/image-subset-selection}
	\begin{figure}[h!]
		\vspace{-3mm}
		\begin{center}
		\begin{tabular}{|p{1.3cm}|p{6.5cm}|}\hline
		& Occupation name \\\hline
		$\stf(0.8)$  Num: $12$& administrative assistant; counselor; dental hygienist; flight attendant; hairdresser; housekeeper; massage therapist; nurse; nurse practitioner; receptionist; social worker; special ed teacher \\\hline
		$\stm(0.8)$ Num: $29$ & announcer; barber; bill collector; building inspector; building painter; butcher; chief executive officer; clergy member; construction worker; courier; crane operator; detective; dishwasher; electrician; exterminator; garbage collector; groundskeeper; logistician; machinist; parking attendant; plumber; police officer; private investigator; roofer; security guard; taxi driver; teller; web developer; welder\\\hline
	\end{tabular}
	\vspace{-4mm}
	\end{center}
	\caption{ \small {\bf\em Occupations with at least 80\% of images from one gender (from the occupations dataset~\cite{celis2019imagesummarization})}.
	We refer the reader to \cref{sec:image_selection} for more details, and the definition of $\stm(\cdot$ and $\stf\cdot)$.
	}
	\vspace{-3mm}
	\label{table:occupations}
	\end{figure}

	\begin{remark}[{\bf Reading the double $x$-axis}]
		We note that both subfigures in Figure~\ref{fig:image_selection:1_additional} have a double $x$-axis: one for $\alpha$ and one for $\lambda$.
		Therefore, one should note compared the $\cF$ or $\cK$ of \MultObj{} with other algorithms at a particular $x$-coordinate.
		It could be useful to consider the limiting their values at the largest $\alpha$ and $\lambda$, respectively.
	\end{remark}

	\begin{figure}[h!]
		\vspace{-3mm}
		\begin{center}
			\subfigure[The $y$-axis shows the utility ratio $\cK$ of different algorithms, and the $x$-axis shows the constraint parameters ($\alpha$ for \FairExpec{}, \FairExpecGrp{}, and \Thresh{}, and $\lambda$ for \MultObj); $\cK$ values are averaged over 200 trials, and the error bars represent the standard error of the mean.]
			{\begin{tikzpicture}
			\node (image) at (0,0) {\includegraphics[width=0.6\linewidth, trim={1.7cm 0.3cm 0.5cm 0.26cm},clip]{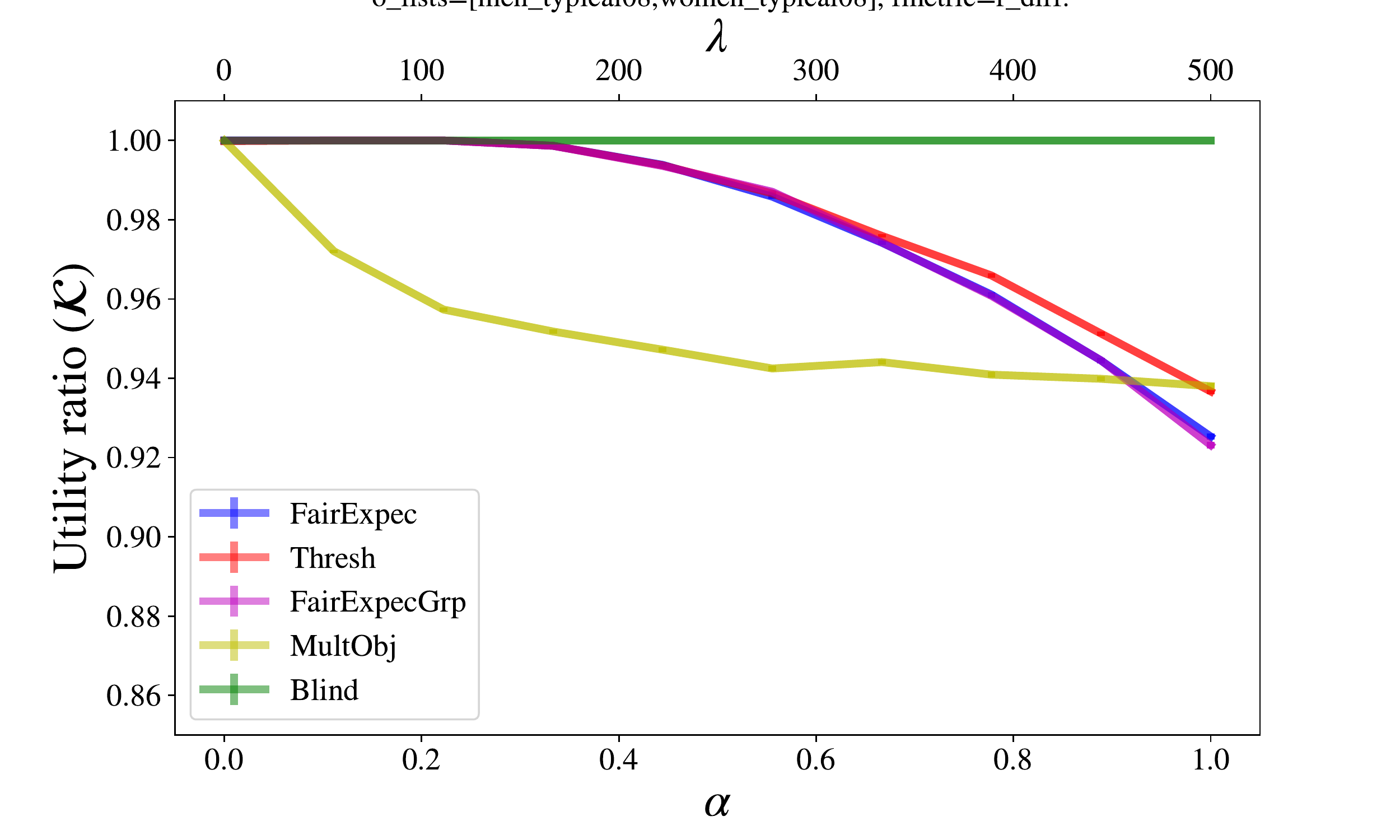}};
			\node[rotate=90, fill=white] at (-4.1*1.25,0) {\small \white{AAAAAAAAAAAAAAAA}};
			\node[rotate=90, fill=white] at (-4.1*1.25,0) {\small Utility Ratio ($\cK$)};
			\end{tikzpicture}}\par
			\subfigure[The $y$-axis shows the risk difference $\cF$ of different algorithms, and the $x$-axis shows the constraint parameters ($\alpha$ for \FairExpec{}, \FairExpecGrp{}, and \Thresh{}, and $\lambda$ for \MultObj); $\cF$ values are averaged over 200 trials, and the error bars represent the standard error of the mean.]
		{ \begin{tikzpicture}
		\node (image) at (0,0) {\includegraphics[width=0.6\linewidth, trim={1.7cm 0.3cm 0.5cm 0.26cm},clip]{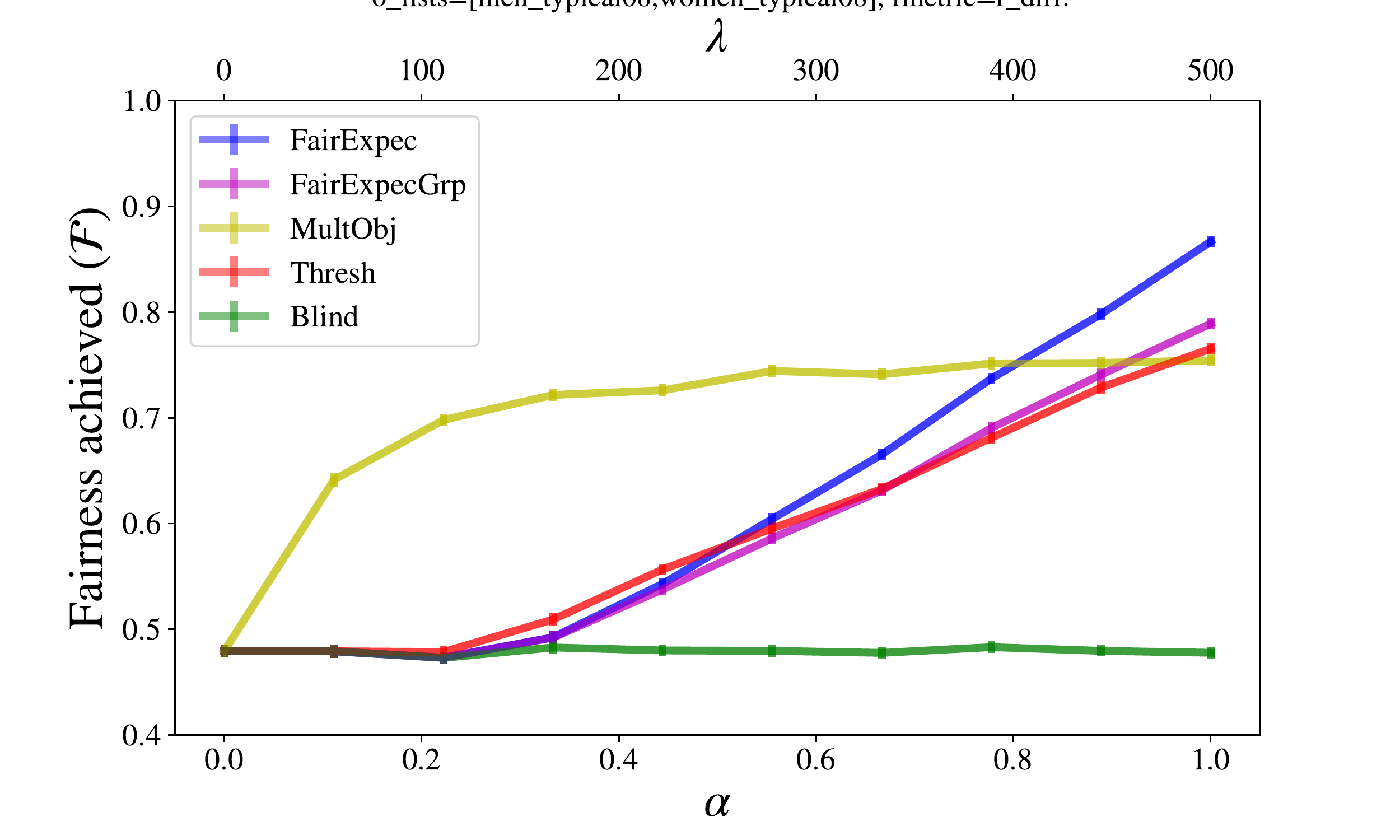}};
		\node[rotate=90, fill=white] at (-4.1*1.23, 0) {\small\white{AAAAAAAAAAAAAAAAAAA}};
		\node[rotate=90] at (-4.1*1.23, 1.3*1.6) {\scriptsize\textit{(more fair)}};
		\node[rotate=90, fill=none] at (-4.1*1.23, 0) {\small Risk difference ($\cF$)};
		\node[rotate=90] at (-4.1*1.23, -1.3*1.6) {\scriptsize\textit{(less fair)}};
	\end{tikzpicture}}
	\end{center}
	\caption{\small
	{\em {Additional for image selection (\cref{sec:image_selection}):}}
	This simulation considers gender as the protected attribute and uses a CNN-based classifier to derive noisy information about the gender of the person in the image.
	We refer the reader to \cref{sec:image_selection} for the details of the simulation.
	}
	\label{fig:image_selection:1_additional}
	\end{figure}
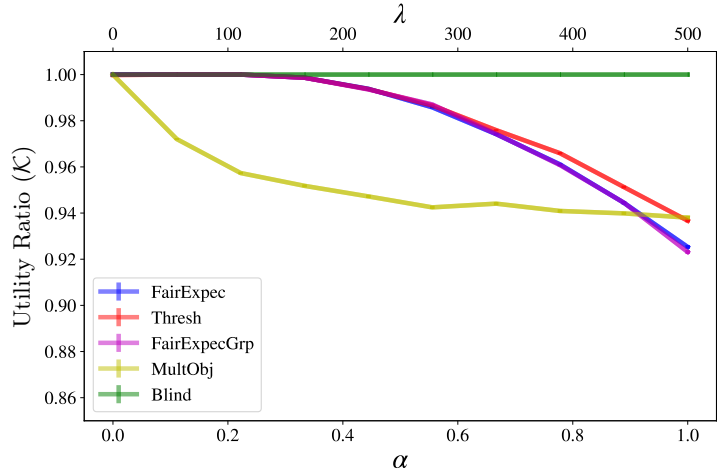
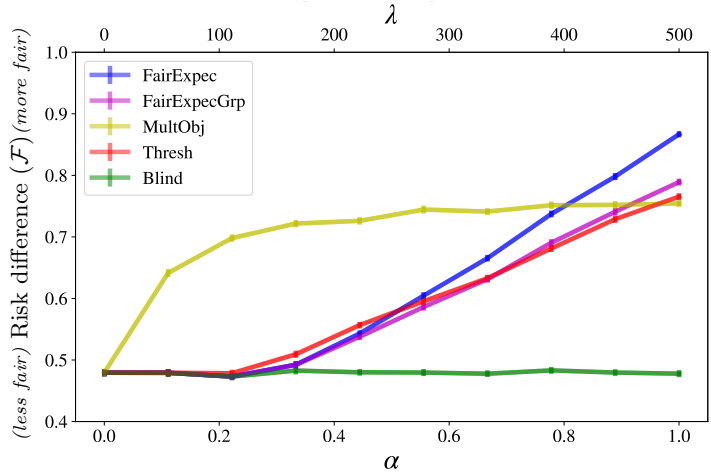

	\renewcommand{\folder}{./figures/image-subset-selection}
	\begin{figure}[h!]
		\vspace{-3mm}
		\begin{center}
			\subfigure[The $y$-axis shows the NDCG value of different algorithms, and the $x$-axis shows selection lift ($\cF_L$); both values are averaged over 100 trials.]
		{\begin{tikzpicture}
		\node (image) at (-0.5,0) {\includegraphics[width=0.6\linewidth, trim={1.7cm 1.2cm 0.5cm 1.7cm},clip]{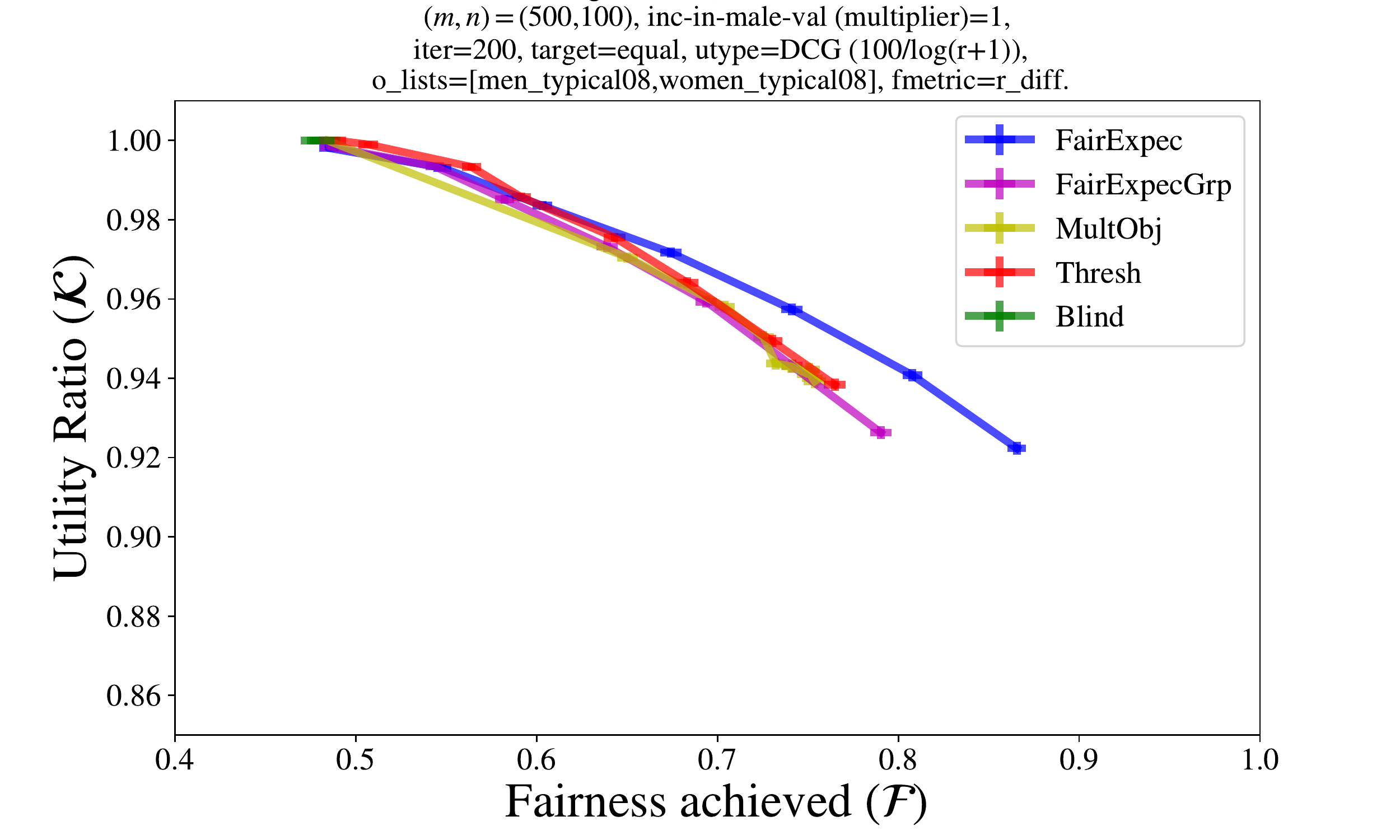}};
		\node[rotate=0, fill=white] at (-0.3*1.23,2.3*1.25) {\small \white{.........................}};
		\node[rotate=90, fill=white] at (-4.6*1.23,-0.5*1.25) {\small \white{AAAAAAAAAAAAAAAAAAAAA}};
		\node[rotate=90, fill=white] at (-4.6*1.23,0*1.25) {\small NDGC};
		\node[rotate=0] at (2.3*1.2,-2.4*1.25) {\footnotesize\textit{(more fair)}};
		\node[rotate=0] at (-0.6*1.2,-2.4*1.25) {\small Selection lift ($\cF_L$)};
		\node[rotate=0] at (-3.6*1.2,-2.4*1.25) {\footnotesize\textit{(less fair)}};
	\end{tikzpicture}}
		\vspace{-4mm}
	\end{center}
	\vspace{0mm}
	\caption{\small
	{\em {Additional figure for \cref{sec:image_selection}:}}
	This simulation considers gender as the protected attribute and uses a CNN-based classifier to derive noisy information about the gender of the person in the image.
	The $y$-axis shows the NDCG value of different algorithms, and the $x$-axis shows selection lift ($\cF_L$); both values are averaged over 100 trials.
	We refer the reader to \cref{sec:image_selection} for details of the simulation.
	We observe similar results to those with utility ratio ($\cK$) (see Figure~\ref{fig:image_selection}).
	\vspace{0mm}
	}
	\label{fig:image_selection:ndgc}
\end{figure}
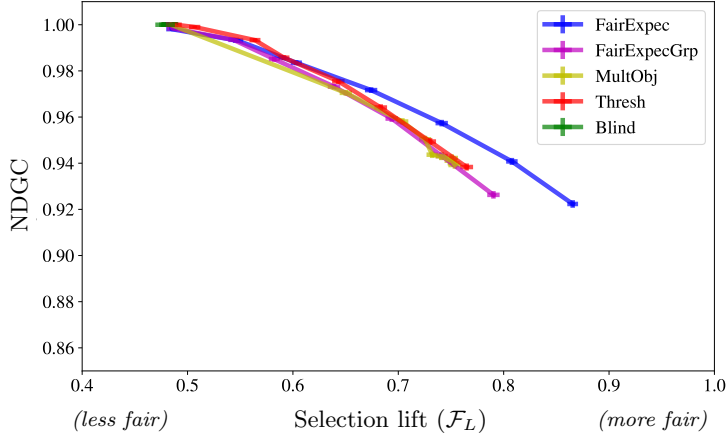
	\bigskip

	\newpage\white{.}\newpage
	\renewcommand{\folder}{./figures/image-subset-selection}
	\begin{figure}[h!]
		\vspace{0mm}
		\begin{center}
			\begin{tikzpicture}
				\node (image) at (-0.5,0) {\includegraphics[width=0.6\linewidth, trim={1.7cm 1.2cm 0.5cm 1.7cm},clip]{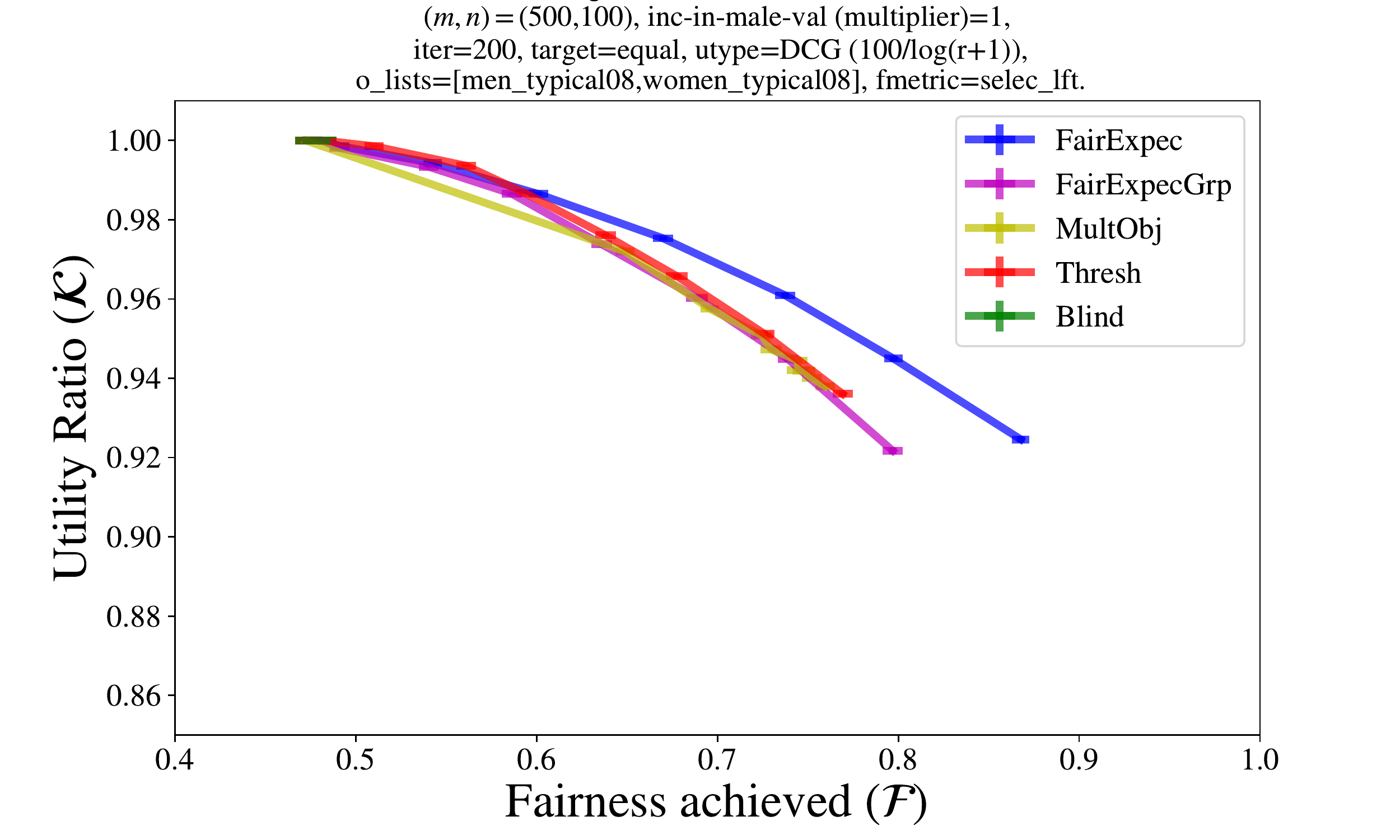}};
				\node[rotate=0, fill=white] at (-0.3*1.23,2.3*1.25) {\small \white{.........................}};
				\node[rotate=90, fill=white] at (-4.6*1.23,-0.5*1.25) {\small \white{AAAAAAAAAAAAAAAAAAAAA}};
				\node[rotate=90, fill=white] at (-4.6*1.23,0*1.25) {\small Utility ratio ($\cK$)};
				\node[rotate=0] at (2.3*1.2,-2.4*1.25) {\footnotesize\textit{(more fair)}};
				\node[rotate=0] at (-0.6*1.2,-2.4*1.25) {\small Selection lift ($\cF_L$)};
				\node[rotate=0] at (-3.6*1.2,-2.4*1.25) {\footnotesize\textit{(less fair)}};
			\end{tikzpicture}
		\end{center}
		\vspace{-4mm}
		\caption{\small
		{\em {Additional results for \cref{sec:image_selection} with selection lift:}}
		The $y$-axis shows the utility ratio ($\cK$) of different algorithms, and the $x$-axis shows the selection lift ($\cF_L$)  of different algorithms; both $\cK$ and $\cF_L$ values are averaged over 200 trials, and the error bars represent the standard error of the mean.
		We refer the reader to \cref{sec:image_selection} for the details of the simulation.
		We observe similar results as with risk difference ($\cF$).
		See Equation~\eqref{eq:selection_lift_def} for a definition of selection lift ($\cF_L$).
		}
		\vspace{-3mm}
	\end{figure}

\end{document}